\def\llncs{0}
\def\fullpage{1}
\def\anonymous{0}

\def\notxfont{0}
\def\submission{0}
\def\cameraready{0}

\def\commentanonymous{0}
\ifnum\submission=1
\def\anonymous{1}
\def\llncs{1}

\else
\fi

\ifnum\cameraready=1
\def\submission{1}
\def\llncs{1}

\def\anonymous{0}
\else
\fi

\ifnum\anonymous=1
\def\llncs{0}

\else
\fi

\ifnum\llncs=1
	\documentclass[envcountsect,a4paper,runningheads,11pt]{llncs}
\else
	\documentclass[letterpaper,hmargin=1.05in,vmargin=1.05in,11pt]{article}
			\ifnum\fullpage=1
		\usepackage{fullpage}
		\fi
\fi

\usepackage{amsmath, amsfonts, amssymb, mathtools,amscd}

\usepackage{amsthm}

\usepackage{fancyhdr}
\usepackage{lmodern}
\usepackage[T1]{fontenc}
\usepackage[utf8]{inputenc}
\usepackage{soul}
\usepackage{arydshln} % In order to use \hdashline
\usepackage{url}
\usepackage{ifthen}
\usepackage{bm}
\usepackage{multirow}
\usepackage[dvips]{graphicx}
\usepackage[usenames]{color}
\usepackage{xcolor,colortbl} % Leave out in case the usepackage cannot be found. Not critical
\usepackage{threeparttable}
\usepackage{comment}
\usepackage{paralist,verbatim}
\usepackage{cases}
\usepackage{booktabs}
\usepackage{braket}
\usepackage{cancel} 
\usepackage{ascmac} 
\usepackage{framed}
\usepackage{physics}
\usepackage{authblk}
\usepackage{pifont}
\usepackage[margin=1in]{geometry}

\definecolor{darkblue}{rgb}{0,0,0.6}
\definecolor{darkgreen}{rgb}{0,0.5,0}
\definecolor{maroon}{rgb}{0.5,0.1,0.1}
\definecolor{dpurple}{rgb}{0.2,0,0.65}

\usepackage[absolute]{textpos}
\usepackage[final]{microtype}
\usepackage[absolute]{textpos}
\usepackage{everypage}

\usepackage{amsthm}
\usepackage{thmtools}

\usepackage[%
  colorlinks=true,
  citecolor=blue,
  pagebackref=true
]{hyperref}

\usepackage[capitalise,noabbrev]{cleveref}

\newtheoremstyle{thicktheorem}%
{\topsep}
{\topsep}
{\itshape}{}%
{\bfseries}%
{.}
{ }%
{\thmname{#1}\thmnumber{ #2}%
		\thmnote{ (#3)}%
}

\newtheoremstyle{remark}%name
{\topsep}
{\topsep}
	{}%body font
	{}%indent amount
	{}%theorem head font
	{.}%punctuation after theorem head
	{ }%space after theorem head
	{\textit{\thmname{#1}}\thmnumber{ #2}%theorem head specs
			\thmnote{ (#3)}%
	}

\ifnum\llncs=0
	\theoremstyle{thicktheorem}
	\newtheorem{theorem}{Theorem}[section]
	\newtheorem{lemma}[theorem]{Lemma}
    \newtheorem{assumption}[theorem]{Assumption}
	\newtheorem{corollary}[theorem]{Corollary}
	\newtheorem{proposition}[theorem]{Proposition}
	\newtheorem{definition}[theorem]{Definition}
	\newtheorem{game}[theorem]{Game}

	\theoremstyle{remark}
	\newtheorem{claim}[theorem]{Claim}
	\newtheorem{remark}[theorem]{Remark}

    \newtheorem{observation}[theorem]{Observation}

	\crefname{theorem}{Theorem}{Theorems}
	\crefname{assumption}{Assumption}{Assumptions}
	\crefname{construction}{Construction}{Constructions}
	\crefname{corollary}{Corollary}{Corollaries}
	\crefname{conjecture}{Conjecture}{Conjectures}
	\crefname{definition}{Definition}{Definitions}
	\crefname{exmaple}{Example}{Examples}
	\crefname{experiment}{Experiment}{Experiments}
	\crefname{counterexample}{Counterexample}{Counterexamples}
	\crefname{lemma}{Lemma}{Lemmata}
	\crefname{observation}{Observation}{Observations}
	\crefname{remark}{Remark}{Remarks}
	\crefname{claim}{Claim}{Claims}
	\crefname{fact}{Fact}{Facts}
	\crefname{note}{Note}{Notes}
\else
\fi
\ifnum\llncs=1
 \crefname{appendix}{App.}{Appendices}
 \crefname{section}{Sec.}{Sections}
\else
\fi

\ifnum\llncs=1
\renewcommand*{\backref}[1]{}
\else
	\renewcommand*{\backref}[1]{(Cited on page~#1.)}
	\ifnum\notxfont=1
	\else
		\usepackage{newtxtext}
	\fi
\fi

\definecolor{bcolor}{rgb}{0.53, 0.66, 0.42}

\ifnum\commentanonymous=1
\newcommand{\taiga}[1]{$\ll$\textsf{\color{blue} Taiga: { #1}}$\gg$}
\newcommand{\mor}[1]{$\ll$\textsf{\color{red} Tomoyuki: { #1}}$\gg$}
\newcommand{\matthew}[1]{$\ll$\textsf{\color{violet} Matthew: { #1}}$\gg$}
\newcommand{\bnote}[1]{$\ll$\textsf{\color{bcolor} Bruno: { #1}}$\gg$}
\newcommand{\eli}[1]{$\ll$\textsf{\color{orange} Eli: { #1}}$\gg$}
\else
\newcommand{\taiga}[1]{}
\newcommand{\mor}[1]{}
\newcommand{\matthew}[1]{}
\newcommand{\bnote}[1]{}
\newcommand{\eli}[1]{}
\fi

\newcommand{\OWP}{\mathsf{OWPuzz}}

\newcommand{\puzz}{\mathsf{puzz}}
\newcommand{\ans}{\mathsf{ans}}
\newcommand{\Ver}{\mathsf{Ver}}

\newcommand{\Samp}{\algo{Samp}}

\newcommand{\chosen}{\leftarrow}

\newcommand{\la}{\leftarrow}
\newcommand{\ra}{\rightarrow}

\newcommand{\seteq}{\coloneqq}

\newcommand{\cA}{\mathcal{A}}
\newcommand{\cB}{\mathcal{B}}
\newcommand{\cC}{\mathcal{C}}
\newcommand{\cD}{\mathcal{D}}

\newcommand{\cG}{\mathcal{G}}
\newcommand{\cH}{\mathcal{H}}
\newcommand{\cI}{\mathcal{I}}
\newcommand{\cK}{\mathcal{K}}

\newcommand{\cL}{\mathcal{L}}
\newcommand{\cM}{\mathcal{M}}
\newcommand{\cN}{\mathcal{N}}

\newcommand{\cP}{\mathcal{P}}
\newcommand{\cQ}{\mathcal{Q}}
\newcommand{\cR}{\mathcal{R}}
\newcommand{\cS}{\mathcal{S}}
\newcommand{\cT}{\mathcal{T}}
\newcommand{\cU}{\mathcal{U}}
\newcommand{\cV}{\mathcal{V}}

\def\makeuppercase#1{
\expandafter\newcommand\csname tl#1\endcsname{\widetilde{#1}}
}

\def\makelowercase#1{
\expandafter\newcommand\csname tl#1\endcsname{\widetilde{#1}}
}

\newcommand{\N}{\mathbb{N}}

\newcommand{\R}{\mathbb{R}}

\newcommand{\A}{\entity{A}}

\newcommand*{\PP}{\keys{PP}}

\newcommand*{\algo}[1]{\ensuremath{\mathsf{#1}}}

\newcommand*{\entity}[1]{\mathcal{#1}}
\newenvironment{boxfig}[2]{\begin{figure}[#1]\fbox{\begin{minipage}{0.97\linewidth}
                        \vspace{0.2em}
                        \makebox[0.025\linewidth]{}
                        \begin{minipage}{0.95\linewidth}
            {{
                        #2 }}
                        \end{minipage}
                        \vspace{0.2em}
                        \end{minipage}}}{\end{figure}}

\newcommand{\bit}{\{0,1\}}

\newcommand{\Gen}{\algo{Gen}}

\newcommand{\Vrfy}{\algo{Ver}}

\newcommand{\vrfy}{\algo{Vrfy}}

\newcommand{\negl}{{\mathsf{negl}}}

\newcommand{\poly}{{\mathrm{poly}}}

\newcommand{\zo}[1]{\{0,1\}^{#1}}

\newcommand{\complexityfont}[1]{\mathsf{#1}}

\def\PP           {\complexityfont{PP}}

\newcommand{\SD}{\Delta}
\newcommand{\E}{\mathbb{E}}

\newcommand{\BigPr}[2]{
\Pr\left[
\begin{array}{c}
#1
\end{array}
:
\begin{array}{c}
#2
\end{array}
\right]
}

\usepackage{tikz}

\ifnum\llncs=1
\let\oldvec\vec% Store \vec in \oldvec
\let\vec\oldvec% Restore \vec from \oldvec
\makeatletter
\renewcommand*\l@author[2]{}
\renewcommand*\l@title[2]{}
\makeatletter

\else
\fi    
\title{
\textbf{Cryptographic Conditions for Efficient Testing of Distributions and Quantum States
}}

\begin{document}
%\author{}
%\institute{}

\ifnum\anonymous=1
\author{\empty}
%\institute{\empty}
\else
%
%  For camera ready version.
%
\ifnum\llncs=1
\index{Taiga, Hiroka}
\author{
	Taiga Hiroka\inst{1} 
}
\institute{
	Yukawa Institute for Theoretical Physics, Kyoto University, Japan  \and NTT Corporation, Tokyo, Japan
}
\else
%
%   For full/eprint version, etc.
%
\author[1]{Bruno Cavalar}
\author[2]{Eli Goldin}
\author[1]{Matthew Gray}
\author[3,4]{Taiga Hiroka}
\author[3]{Min-Hsiu Hsieh}
\author[4,3]{Tomoyuki Morimae}
\affil[1]{{\small University of Oxford, UK}\authorcr{\small \{bruno.cavalar,matthew.gray\}@cs.ox.ac.uk}}
\affil[2]{{\small New York University, USA}\authorcr{\small eli.goldin@nyu.edu}}
\affil[3]{{\small Hon-Hai Research Institute, Taiwan}\authorcr{\small \{min-hsiu.hsieh,taiga.hiroka\}@foxconn.com}}
\affil[4]{{\small Yukawa Institute for Theoretical Physics, Kyoto University, Japan}\authorcr{\small tomoyuki.morimae@yukawa.kyoto-u.ac.jp}}

\renewcommand\Authands{, }
\fi %%%%% END OF LNCS branch
\fi

\ifnum\llncs=1
\date{}
\else
\date{}
\fi

\maketitle

%\ifnum\llncs=0
%\thispagestyle{fancy}
%\rhead{YITP-24-125}
%\else
%\fi

\pagenumbering{gobble} % Turn off page numbering temporarily

\begin{abstract}

One of the most fundamental problems in distribution testing is the identity testing problem:
given samples $x_1,\ldots,x_s$, the goal is to determine whether the samples are drawn from a target distribution $\mathcal{D}$.
When $\mathcal{D}$ is a distribution over $\bit^n$, the optimal sample complexity of identity testing is known to be $\Omega(\sqrt{2^n})$.
Furthermore, most existing results assume that the samples $x_1,\ldots,x_s$ are generated independently from an unknown distribution.

In this work, we overcome both of these limitations by initiating study of distribution testing in a more realistic setting.
In our model, the unknown distribution is promised to be efficiently samplable, while allowing the observed samples $x_1,\ldots,x_s$ to be adversarially generated and arbitrarily correlated.
Under this model, we show that polynomially many samples suffice to verify distributions.
We further characterize the computational complexity of verifying classically- and quantumly-samplable distributions.
In particular, we prove that every classically-samplable distribution is efficiently verifiable if and only if one-way functions do not exist.
For quantumly samplable-distributions, we show that efficient verification is possible assuming the non-existence of one-way puzzles, while QEFID implies the hardness of such verification.
Our techniques also extend to verifications of quantum states: every efficiently generatable quantum state is verifiable assuming the non-existence of weak non-uniform EFI, whereas the existence of EFI implies hardness.

In establishing some of our results, we employ Kolmogorov complexity techniques in a novel manner.
We also present multiple applications of Kolmogorov complexity that are of independent interest.
In particular, we show that certified randomness with a classical efficient prover can be achieved without computational assumptions when inefficient verification is allowed.
Furthermore, we also show that a natural quantum extension of a well-studied Kolmogorov complexity measure provides a good benchmark for certifying sampling-based quantum advantage.

\end{abstract}

\ifnum\cameraready=1
\else
\ifnum\llncs=1
\else
\newpage
  \setcounter{tocdepth}{2}      % sections in table if depth < i
  \setcounter{secnumdepth}{3}   % sections numbered if depth < i
  \tableofcontents
  \pagenumbering{arabic}
  \setcounter{page}{0}          % set the table contents page as 0-th page
  \thispagestyle{empty}
  \clearpage
\fi
\fi

\section{Introduction}\label{sec:intro}

One of the most fundamental problems in the field of hypothesis testing is
the following: is it possible to check whether samples from some unknown
distribution (over classical strings) $\mathcal{G}$ are actually from some explicit distribution
$\mathcal{D}$? This is often known as the identity testing problem. The
identity testing problem has been extremely well-studied in the statistical
setting
(\cite{idtestclose,optidtest,structidtest,idtestlower,idtestupper}
for some examples). Ignoring error terms, it is known that when the
distribution $\mathcal{D}$ has support $\{0,1\}^n$, the optimal sample complexity
for the identity testing problem is roughly
$O(\sqrt{2^n})$~\cite{optidtest,idtestlower,idtestupper}. As such, implementing this identity tester requires a number of samples exponential in the problem size and is untenable in practice.

\if0However, many distributions of interest have large support size. In particular, distributions which can be sampled by efficient algorithms may have support exponential in their output length. When this is the case, the optimal identity tester also requires exponential samples, and so is untenable in practice. 
\fi

One may hope to bypass this lower bound by considering restricted settings.
The above optimal identity tester is constructed so that it is not fooled
by any (even inefficiently-sampled) distributions. However, in most
applications, the distributions under consideration are efficiently-samplable,
and therefore it is enough to consider only identity testers that are not fooled by efficiently-sampled distributions.

We then have the following natural question, which is the main question of this work: 
\begin{center}
\emph{
Which distributions have identity testers which cannot be fooled by efficiently samplable distributions, and under what assumptions?
}
\end{center}

\paragraph{On the independence assumption of identity testers.}

Identity testers typically assume that the samples received all come from the same underlying distribution, i.e. they are independently and identically distributed (i.i.d.). However, in the real world, there is no way to enforce this form of independence. In general, and as we will see in particular in the setting of quantum advantage, we may want our tester to also reject correlated samples which are far from our target distribution.
\if0
If we allow the adversarial distribution to be arbitrary, then it is impossible to detect correlated samples without exponential blowup. The adversarial distribution may simply be constant on some sequence of samples accepted by our verifier\footnote{In detail, consider an identity tester for the uniform distribution over $\{0,1\}^n$ taking in polynomially many samples. This identity tester must accept some string $x_1,\dots,x_{poly}$. And so the constant distribution outputting $x_1,\dots,x_{poly}$ will be accepted by the tester, despite its marginal distribution being far from uniform over $\{0,1\}^n$.}.
However, if we restrict our adversarial distributions to be efficiently samplable, this issue no longer applies. Thus, we may hope to achieve this strong form of identity testing in our setting.
\fi

We will call a distribution \textit{efficiently verifiable} if there exists an efficient identity tester which cannot be fooled by any efficiently-samplable distribution, \textbf{even if the samples may be correlated}. We then ask
\begin{center}
    \emph{Under what conditions are distributions efficiently verifiable?}
\end{center}

\paragraph{Quantum Scenarios.}

The question raised above is also well-motivated in several quantum scenarios.
One notable scenario is verification of sampling-based quantum advantage, which is one of the most well-studied frameworks for demonstrating quantum advantage~\cite{STOC:AarArk11,BreJozShe10,BreMonShe16}.
Sampling-based quantum advantage refers to the existence of distributions that are hard to classically efficiently
sample but easy to sample by an efficient quantum algorithm. 
One often-claimed disadvantage of sampling-based quantum advantage compared with
more sophisticated frameworks, such as proofs of quantumness~\cite{FOCS:BCMVV18}, is the lack of efficient (or even inefficient) verification.
While several inefficient verification methods were proposed, these protocols rely on newly-introduced assumptions~\cite{CCC:AarChe17,AaronsonGunn}.
Moreover, these inefficient verification schemes consider search problems, and do not verify the distribution itself.
So far, there is generally no way for the public to verify whether the output of a sampling-based quantum advantage protocol was genuinely sampled from the required distribution.

\if0
Another important scenario is the verification of quantum states~\cite{PRXQuantum.2.010201,gs007}.
This problem can be interpreted as a quantum extension of classical identity testing.
Given $s$ copies of a quantum state $\rho$, the verifier aims to determine whether $\rho$ is close to a known target state $\sigma$.
When the target state consists of $n$ qubits, the optimal sample complexity is known to scale as $\Omega(2^n)$.
In addition, the standard formulation assumes that the verifier receives independent and identical copies $\rho^{\otimes s}$.
Therefore, quantum state verification inherits two fundamental limitations that also appear in classical identity testing: exponential sample complexity and a strong i.i.d.~assumption on the samples.
\fi

Another important scenario is the verification of quantum states~\cite{PRXQuantum.2.010201,gs007}. This problem can be viewed as a quantum analogue of classical identity testing. Given $s$ copies of a quantum state $\rho$, the verifier aims to decide whether $\rho$ is close to a known target state $\sigma$. When arbitrary $n$-qubit quantum states are allowed, the optimal sample complexity scales as $\Omega(2^n)$. 
Although several works~\cite{ABCL25,PFMO25} have studied verification in non-i.i.d. settings, these approaches still incur exponential sample complexity. Therefore, it is an important open question whether quantum state verification can be carried out in a genuinely non-i.i.d. manner with only polynomially many samples in realistic settings.

These two scenarios naturally raise the question of whether we can
efficiently verify quantumly-samplable distributions and quantum states in
a realistic setting.
In this work, we focus on the following open question:
\begin{center}
\emph{
What is the computational complexity of verifying quantumly-efficiently-samplable distributions, and efficiently-generatable quantum states?
}
\end{center}

\subsection{Our Results}

In this work, we answer those questions, and obtain the following two main results. 
\begin{enumerate}
    \item 
    Every classically-efficiently-samplable (quantumly-efficiently-samplable) distribution can be verified efficiently if infinitely-often one-way functions (infinitely-often one-way puzzles) do not exist (\cref{inf:owf,inf:q_verify}).
   \item We extend the result above to the quantum state setting, and show that every efficiently-generatable quantum state can be 
   verified efficiently if weak non-uniform infinitely-often EFIs do not exist (\cref{informal:state_verification}).
\end{enumerate}
    Here, efficient verification means that the number of samples is polynomial and the running time of the verification algorithm is polynomial.
    We could also consider the setting where the number of samples is polynomial but the running time of the verification algorithm is not necessarily polynomial.
    In this setting, we obtain the following results without any assumptions.
    \begin{enumerate}
        \item 
     Every classically-efficiently-samplable and quantumly-efficiently-samplable distribution can be verified with a polynomial number of samples
     by a potentially inefficient verification algorithm. 
   \item We extend the result above to the quantum state setting, and show that every efficiently-generatable quantum state can be 
   verified with a polynomial number of samples by a potentially inefficient verification algorithm.
    \end{enumerate}

In the course of proving our main results, we also obtain several results of broader interest,
including a construction of a protocol for classical certified randomness with inefficient verification~(\Cref{inf:cert_rand}), and a ``universal'' verifier for quantum sampling advantage~(\Cref{inf:qcd}).
We describe our contributions in detail below.

\subsubsection{Distribution Verification}

\paragraph{Formulating the Definition of Efficiently-Samplable-Distributions.}

Our first contribution is to initiate the study of distribution verification and quantum state verification in a realistic setting.
Informally, in distribution verification, the model assumes that the unknown distribution $\cA$ is efficiently samplable, while allowing the observed samples $x_1,...,x_s$ to be adversarially generated in a correlated way.
Quantum state verification is defined in a similar way: the model assumes that the unknown state $\rho$ is efficiently generatable, while allowing the observed samples are not necessarily $\rho^{\otimes s}$ but an entangled state over the $s$ registers.
More precisely, we define distribution verification as follows:
\begin{definition}[Informal version of Definition~\ref{def:adaptive}]\label{inf:def_ver}
We say that the algorithm $\cD$ is verifiable if, for any function $\epsilon:\N\ra(0,1)$ and any polynomial $t$, there exist a polynomial $s$ and an efficient algorithm $\Vrfy$ such that the following hold.
\begin{description}
    \item[Correctness:] If $x_1,\ldots,x_{s(n)}\in\bit^n$ are generated independently and identically according to the known target algorithm $\cD(1^n)$, then $\Vrfy\left(x_1,...,x_{s(n)}\right)$ accepts with overwhelming probability.
    \item[Adaptive-Soundness:] 
    For any uniform algorithm $\cA$ running in time $t(n)$ such that $$\mathsf{SD}(\mathsf{Marginal}_{\cA}(1^n),\cD(1^n))\geq \epsilon(n),$$ $\Vrfy\left(x_1,...,x_{s(n)}\right)$ accepts with probability at most $1-\epsilon(n)+\frac{1}{\poly(n)}$ for all sufficiently large $n\in\N$, where $x_1,...,x_{s(n)}\gets\cA(1^n)$.
Here, $\mathsf{SD}(\cdot,\cdot)$ denotes statistical distance (a.k.a.\ total variation distance), and $\mathsf{Marginal}_{\cA}(1^n)$ is the distribution defined as follows:
\begin{enumerate}
\item Run $x_1,...,x_{s}\la \cA(1^n)$.
\item Sample $i\la[s]$.
\item Output $x_i$.
\end{enumerate}
\end{description}
\end{definition}

This model is more general than the standard identity testing model in that it permits the unknown algorithm $\cA$ to produce samples with arbitrary correlations.
At the same time, it is more restrictive in that $\cA$ is promised to be efficiently samplable, rather than allowing arbitrary distributions.
We view this restriction as well motivated: in applications such as verifying quantum sampling advantage, it is natural to consider that an unknown distribution is promised to be efficiently samplable. 
It is also desirable that 
a tester would accept only 
a sampler
whose marginal distribution is close to
the target distribution. If the test passes on input $x_1,\dots,x_s$, then it is easy to recover a sample from the correct distribution by picking a random input $x_i$.

\matthew{[FOR CAMERA READY] Add in something about getting us many samples via parallel repetition.}

\paragraph{Classically-Efficiently-Samplable Distribution Verification.}
We begin by establishing a distribution-verification result for classically-efficiently-samplable distributions.
We show the connection between verification of classically-efficiently-samplable distribution and classical cryptography.

\begin{theorem}[Restated in \cref{thm:easiness}]
\label{inf:owf}
Assume that infinitely-often one-way functions (OWFs) do not exist.
Then, any PPT-samplable distribution $\cD$ is verifiable with a PPT algorithm.
\end{theorem}

Since one-way functions can be broken by a $\mathsf{NP}$ oracle, we immediately get the following corollary
\begin{corollary}
    Any PPT-samplable distribution $\cD$ is verifiable with a PPT algorithm querying an $\mathsf{NP}$ oracle.
\end{corollary}

Here, a OWF is a function $f$ that can be computed in deterministic polynomial-time, but is hard to invert for any probabilistic polynomial-time algorithm. 
OWFs are one of the most central classical cryptographic primitives~\cite{FOCS:ImpLub89,Hill99,C:GolGolMic84,STOC:Rompel90,C:Naor89}.

Let us make two remarks on this result.
First, the assumption of the non-existence of OWFs is used only for constructing a PPT verification algorithm. Without this assumption, our construction still yields a verifier that uses only a polynomially number samples though it is no longer efficient.
Importantly, even constructing such an inefficient verifier appears to be new to the field.
Second, somewhat interestingly, we use the Kolmogorov complexity of classical bit strings in a novel way to obtain \cref{inf:owf}.
The Kolmogorov complexity of a classical bit string $x$ is the length of the shortest program that outputs $x$.
To the best of our knowledge, this is the first time where the techniques of Kolmogorov complexity are applied for establishing a statement of distribution testing.

We also observe the following hardness result.

\begin{proposition}[Restated in \cref{thm:hardness}]\label{inf:hard_owf}

Assume that (infinitely-often) OWFs exist.
Then, any PPT-samplable distribution $\cD(1^n)$ with $H(\cD(1^n))=n^{\Omega(1)}$ is not verifiable with a PPT verification algorithm.
Here, $H(\cD(1^n))$ represents the Shannon entropy of $\cD(1^n)$.
\end{proposition}

Note that \Cref{inf:hard_owf} establishes a strong impossibility result.
In particular, it shows that when the entropy of the known distribution is sufficiently large, constructing an efficient verification protocol for distributions is impossible as long as OWFs exist.
Furthermore, by combining this result with \cref{inf:owf}, we can see that the hardness of verifying classically-efficiently-samplable distributions is characterized by the existence of OWFs.

The proof of~\Cref{inf:hard_owf} is relatively simple.
We rely on a pseudorandom generator (PRG), which can be derived from OWFs~\cite{Hill99}.
Let us consider a PPT-samplable distribution $\cD(1^n)$ with sufficiently large Shannon entropy.
The key idea is to replace the internal randomness $r$ used in $\cD(1^n; r)$ with pseudorandom output from the PRG.
This yields an adversarial distribution $\cA(1^n)$ that is statistically far, but computationally indistinguishable from $\cD(1^n)$.

\paragraph{Quantumly-Efficiently-Samplable Distribution Verification.}
Now, we move to the quantum-case. First, we show a connection between verification of quantumly-efficiently-samplable distributions, and two central cryptographic primitives in the quantum-computation classical-communication (QCCC \cite{DBLP:conf/crypto/AustrinCCFLM22}) setting.

\begin{theorem}[Restated in \cref{thm:OWPuzz_Verify,thm:q_hardness}]\label{inf:q_verify}
The condition 1 implies the condition 2, and the condition 2 implies the condition 3.
\begin{enumerate}
\item Infinitely-often one-way puzzles (OWPuzss) do not exist.
\item Any QPT-samplable distribution, which outputs a classical string, is verifiable with a QPT algorithm.
\item Infinitely-often QEFIDs do not exist.
\end{enumerate}
\end{theorem}

Since one-way puzzles can be broken by a $\mathsf{PP}$ oracle~\cite{Cavalar_2025}, we immediately get the following corollary.
\begin{corollary}
\label{cor:pp-ver}
    Any QPT-samplable distribution $\cD$ is verifiable with a PPT algorithm querying a $\mathsf{PP}$ oracle.\footnote{
        We also give a different, more direct, proof of \Cref{cor:pp-ver}
        without using \Cref{inf:q_verify} in \Cref{sec:BQP=PP},
        employing a quantum variant of Kolmogorov complexity.
    }
\end{corollary}

Here, OWPuzzs are a quantum analogue of OWFs~\cite{STOC:KhuTom24}.
A OWPuzz consists of a pair of algorithms $(\Gen,\Vrfy)$.
$\Gen$ is a quantum polynomial-time (QPT) algorithm that outputs a pair $(\puzz,\ans)$ of bit strings, 
and $\Vrfy$ is an unbounded algorithm that verifies their consistency.
The security condition requires that, given $\puzz$, no QPT algorithm can produce a valid $\ans$.
We also consider QEFID, which are defined as pairs of quantumly efficiently samplable distributions that are statistically far but computationally indistinguishable from each other.

Let us make two remarks on this result.
First, \cref{inf:q_verify} nearly yields a full computational complexity characterization of verifying quantumly-efficiently-samplable distributions.
On the one hand, \cref{inf:q_verify} shows that the hardness of verifying quantumly-efficiently-samplable distributions implies the existence of OWPuzzs.
On the other hand, \cref{inf:q_verify} implies that the existence of QEFID implies the hardness of verifying quantumly-efficiently-samplable distributions.
At present, we do not know how to construct QEFID from OWPuzzs, although OWPuzzs do imply a non-uniform variant of QEFID~\cite{STOC:KhuTom24,C:ChuGolGra24}.
Consequently, if QEFID can be constructed from its non-uniform variant, then the tight characterization of 
verifying quantumly-efficiently-samplable distributions 
with OWPuzzs is achieved.

Second, the proof of \cref{inf:q_verify} differs from that of \cref{inf:owf} due to a subtle technical obstacle.
In particular, the Kolmogorov-complexity-based techniques of \cref{inf:owf} rely on an average-case classical probability estimation algorithm obtained from the non-existence of OWFs.
While the non-existence of OWPuzzs is also known to imply the easiness of average-case quantum probability estimation~\cite{EC:CGGH25,C:HirMor25,STOC:KhuTom25}, the resulting guarantee is not strong enough to support the Kolmogorov-complexity-based techniques used in the proof of \cref{inf:owf}.
Consequently, we need to develop new techniques, which do not rely on Kolmogorov complexity. 
For details, see \cref{sec:tech_overview}.

\subsubsection{Quantum State Verification}
% Third, we obtain a result on quantum state verification.
Next we show that the tractability of efficient quantum state verification is tightly connected with the existence of minimal quantum cryptography.

\begin{theorem}[Restated in \cref{thm:state_verification,thm:hardness_EFI}]\label{informal:state_verification}
Condition 1 implies condition 2, and condition 2 implies condition 3.
\begin{enumerate}
\item Infinitely-often weak non-uniform EFI do not exist.
\item Any QPT-generatable mixed state is efficiently verifiable.
\item Infinitely-often EFI do not exist.
\end{enumerate}
\end{theorem}

Here, an EFI is a pair of QPT algorithms $(\Gen_0,\Gen_1)$ that output quantum mixed states that are statistically far but computationally indistinguishable from each other~\cite{ITCS:BCQ23}.
A weak non-uniform EFI is a pair of QPT algorithms $(\Gen_0,\Gen_1)$, where each $\Gen_b$ takes as input the security parameter $1^n$ and $\mu\in[n]$, and outputs a quantum state $\rho(1^n,b,\mu)$.
Informally, the requirement is that for some $\mu\in[n]$ and for some polynomial $p$, no QPT algorithms can distinguish $\rho(1^n,0,\mu)$ from $\rho(1^n,1,\mu)$ with probability greater than or equal to $$\mathsf{TD}(\rho(1^n,0,\mu),\rho(1^n,1,\mu))-\frac{1}{p(n)}.$$

\cref{informal:state_verification} nearly characterizes the complexity of quantum state verification in terms of EFI-type assumptions. Specifically, it shows that if there exists an efficiently-generatable quantum state that is hard to verify, then weak non-uniform EFI must exist. Conversely, it shows that the existence of EFI implies the hardness of verifying quantum states. Therefore, if one could construct (uniform) EFI from weak non-uniform EFI, then the hardness of quantum state verification would be characterized by the existence of EFI.

\subsubsection{Other Fundamental Results on Distribution Verification}

\paragraph{Unconditional Verification of Distributions with Small Entropy.}

So far, we have studied how to efficiently verify any general efficiently-samplable distribution or efficiently-generatable quantum states under the non-existence of certain cryptographic primitives.
It is natural to ask which family of distributions can be verified efficiently without any computational assumptions.
We also study this natural direction, and show how to verify any efficiently-samplable distributions with small entropy.

\begin{proposition}
    [Restated in \cref{thm:small_entropy}]\label{thm:unconditional_verifiaction}
Every PPT (resp. QPT)-samplable distribution \(\mathcal{D}(1^n)\) with \(H(\mathcal{D}(1^n)) = O(\log n)\) is verifiable with a PPT (resp. QPT) algorithm. Here, $H(\mathcal{D}(1^n))$ represents the Shannon entropy of $\cD(1^n)$. 
\end{proposition}

The proof is based on a simple observation. When the entropy of $\cD$ is small, the high-probability outputs of $\cD$ can be succinctly described.
This allows one, given samples from an unknown distribution, to test distance by comparing them against these representative elements.

Although \Cref{thm:unconditional_verifiaction} follows from a relatively simple observation, its statement is essentially optimal, as it already achieves the strongest guarantee one can hope for without relying on computational assumptions.
In particular, it cannot be extended to distributions with higher entropy.
This is because any classically efficiently samplable distribution with sufficiently high entropy is not verifiable assuming the existence of OWFs, as shown in \Cref{inf:hard_owf}.~\footnote{Strictly speaking, \cref{thm:unconditional_verifiaction} does not preclude the possibility of slightly improving the entropy parameter.
Indeed, \Cref{inf:hard_owf} only rules out verification of classically-efficiently-samplable distributions $\cD(1^n)$ whose entropy satisfies $H(\cD(1^n)) = n^{\Omega(1)}$.
However, we expect that our impossibility can be extended to distributions with entropy $H(\cD(1^n)) = \omega(\log n)$ under the assumption of sub-exponentially secure OWFs.
}

\paragraph{Impossibilities in the non-i.i.d Setting.}

To complement the feasibility results, we also establish impossibility results that clarify the limits of feasibility.
More specifically, in \cref{informal:state_verification,inf:q_verify,inf:owf}, we have constructed $\Vrfy$ algorithms such that $\Vrfy\left(x_1,...,x_{s(n)}\right)$ accepts with probability at most $1-\epsilon(n)+\frac{1}{\poly(n)}$ for any time-bounded uniform algorithm $\cA$ with $\mathsf{SD}\left(\mathsf{Marginal}_{\cA}(1^n),\cD(1^n)\right)\geq \epsilon(n)$.
It is natural to ask whether one can further decrease the acceptance probability below $1-\epsilon(n)$. We argue that, in general, this is impossible in our model, as suggested by the following theorem.

\begin{proposition}[Restated in \cref{thm:impossibility}]\label{informal:impossibility}
\if0
Let $\cD$ be a samplable distribution such that there exists a samplable distribution $\cD^*$ with $\mathsf{SD}(\cD^*(1^n),\cD(1^n))=1-\alpha$.
Then, for any $\beta>0$, any function $s$, and any $\Vrfy$, if
\begin{align*}
\Pr[\top\la\Vrfy\left(x_1,...,x_{s(n)}\right):x_1,...,x_{s(n)}\la\cD(1^n)^{\otimes s(n)}]\geq 1-\beta,
\end{align*}
then there exists an algorithm $\cA$ such that $\mathsf{SD}(\mathsf{Marginal}_{\cA}(1^n),\cD(1^n))\geq \epsilon(n)$, yet 
\begin{align*}
\Pr[\top\la\Vrfy\left(x_1,...,x_{s(n)}\right):x_1,...,x_{s(n)}\la\cA(1^n) ]\geq\left(1-\frac{\epsilon}{1-\alpha}\right)(1-\beta).
\end{align*}
\fi
For any function $\epsilon$, for any polynomials $s$, any distribution $\cD$, there exists an adversarial distribution $\cA$ such that 
$\mathsf{SD}(\mathsf{Marginal}_{\cA}(1^n),\cD(1^n))\geq \epsilon(n)$, and for any verifier $\Vrfy$ if
\begin{align*}
\Pr[\top\la\Vrfy\left(x_1,...,x_{s(n)}\right):x_1,...,x_{s(n)}\la\cD(1^n)^{\otimes s(n)}]=1
\end{align*}
then
\begin{align*}
\Pr[\top\la\Vrfy\left(x_1,...,x_{s(n)}\right):x_1,...,x_{s(n)}\la\cA(1^n) ]\geq 1-\epsilon.
\end{align*}
\end{proposition}

We emphasize that \cref{informal:impossibility} remains valid even when the sample complexity $s$ is allowed to be an arbitrary function. Consequently, the soundness bound in \cref{informal:state_verification,inf:q_verify,inf:owf} appears essentially tight from an information-theoretic perspective, and we therefore focus on the parameter regime specified in \cref{informal:state_verification,inf:q_verify,inf:owf}.\footnote{We remark that quantum state verification can be seen as a generalization of distribution verification. Therefore, \cref{informal:impossibility} also suggests the impossibility of significantly improving the parameter of \cref{informal:state_verification}.}

The proof of \cref{informal:impossibility} is based on a simple observation: Consider the algorithm $\cA$ that outputs samples from $D(1^n)^{\otimes s}$ with probability $1-\epsilon$, and from $D^*(1^n)^{\otimes s}$ with probability $\epsilon$, where $\cD^*$ is an arbitrary distribution satisfying $\mathsf{SD}(\cD^*(1^n),\cD(1^n))=1$.
It follows from the construction that $\Ver$ accepts with probability at least $1-\epsilon$, whereas $\mathsf{SD}(\mathsf{Marginal}_{\cA}, D) \ge \epsilon$.

Since our results concern verification of the marginal, it is natural to ask whether one could go further and {learn} the marginal distribution—or even reconstruct the underlying quantum state. Beyond the limitation implied by \cref{informal:impossibility}, we establish additional negative results for learning: in the parameter regimes relevant to applications, learning the marginal distribution of an unknown (possibly correlated) sampler is not achievable.

\begin{proposition}[Restated in \cref{thm:impossibility_learn}]\label{informal:impossibility_learn}
For any function $s$ and any constants $1/2>\alpha,\beta>0$, there is no algorithm $\mathsf{Learn}$ satisfying the following guarantee: for every uniform algorithm $\cA$ that, on input $1^n$, outputs $x_1,...,x_{s(n)}\in\bit^{n}$, 
\begin{align*}
\Pr\left[\mathsf{SD}\left(h(1^n),\mathsf{Marginal}_{\cA}(1^n)\right)\leq \alpha(n):
\begin{array}{ll}
     &x_1,...,x_{s(n)}\la \cA(1^n)  \\
     & h\la\mathsf{Learn}(1^n,x_1,...,x_{s(n)})
\end{array}
\right]\geq 1-\beta(n)
\end{align*}
for some $n\in\N$.
\end{proposition}

We prove \cref{informal:impossibility_learn} via a simple contrapositive argument.
Specifically, we show that the existence of a learning algorithm for the marginal distribution would contradict \cref{informal:impossibility}.

\paragraph{Inefficiently-Verifiable Certified Randomness Without Computational Assumptions.}
A related primitive to quantum advantage is something called ``certified randomness''. Informally, a certified randomness protocol is an interactive protocol between a prover $\cP$ and a (potentially inefficient) verifier $\cV$ such that whenever the verifier accepts, the output produced by a malicious prover must have high min-entropy. Note that it is \textit{impossible} to achieve certified randomness with a classical prover secure against non-uniform adversaries, since an adversarial prover can always hardcode the randomness of the honest prover. However, recent works have constructed certified randomness protocols with \textit{quantum} provers from the hardness of LWE~\cite{JACM:BCMVV21}, from assumptions about random quantum circuits~\cite{STOC:AarHun23}, and in the random oracle model~\cite{FOCS:YamZha22,BBFGT26}.

We will focus on a particular setting: single round publicly verifiable certified randomness (1PVCR)~\cite{BBFGT26,FOCS:YamZha22}. These are protocols where the prover publishes a single message that can be verified by anyone. In particular, the verifier will not send any challenge to the prover. Using Kolmogorov complexity in a similar way as our proof of~\Cref{inf:owf}, we construct 1PVCR unconditionally. In fact, our prover will be a \textit{classical} algorithm, with the caveat that our verifier will be inefficient. 

Recall that the impossibility of certified randomness for classical provers only applies when the adversarial prover is allowed non-uniform advice. We thus bypass this impossibility by only requiring security against \textit{uniform} adversaries. In fact, 1PVCR is anyway impossible against non-uniform adversaries, even with a quantum prover~\cite{FOCS:YamZha22}.

\begin{theorem}[Restated in \cref{thm:certified_randomness}]\label{inf:cert_rand}
There exists a pair $(\cV,\cP)$ of an unbounded algorithm $\cV$ and a PPT algorithm $\cP$ such that the following two properties are satisfied.
Here, $\cP$ takes the security parameter $1^n$ as input and outputs a classical bit string $x$, 
and $\cV$ takes $1^n$ and $x$ as input and outputs $\top/\bot$. 
\begin{description}
\item[Correctness:]
The verifier $\cV(1^n,x)$ outputs $\top$ with overwhelming probability when $x$ is generated by the honest prover $\cP(1^n)$. 
\item[(Side-Information Free) Security:] 
If a malicious, uniform QPT adversary $\cA(1^n)$ can make the verifier $\cV$ accept with non-negligible probability,
then the min-entropy of $\cA_{\top}(1^n)$ must be large. 
Here, $\cA_{\top}(1^n)$ denotes the output distribution of the algorithm $\cA(1^n)$ 
conditioned on the event that $\cV$ accepts.
\end{description}
\end{theorem}

All previous constructions of 1PVCR are in the random oracle model~\cite{BBFGT26,FOCS:YamZha22}. Like our construction, the protocol from~\cite{BBFGT26} has an inefficient verification protocol. While the construction of~\cite{FOCS:YamZha22} has an efficient verifier, its security only holds assuming the Aaronson-Ambainis conjecture. In the setting of certified randomness more generally, all existing constructions require either computational assumptions or idealized models, in addition to having a quantum prover.\footnote{As a caveat, there are a number of works constructing certified randomness unconditionally in the multi-prover setting (e.g.~\cite{STOC:VazVid12,KA17}), which we do not consider.} In comparison, our protocol's security does not require any idealized models and holds unconditionally. As such, our result is an (almost) strict improvement to the result of~\cite{BBFGT26}\footnote{The one advantage their result has is that they also handle sub-exponential time adversaries, while we restrict our adversaries to polynomial time.}.

Note that a number of existing constructions of certified randomness also provide guarantees against forms of leakage (e.g. \cite{STOC:AarHun23,JACM:BCMVV21,KA17}).\footnote{One form of leakage that has been considered~\cite{STOC:AarHun23} is the following: in this setting we strengthen the adversary by allowing them to entangle an ancilla register with its working register before the interaction begins, and during the interaction the adversary cannot access the ancilla. The stronger security definition requires that the certified randomness is ``uncorrelated'' with the information in the ancilla.} In fact, certified randomness protocols with classical provers cannot be secure against leakage, even in the interactive setting against uniform adversaries~\cite{HHM_inpreparation}.

We remark that the idea that Kolmogorov complexity is a good metric for the entropy of a distribution has been previously observed~\cite{GV04}, but as far as the authors are aware there has been no formal connection to certified randomness.

\paragraph{Kolmogorov Complexity as a Benchmark Measure for Sampling-Based Quantum Advantage.}

Besides certified randomness, we show a further application of Kolmogorov complexity to the verification of sampling-based quantum advantage.
In proving \cref{inf:cert_rand}, we introduce a quantum variant of time-bounded Kolmogorov complexity $quK^t(x)$.
Informally, 
$quK^t(x)$ is the negative logarithm of the probability that a randomly chosen program $\Pi$ of length at most $t$ when processed by a universal quantum Turing machine, outputs $x$ within $t$ steps.
We show that the gap between $quK^t(x)$ and its classical analogue $uK^t(x)$ serves as a benchmark that unconditionally certifies sampling-based quantum advantage.

\begin{theorem}[Restated in \cref{lem:QAS-qkt}]\label{inf:qcd}
Let us define $qcd^t(x)\seteq uK^t(x)-quK^t(x)$.
The value $qcd^t(x)$ satisfies the following two conditions:
\begin{description}
    \item[Quantum Easiness:]
    For any QPT algorithm $\cD_Q$ such that, for any PPT algorithm $\cD$, $\mathsf{SD}(\cD_Q(1^n), \cD(1^n))\geq 1-\negl(n)$,
    the value $qcd^t(x)$ must be large with high probability over $x\la\cD_\cQ(1^n)$.
    \item[Classical Hardness:] For any PPT algorithm $\cA$, the value $qcd^t(x)$ must be small with high probability over $x\la\cA(1^n)$.
\end{description}
\end{theorem}
%There is one remark. 
We remark that,
in the quantum easiness requirement above, 
a QPT algorithm $\cD_Q(1^n)$ must satisfy \emph{strong} quantum advantage (i.e.
$\mathsf{SD}(\cD_Q(1^n), \cD(1^n))\geq 1-\negl(n)$ for any PPT algorithm $\cD$).
This is stronger than what is achieved in
the standard sampling-based quantum advantage proofs~\cite{STOC:AarArk11,BreMonShe16}
where
a QPT algorithm $\cD_Q(1^n)$ satisfies 
\emph{weak} quantum advantage
(i.e.
$\mathsf{SD}(\cD_Q(1^n), \cD(1^n))\geq \epsilon$ with some constant $0<\epsilon<1$ for any PPT algorithm $\cD$).
However, we show in \cref{app:strong_qas} that
%what is required in the quantum easiness, namely,
%that
%$\cD_Q(1^n)$ satisfies
%$\mathsf{SD}(\cD_Q(1^n), \cD(1^n))\geq 1-\negl(n)$ for any PPT algorithm $\cD$, 
strong quantum advantage
is possible assuming the existence of quantum PRGs (QPRGs) secure against $QPT^{\mathbf{NP}}$.

Based on the properties of $qcd^t(x)$ in~\cref{inf:qcd}, we construct a universal verification algorithm for quantum sampling-based advantage.
In particular, the verifier accepts with high probability on samples generated by any QPT algorithm that is hard to simulate by PPT algorithms, while rejecting with high probability on samples generated by any PPT algorithm.
\begin{corollary}[Restated in \cref{thm:QAS_Verification}]\label{inf:QAS_Verification}
There exists a potentially inefficient algorithm $\Vrfy$ satisfying the following properties, and it can be implemented by a QPT algorithm if infinitely-often OWPuzzs do not exist.
\begin{description}
\item[Correctness:] If $x\in\bit^n$ is generated by an arbitrary QPT algorithm $\cD_\cQ(1^n)$ such that for any PPT algorithms $\cD(1^n)$, $\mathsf{SD}(\cD_{\cQ}(1^n),\cD(1^n))\geq 1-\negl(n)$, then $\Vrfy(x)$ accepts with high probability.
\item[Soundness:] 
If $x\in\bit^n$ is generated by any PPT algorithm $\cD(1^n)$, then $\Vrfy(x)$ rejects with high probability.
\end{description}
\end{corollary}

As far as we know, all previous benchmarks of quantum advantage (such as \cite{CCC:AarChe17,AaronsonGunn}) 
are model-dependent: the benchmarks are constructed based on specific properties
of the model (such as the randomness of random quantum circuits). In contrast, \cref{inf:QAS_Verification} is model-independent: any distribution that is not PPT samplable can be detected.

We emphasize that the verification algorithm constructed in \cref{inf:QAS_Verification} is incomparable with the one obtained in \cref{inf:q_verify}.
In \cref{inf:q_verify}, for each quantum algorithm $\cQ$, we constructed a verifier $\Vrfy_{\cQ}$ that tests whether an unknown algorithm is close to the specific algorithm $\cQ$. 
Consequently, $\Vrfy_{\cQ}$ may reject even when an unknown algorithm $\cD^*$ achieves sampling-based quantum advantage because $\cD^*$ might be statistically far from $\cQ$. 
In contrast, \cref{inf:QAS_Verification} enables us to construct a single verifier $\Vrfy$ that certifies whether $x$ is generated by an arbitrary quantum algorithm that is hard to simulate for any PPT algorithm.

Beyond the applications discussed above, we believe that the quantity $qcd^t(x)$ introduced in this paper has further theoretical applications.
In particular, it can be regarded as a quantum generalization of computational depth $cd^t(x)$, defined as the gap between $uK^t(x)$ and $K(x)$ introduced in \cite{AntFor06}.
Computational depth has been shown to admit several applications in worst-case to average-case reductions.\footnote{As a notable example, it plays a central role in the recent breakthrough establishing that
$\mathbf{PH}\nsubseteq \mathbf{DTIME}(2^{O(n/\log n))}$ (reps. $\mathbf{UP}\nsubseteq \mathbf{DTIME}(2^{O(n/\log n))}$) implies $\mathbf{DisPH}\nsubseteq \mathbf{AvgP}$ (resp. $\mathbf{DisNP}\nsubseteq \mathbf{AvgP} $)~\cite{CCC:GKLO22,STOC:Hir21}.
To the best of our knowledge, every existing proof fundamentally relies on computational depth, even though the statement itself appears entirely unrelated to Kolmogorov complexity.
In addition, computational depth has found further applications in worst-case characterizations of one-way functions~\cite{TCC:LiuPas24,FOCS:HirNan23}.
}
It is therefore natural to ask whether its quantum analogue $qcd^t(x)$ can serve a similar role in the study of quantum average-case complexity.
We leave this question for future work.

\subsection{Technical Overview}\label{sec:tech_overview}

\newcommand{\Dis}{\mathsf{Dis}}
\newcommand{\Marg}{\mathsf{Marginal}}

We will begin by presenting the classical variation of our main result (\Cref{inf:owf}):
if one-way functions do not exist, then every classically samplable distribution is verifiable. We will first construct a simple
proof using tools from the field of meta-complexity.
This approach will fail in the quantum setting, but will lead to several new and interesting applications. We will then detail the (more complex) techniques needed to verify quantum samplable distributions (over both classical strings and quantum states) from assumptions about quantum cryptography.

\paragraph{Primer on meta-complexity.}
\textit{Meta-complexity} studies the computational complexity of problems about computational complexity. This field was initiated by the study of a metric on strings known as \textit{Kolmogorov complexity}~\cite{SOLOMONOFF19641,Kolmogorov1968ThreeAT,Chai69}. The Kolmogorov complexity of a string $x$ (denoted $K(x)$) is the length of the shortest Turing machine outputting $x$. It turns out that this (and related notions) are very good proxies for the amount of "randomness" contained in a string.

Importantly for us, Kolmogorov complexity satisfies two key properties on strings output by samplable distributions: a \textit{coding theorem} and \textit{incompressability}. 
\begin{enumerate}
    \item Coding theorem: For all algorithms $\mathcal{D}$, for all inputs $x$,
    $$K(x) \leq \log_2 \frac{1}{\Pr[x \gets \mathcal{D}]} + O(1).$$
    \item Incompressability: For all algorithms $\mathcal{D}$, with high probability over $x\gets \mathcal{D}$,
    $$K(x) \geq \log_2 \frac{1}{\Pr[x \gets \mathcal{D}]} - O(1).$$
\end{enumerate}
Together, these two properties give that for $x$ output by any samplable distribution, $2^{-K(x)}$ is a very good proxy for the probability that the distribution outputs $x$. 

We further note that this error is one-sided. $K(x)$ will always be a upper bound on $-\log \Pr[x \gets \mathcal{D}]$.

\paragraph{Starting point: Aaronson's inefficient verification algorithm.}
Aaronson~\cite{Aaronson14} used this idea to show the equivalence between $\mathbf{SampBQP}=\mathbf{SampBPP}$ and $\mathbf{FBQP}=\mathbf{FBPP}$. At the core of his idea is an \textit{inefficient} verification of distributions protocol using polynomially-many samples. The key idea is as follows, since $K(x)$ is a good proxy for $-\log \Pr[x \gets \mathcal{D}]$, we can verify that $x$ really was sampled from $\mathcal{D}$ just by comparing $K(x)$ and $-\log \Pr[x \gets \mathcal{D}]$!

In particular, Aaronson's verifier works as follows. On input $y_1,\dots,y_s$ (for some particular polynomial $s$), it accepts if and only if
$$K(y_1,\dots,y_s)\geq \log_2\left(\frac{1}{\Pr[y_1\gets \mathcal{D}]\cdots \Pr[y_s\gets \mathcal{D}]}\right) - (\log n)^2.$$
It follows from Aaronson's argument that the following two properties are satisfied
\begin{itemize}
    \item 
{\bf Correctness.} \(\Vrfy\left(y_1,\dots,y_{s}\right)=\top\) with high probability when
$(y_1,\dots,y_{s})\la\cD^{\otimes s(n)}$.
\item
{\bf Soundness.} 
For any 
uniform adversary $\cA$ that outputs $(y_1,...,y_{s})$ such that 
the marginal distribution
\(\Marg_{\cA}\) is statistically far from \(\cD\),
\(\Vrfy\left(y_1,\dots,y_{s}\right)=\bot\) with high probability.
Here, \(\Marg_{\cA}\) denotes the algorithm that samples \(y_1,\dots,y_{s(n)}\la\cA\),
chooses $i\gets [s]$, and outputs \(y_i\).
\end{itemize}

In particular, correctness follows immediately from the incompressibility applied to the distribution $\mathcal{D}^{\otimes s}$. On the other hand, if \(\Marg_{\cA}\) is statistically far from \(\cD\), a simple probabilistic argument shows that 
$${\Pr[(y_1,\dots,y_s)\gets \A]} \gg {\Pr[y_1\gets \cD]\cdots \Pr[y_s\gets \cD]}$$
and so by coding theorem
$$K(y_1,\dots,y_s) \approx \log_2\left(\frac{1}{\Pr[(y_1,\dots,y_s)\gets \A]}\right) \ll \log_2\left(\frac{1}{\Pr[y_1\gets \cD]\cdots \Pr[y_s\gets \cD]}\right).$$
The verifier rejects outputs of $\cA$ and is thus sound.

In this construction, although the number of samples is polynomial, one crucial issue of the construction is that Kolmogorov complexity is in general uncomputable, and so this verifier cannot be implemented by any Turing machine (even inefficient ones)\footnote{\cite{Aaronson14} also discusses computable variations of his verifier, but these do not have efficient implementations.}. 

\paragraph{Making Aaronson's verification efficient.}We observe that the \textit{only} properties of Kolmogorov complexity used by Aaronson are the coding theorem and incompressability. And so the same inefficient verifier works if we replace $K(x)$ with \textit{any} metric satisfying these two properties. Fortunately, it turns out that there is such a metric which is additionally easy to compute \textit{as long as one-way functions do not exist}.

In particular, we use a variant of time-bounded Kolmogorov complexity which we denote $uK^t$ (\cite{FOCS:HirNan23} calls this $q^t$), which is defined as the negative logarithm of the output probabilities of the \textit{universal classical time-bounded distribution}. Formally,
$$uK^t(x) = -\log_2 \left(\Pr_{\Pi\gets \{0,1\}^t}\left[x\gets \mathcal{U}^t(\Pi)\right]\right)$$
where $\mathcal{U}$ is any universal Turing machine and $\mathcal{U}^t(\Pi)$ is the output tape of $\mathcal{U}$ on input $\Pi$ after $t$ steps. The complexity measure $uK^t(x)$ is known to satisfy a coding theorem and incompressability~\cite{LiVitanyi93,FOCS:HirNan23}. And so Aaronson's verifier works when we replace $K(x)$ with $uK^t(x)$.

In order to implement Aaronson's verifier efficiently, we need an efficient algorithm to compute the two values $\Pr[y_1\gets \cD]\cdots \Pr[y_s\gets \cD]$ and $uK^t(y_1,\dots,y_s)$.
For both of these tasks, we use the well-known equivalence between the existence of one-way functions and the task of \textit{probability estimation}~\cite{FOCS:ImpLev90,STOC:IlaRenSan22,CHK25}. In particular, if one-way functions do not exist, then for any efficient distribution $\cD$ there exists an algorithm $\mathsf{Approx}_{\cD}$ such that with high probability over $x\gets \cD$, $\mathsf{Approx}_{\cD}(x) = (1\pm \epsilon)\cdot \Pr[x\gets \cD]$. Furthermore, this algorithm operates with \textit{one-sided error} in the following sense: for \textit{all} $x$,
$$\mathsf{Approx}_{\cD}(x) \leq \Pr[x\gets \cD]$$
with high probability over the internal randomness of $\mathsf{Approx}_{\cD}$.

This gives us an efficient means of estimating $\Pr[y_1\gets \cD]\cdots \Pr[y_s\gets \cD]$ as long as one-way functions do not exist. The ability to solve probability estimation also gives rise to an approximation algorithm for $uK^t$~\cite{FOCS:ImpLev90,FOCS:HirNan23}. Formally, if one-way functions do not exist, then there exists a PPT estimation algorithm $\mathsf{Approx}_{uK^t}$ such that with high probability over any PPT-samplable distribution $\cD$, $\mathsf{Approx}_{uK^t}(x)\approx uK^t(x)$.

And so, we can efficiently implement Aaronson's verifier by simply checking whether
$$-\log_2\left(\mathsf{Approx}_{uK^t}(y_1,\dots,y_s)\right) \geq -\log_2 \left(\mathsf{Approx}_{\cD^{\otimes s}}(y_1,\dots,y_s)\right)-(\log n)^2.$$

Correctness follows immediately from combining the correctness of the approximation algorithms with incompressibility. On the other hand, there is a slight problem with soundness. In particular, the behavior of $\mathsf{Approx}_{\cD^{\otimes s}}$ is not guaranteed when $y_1,\dots,y_s$ are adversarially produced and do not come from $\cD^{\otimes s}$. However, one-sided error guarantees that the value received is at least a lower bound, and so the verifier will still reject on far sources.
%\footnote{In fact, our estimate of $uK^t(y_1,\dots,y_t)$ will always be correct, since all samplable distributions have constant overlap with the universal distribution produced by random Turing machines.\taiga{I think that the estimation algorithm of $uK^t(x)$ is not sufficiently explained, so it might be unclear.}
Putting this all together,~\Cref{inf:owf} follows.

\paragraph{Failure in the quantum setting.} As in the classical setting, one can define a quantum analogue $quK^t(x)$ of the classical time-bounded Kolmogorov complexity $uK^t(x)$ and show that $quK^t(x)$ satisfies both incompressibility and the coding theorem.
Consequently, if one can efficiently test
\begin{align*}
quK^t\left(y_1,...,y_{s(n)}\right) \geq \log_{2}\!\left(\frac{1}{\Pr[(y_1,...,y_{s(n)})\leftarrow\cD(1^n)^{\otimes s(n)}]}\right) - (\log n)^2    
\end{align*}
for any QPT algorithm $\cD$, then it would be possible to verify any QPT-samplable distributions (over classical strings).
Indeed, a deterministic polynomial-time algorithm with access to a $\mathbf{PP}$ oracle can compute both
$\Pr[(y_1,\dots,y_{s(n)})\leftarrow\mathcal{D}(1^n)^{\otimes s(n)}]$
and
$quK^t(y_1,\dots,y_{s(n)})$
for all $y_1,\dots,y_{s(n)}$~\cite{FR99}.
Therefore, such an algorithm can verify any QPT-samplable distributions.

Just as in the classical setting, one could hope to show that probability estimation for quantum samplable sources is equivalent to some quantum cryptographic object. Indeed, \cite{EC:CGGH25,C:HirMor25,STOC:KhuTom25} showed that if no
OWPuzzs exist, then, for every QPT algorithm $\cD$, there exists a QPT algorithm $\mathsf{Approx}$ that can estimate $\Pr[y\la\cD]$ with high probability on average over $y\gets \cD$.
This might suggest that the $\mathbf{PP}$ oracle in the above construction could be replaced using the probability estimator achieved from the non-existence of OWPuzzs. However, the quantum setting lacks the one-sided error property necessary to achieve soundness in the classical setting. Intuitively, the one-sided error property comes from the fact that given some randomness, one can verify whether it indeed produces the correct output. This is a task not always possible quantumly.

\paragraph{Verifying quantum samplable distributions (of classical strings).}
We now give intuition for our proof of~\Cref{inf:q_verify}. Our verifier for quantum samplable distributions begins with the following simple observation: for any given fixed constant length, there are only a constant number of Turing machines to consider. This gives us the following template for a distinguisher: First, develop a way to distinguish samples from $\cD^{\otimes s}$ from samples from $\A$ for a \textit{fixed} adversarial distribution $\A$. Then, on input $(y_1,\dots,y_s)$, try every single possible adversarial distribution $\A$ and run the corresponding distinguisher. If $(y_1,\dots,y_s)$ came from $\cD$, then all of these distinguishers should output "$\cD$". But if $(y_1,\dots,y_s)$ came from a bad distribution, then one of them should output "$\A$". Note that this approach requires the individual distinguishers to succeed with very high probability in order to maintain correctness when taking a union bound over all adversarial distributions.\footnote{It also requires that the individual distinguishers always return "$\cD$" when the distributions are close, but this will happen automatically in our approach.}

At the heart of instantiating this approach lies a (nearly) optimal \textit{universal distinguisher}. In particular, this will be an algorithm $\Dis_{\cD_0,\cD_1}(x)$ which distinguishes between samples from distributions $\cD_0,\cD_1$ with optimal advantage for all $\cD_0,\cD_1$ samplable by constant length Turing machines. Formally, let $\SD(\cD_0,\cD_1)$ be the total variation distance between $\cD_0$ and $\cD_1$. Then
\begin{equation*}
    \begin{split}
        \Pr_{x\gets \cD_1}[1\gets\Dis_{\cD_0,\cD_1}(x)] - \Pr_{x\gets \cD_0}[1\gets \Dis_{\cD_0,\cD_1}(x)] \approx \SD(\cD_0,\cD_1).
    \end{split}
\end{equation*}

% \taiga{I would like to suggest modify From here}
As a considerably simplified example, consider the setting where the adversarial distribution $\A^{\otimes s}$ is non-adaptive and so consists of $s$ independent samples. Then we could check which distribution $(y_1,\dots,y_s)$ comes from by running $\Dis_{\cD,\A}(y_1),\dots,\Dis_{\cD,\A}(y_s)$ and seeing if the proportion of answers saying "$\cD$" is larger than $\Pr_{x\gets \cD}[1\gets \Dis_{\cD,\A}(x)]$. 

The key challenge then is to develop a verifier which also works for adaptive adversaries $(y_1,\dots,y_s)\gets \A$. Our insight will be to focus our attention on the individual marginal distributions of $\A$, followed by a careful probabilistic argument.

To verify whether $y_1,\dots,y_s$ came from some fixed distribution $\cD$, our verifier will do the following
\begin{enumerate}
    \item For all distributions $\cA$ samplable by a constant-size quantum Turing machine:
    \begin{enumerate}
        \item Define $\Marg_\A^i$ to be the $i$th marginal of $\A$. That is, the distribution derived by sampling $(y_1,\dots,y_s)\gets \A$ and outputting $y_i$.
        \item Run $\Dis_{\cD,\Marg_\A^1}(y_1),\dots,\Dis_{\cD,\Marg_\A^s}(y_s) \to a_1^{\A},\dots,a_s^{\A}$.
        \item Sample $(x_1,\dots,x_s)\gets \cD^{\otimes s}$. Run $\Dis_{\cD,\Marg_\A^1}(x_1),\dots,\Dis_{\cD,\Marg_\A^s}(x_s) \to b_1^{\A},\dots,b_s^{\A}$.
        \item Define $A^\A = \frac{1}{s}\sum_i a_i$, $B^\A = \frac{1}{s}\sum_i b_i$.
        \item If the gap $\abs{A^\A-B^\A}$ is small, reject.
    \end{enumerate}
    \item If none of the prior tests reject, then accept.
\end{enumerate}
Note that since there are only a constant number of constant-size Turing machines, assuming $\Dis$ is efficient, this procedure is efficient is as well.

Correctness follows from a simple concentration bound. Fix any adversary $\A$. For $(y_1,\dots,y_s)\gets \cD$, then $a_i^\A$ and $b_i^\A$ are identically distributed. And so $$\E[A^\A - B^\A] =\frac{1}{s}\sum_i \E[a_i^\A-b_i^\A] = 0.$$
Since $A^\A-B^\A$ is the sum of i.i.d. random variables, it is close to its expectation with all but inverse exponential probability. Taking the union bound over all Turing machines, the verifier will accept honest samples with very high probability.

Adaptive soundness also follows from analyzing $\E[A^\A-B^\A]$. In particular, fix any adversary $\A$ such that $\SD(\cD,\Marg_\A)$ is large. We show that with high probability, $A^\A - B^\A$ is also large, and so the verifier will reject in the round corresponding to $\A$. However, if $(y_1,\dots,y_s)$ came from $\A$, then we have for any fixed $i$,
$$\E\left[\abs{b_i^\A-a_i^\A}\right]=\abs{\Pr_{x_i\gets\cD}[\Dis_{\cD,\Marg_\A^i}(x_i)\to 1] - \Pr_{x\gets\Marg_\A^i}[\Dis_{\cD,\Marg_\A^i}(x_i)\to 1]}$$ $$\approx \SD(\cD,\Marg_\A^i).$$
Writing down the expectation of $B^\A-A^\A$, we get 
$$\E\left[\abs{B^\A-A^\A}\right]=\E_i\left[\abs{b_i^\A-a_i^\A}\right] \approx  \E_i[\SD(\cD,\Marg_\A^i)] = \SD(\cD,\Marg_\A).$$
By Markov's inequality, $B^\A-A^\A$ will by large with reasonably high probability.

\paragraph{Implementing the universal distinguisher efficiently.}

It thus remains to be shown how to efficiently implement the universal distinguisher $\Dis$. We will do this by assuming the non-existence of one-way puzzles. Although the proof of this statement in detail is technically involved, the intuition is fairly simple. 

It is well known that if one-way puzzles do not exist, it is possible to solve the task of \textit{universal extrapolation}~\cite{C:ChuGolGra24}. In particular, there exists an efficient algorithm $\mathsf{Ext}$ such that for all efficient samplers $(x,y)\gets \mathsf{Samp}$, $(x,y)$ and $(x,\mathsf{Ext}(x))$ are close in statistical distance. In other words, $\mathsf{Ext}_{\mathsf{Samp}}(x)$ samples from the conditional distribution on $y$ given $x$. This leads to the following distinguisher for $\cD_0,\cD_1$. Define the sampler $\mathsf{Samp}_{\cD_0,\cD_1}$ which samples $b\gets \{0,1\}$, $x_b\gets \cD_b$, and outputs $(x_b,b)$. By definition, $\mathsf{Ext}_{\mathsf{Samp}_{\cD_0,\cD_1}}(x_b)=b$ with good probability, and so $\mathsf{Ext}_{\mathsf{Samp}_{\cD_0,\cD_1}}$ is a pretty good distinguisher. Amplifying with parallel repetition gives a distinguisher between $\cD_0$ and $\cD_1$ which succeeds with almost optimal probability, i.e. $\SD(\cD_0,\cD_1)$.

\paragraph{Verifying quantum states.}

It turns out that a similar approach can also be used to verify quantum states (\Cref{informal:state_verification}). In this setting, let $\rho_{\cD}$ be the output state of some quantum sampler $\cD$. Given $\rho_{\cR[1],...,\cR[s]}\la\cA$, our goal is to test whether $\rho_{\cR[1],...,\cR[s]}$ is equivalent to a known $\rho_{\cD}$, or whether the marginal state $$\Marg_{\cA}\seteq \frac{1}{s(n)}\sum_{i\in[s]}\rho_{\cR[i]}$$ is statistically far from $\rho_{\cD}$.

At a high level, our distribution verification protocol only relies on having an efficient algorithm which optimally distinguishes between two arbitrary distributions. One may hope that if we could efficiently solve the analogous task for quantum states, we could also solve the verification problem for quantum states. With some work to handle errors induced by measuring the input state, it turns out that this is indeed possible.

In particular, we get that if weak non-uniform EFI pairs do not exist, then efficiently generatable quantum states are also efficiently verifiable. As stated in the introduction, as a sufficiently strong amplification theorem for EFI pairs is not known, we do not manage to show such a result from the non-existence of EFI.

\subsubsection{Other applications of Kolmogorov based verification}

% \taiga{Remove the section below?}
Although meta-complexity is not necessary to prove our main theorems, the techniques used naturally give rise to several other interesting related implications. We discuss two here.

\paragraph{Unconditional certifiable randomness.}

We now give a short sketch of~\Cref{inf:cert_rand}. In particular, our (inefficient) certifiable randomness protocol will simply accept strings with high Kolmogorov complexity. The honest prover will be the uniform distribution. Coding theorem says that distributions with low entropy will output strings with low Kolmogorov complexity and be rejected, while incompressibility says that the uniform distribution will output strings with high Kolmogorov complexity and thus be accepted.

\paragraph{A universal verifier for quantum advantage.}

Recall from the introduction that we claim that if a distribution outputs strings with high "quantum computational depth" $qcd^t(x) = uK^t(x) - quK^t(x)$, then it must demonstrate sampling-based quantum advantage (\Cref{inf:qcd}). If one-way puzzles do not exist, we can estimate computational depth by estimating both $uK^t(x)$ and $quK^t(x)$~\cite{EC:CGGH25}. This gives a natural "universal verifier" $\Ver$ for quantum advantage samplers, consisting of checking if $qcd^t(x)$ is large (\Cref{inf:QAS_Verification}).
%If $\Ver$ accepts outputs from a distribution, then that distribution must be a quantum advantage sampler. 
We now sketch an argument for the correctness of this universal verifier.

For all classical samplers $x\gets \A$, by coding theorem and incompressibility, with high probability $uK^t(x)\approx quK^t(x)\approx -\log_2 \Pr[x\gets \A]$ and so samples produced by classical distributions will have low quantum computational depth.

Therefore, it remains to show that quantum advantage samplers output strings with high $qcd^t$. Let $\mathsf{QAS}$ be any distribution demonstrating sampling-based quantum advantage. By the coding theorem, since $\mathsf{QAS}$ is efficiently samplable by a quantum machine, we must have $quK^t(x)$ small for $x\gets \cD$. Furthermore, since \textit{no} classical machine can sample from $\mathsf{QAS}$, random samples $x\gets \mathsf{QAS}$ must have high $uK^t(x)$.\footnote{Formally, any distribution which outputs strings with low $uK^t(x)$ must have some non-trivial overlap with the classical universal distribution. This contradicts $\mathsf{QAS}$ demonstrating quantum advantage.} And so for any quantum advantage sampler, it must with high probability outputs strings with high quantum computational depth.

% \subsubsection{Lower bounds on distribution verification}

% \eli{This should all go in intro in a single sentence each, it is too easy.}
% \paragraph{Hardness of distribution verification over classical distributions from one-way functions}

% To contrast with our positive results, we can show that if one-way functions exist, then every PPT-samplable distribution $\cD$ with min-entropy at least $n^{c}$ for $c>0$ is not verifiable. In particular, let $\Ver$ be any verifier for $\cD$, and let $\cD(r)$ be the output of the sampler for $\cD$ on input $r$. Let $G$ be any pseudorandom generator. Consider the adversarial distribution $\A(r') = \cD(G(r))$. By pseudorandomness, $\A(r')$ is indistinguishable from $\cD$, and so $\Ver(\A(r_1'),\dots,\A(r_s'))=1$ with high probability. But since $G$ is expanding, the output distribution of $\A$ is statistically far from $\cD$, and so $\Ver$ fails when given outputs of $\A$.

% \paragraph{Hardness of verification in the quantum settings}
% In the quantum setting, we simply observe that an EFI(D) is by definition an unverifiable distribution, and so if EFI(D) exists then quantum distributions over strings/states respectively are unverifiable.

\subsection{Related Works}

\paragraph{Cryptography.}

In classical cryptography the existence of one-way functions (OWFs) is a minimal assumption, and is known to be existentially equivalent to a large class of primitives such as PRGs, PRFs, SKE, MACs, digital signatures, and commitment schemes~\cite{FOCS:ImpLub89,Hill99,C:GolGolMic84,STOC:Rompel90,C:Naor89}.
In the quantum setting a variety of analogous—but potentially weaker—primitives have been proposed, including pseudorandom unitaries and state generators (PRUs/PRSGs)~\cite{C:JiLiuSon18}, one-way state generators (OWSGs)~\cite{C:MorYam22}, one-way puzzles (OWPuzzs)~\cite{STOC:KhuTom24}, and EFI pairs~\cite{ITCS:BCQ23}; these primitives support many applications (private-key quantum money, SKE, MACs, signatures, commitments, MPC, etc.)~\cite{C:JiLiuSon18,C:AnaQiaYue22,C:MorYam22,AC:Yan22,C:BCKM21b,EC:GLSV21}.

Among these primitives, OWPuzzs~\cite{STOC:KhuTom24} and EFI pairs~\cite{ITCS:BCQ23} are especially important. Most cryptographic primitives (with the notable exception of EFI) imply the existence of OWPuzzs~\cite{C:ChuGolGra24}, and the existence of OWPuzzs imply $ \mathbf{BQP}\neq \mathbf{PP}$~\cite{Cavalar_2025}. 
Moreover, the existence of classically-secure OWPuzzs is equivalent to that of inefficiently-verifiable proofs of quantumness~\cite{STOC:MorShiYam25}.
Our results show that the hardness of verification of quantumly-samplable distributions is sandwiched between OWPuzzs and QEFID.

EFI pairs, in turn, are implied by nearly all nontrivial quantum cryptographic primitives~\cite{ITCS:BCQ23} and are believed to be resistant to attacks that use classical oracles~\cite{STOC:LomMaWri24}.
It is known that non-uniform EFI can be converted into EFI~\cite{TCC:HKNY24}.
Furthermore, weak EFI can be amplified to EFI~\cite{STOC:BQSY24}. However, their technique can be applied to weak EFI with certain parameter constraints. 
As a consequence, at present, it is an open problem how to construct EFI from weak non-uniform EFI.
Our results imply that the hardness of verifying quantum states lies between EFI and weak non-uniform EFI.

\paragraph{Verification of distributions.}

Verification of distributions (a.k.a.\ the identity testing problem) has a
long
history~\cite{idtestclose,optidtest,structidtest,idtestlower,idtestupper,Canonne20}.
However, to the best of our knowledge, few works have studied the setting
in which the unknown distribution is promised to be efficiently samplable.
Therefore, here we will only mention the closely related result of
Aaronson~\cite{Aaronson14}.

Aaronson~\cite{Aaronson14} showed that
$\mathbf{SampBQP}\neq\mathbf{SampBPP}$ if and only if
$\mathbf{FBQP}\neq\mathbf{FBPP}$. In proving this result, Aaronson
implicitly showed that every efficiently-samplable distribution is
verifiable with \emph{only polynomially many samples}. However,
the paper does not give a formal definition of verification, and the
verification algorithm presented requires exponential time.\footnote{\cite{Aaronson14} already observed the issue, 
and pointed out that by considering space-bounded Kolmogorov complexity,
at least $\mathbf{PSPACE}$ verification is possible.} In fact, it is
often claimed that verification of sampling-based quantum advantage
requires exponentially many samples~\cite{HKEG19}\footnote{\cite{HKEG19} considers the setting where the unknown distribution is not-necessarily efficient, and therefore cannot be applied to our setting (and the standard setting in sampling-based quantum advantage) where the unknown distribution is efficiently samplable.}; hence the
quantum-advantage community may not be aware that polynomially many samples
suffice for distribution verification under a reasonable definition. By
contrast, we formally define distribution verification and improve the
exponential-time verification algorithm.

\if0
\paragraph{Identity testing beyond independent and identically distributed samples.}

Before our work, only a few works studied the identity testing
problem when the samples are not independent and identically distributed (i.i.d.).
Here, we only mention the related result of Garg et al. \cite{garg2023}.
Their work also considers identity testing beyond the i.i.d. case.
Their definition is different from ours, and we now explain the differences.

First, in their definition, an unknown algorithm is not allowed to output correlated samples.
In their definition, an unknown algorithm $\cA$ is allowed to generate
samples $x_1,\dots,x_s$ from not necessarily identical distributions
$\cD_1,\dots,\cD_s$, but these are not allowed to be correlated. 
More specifically, the distribution $\cD_{i+1}$ cannot depend on the output
of $\cD_1,\dots,\cD_i$.
On the other hand, our definition does not have such a restriction, and
thus our setting is more general.
Second, in their definition, the unknown algorithm $\cA$ is not promised to be efficiently-samplable. 
Hence, their protocol requires exponentially many samples to verify the unknown distribution.
On the other hand, in our definition, an unknown algorithm is promised to
be efficiently-samplable, and thus we are able to show that polynomially-many samples suffice.
\fi

\paragraph{Distribution Learning.}

Solomonoff initiated the foundations of learning theory~\cite{SOLOMONOFF19641,Solomonoff1964_2}.
He proposed a general method for predicting subsequent bits from samples drawn according to an unknown probability distribution, which need not generate bits independently and identically.
While subsequent work introduced  several other learning models~\cite{Valiant84,KMRRSS94,NaoRot06}, such as PAC learning and distribution learning, Solomonoff induction can be viewed as a more general conceptual framework that subsumes many of these models as demonstrated in~\cite{FOCS:HirNan23}.

One might hope to derive our results from Solomonoff’s techniques.
However, there are two obstacles to doing so.
First, Solomonoff’s theory applies to unknown probability distributions, and it is not clear how to adapt his techniques to unknown quantum states.
Second, our verification task requires testing whether the \emph{marginal} of an unknown distribution is far from a given target distribution.
By contrast, Solomonoff prediction focuses on sequence prediction and does not directly yield methods for  learning or verifying marginal distributions.\footnote{In fact, in \cref{informal:impossibility_learn}, we show that learning marginal distribution is impossible in principle. Therefore, it is unclear how to apply Solomonoff's induction to verifying marginal distributions.}
For these reasons, we develop a different approach tailored to the quantum setting and to verifying marginal distributions.

\paragraph{Quantum State Learning.}
Several works study verification of quantum states in the non-i.i.d.~setting, focusing primarily on pure states~\cite{MorTakHay17,MorTak18,ZhuHay19}.
These techniques crucially rely on the fact that the target quantum state is pure, and extending them to mixed-state targets appears non-trivial.
A smaller set of works considers learning or verifying mixed states in non-i.i.d.\ settings, but their definitions differ from ours in ways that make it unclear how to apply their methods in our setting. We briefly highlight these differences below.

In \cite{ABCL25}, they consider the setting, where one receives quantum states $\rho^{\otimes N}$.
Some fraction of the copies might be arbitrarily corrupted, and the task is to learn the original state  prior to the corruption.
By contrast, in our work all quantum states are generated maliciously, and our definition requires verifying properties of these maliciously generated states.

In \cite{PFMO25}, one receives $\rho_1,\dots,\rho_N$, and the goal is to test whether
$$
\rho_{\mathrm{avg}} = \frac{1}{N}\sum_{i \in [N]} \rho_i
$$
is close to a known target quantum state $\sigma$.
Although their definition disallows correlations between distinct states $\rho_i$ and $\rho_j$, our definition permits such correlations.
Conversely, we assume that the unknown states are promised to be efficiently generatable, whereas their setting makes no such promise.

\cite{Fawzi_2024} considers learning from general non-i.i.d. quantum states. In their frame work, a random permutation is first applied to the registers, after which $N-1$ subsystems are used for training and the remaining subsystem is reserved for testing.
The learner is required to output a prediction of the conditional reduced state of the test register induced by its action on the training registers.
In contrast, our framework is tailored to verifying properties of the induced marginal of \(\rho_{\cA_1,\ldots,\cA_N}\), and it is not clear how to adapt their techniques to handle our notion of adversarially correlated generation and verification.

\if0
\paragraph{Meta-Complexity Characterization of OWFs.}
A recent line of research~\cite{FOCS:LiuPas20,C:LiuPas21,STOC:IlaRenSan22,FOCS:HirNan23} shows that the existence of OWFs is equivalent to several average-case hardness notions of Kolmogorov complexity, i.e., the average-case hardness of computing the minimum description length of a given string.  
%Our results rely on these characterizations.
In this paper, we have shown the equivalence between the feasibility of verifying PPT-samplable distributions and
the existence of OWFs. Therefore, by combining it with \cite{FOCS:LiuPas20,C:LiuPas21,STOC:IlaRenSan22,FOCS:HirNan23},
we also obtain the equivalence between the feasibility of verifying PPT-samplable distributions 
and these classical average-case hardness notions of Kolmogorov complexity.
\fi

\paragraph{Kolmogorov-Complexity.}

In the meta-complexity literature~\cite{FOCS:LiuPas20,STOC:Hir21,C:LiuPas21,CCC:GKLO22,STOC:IlaRenSan22,FOCS:HirNan23}, several variants of time-bounded Kolmogorov complexity have been studied.
We introduce a quantum extension of one such variant.
We briefly review the relevant notions below.
Recent work~\cite{STOC:ZOS21} introduced a randomized variant $rK^t(x)$, defined as the minimum description length of a randomized program that outputs $x$ within $t$ steps with high probability.
A probabilistic variant $pK^t(x)$ has also been considered~\cite{CCC:GKLO22}, where a short deterministic program is required to output $x$ correctly for most public random strings.
Our work builds on the notion of $uK^t(x)$, defined as the negative logarithm of the time-bounded universal probability distribution~\cite{SOLOMONOFF19641,FOCS:ImpLev90,FOCS:HirNan23}.%
\footnote{They denote this quantity by $q^t(x)$. To avoid possible confusion with “quantum,” we instead write $uK^t(x)$, where $u$ stands for “universal.”}
In the classical setting, these notions are essentially equivalent: under standard derandomization assumptions, they coincide up to logarithmic factors~\cite{CCC:GKLO22,ZO22,FOCS:HirNan23}.

In contrast, our work focuses exclusively on $quK^t(x)$, a quantum analogue of $uK^t(x)$.
One might ask whether our results could instead be derived using quantum analogues of $rK^t(x)$ or $pK^t(x)$.
We expect that this is unlikely.
Unlike the classical case, the quantum counterparts of $uK^t(x)$, $rK^t(x)$, and $pK^t(x)$ are not necessarily equivalent.
Consequently, it remains unclear how to obtain our results using quantum variants of $rK^t(x)$ or $pK^t(x)$.
% Based on these observations, we establish our results by introducing quantity $quK^t(x)$, which is found to be an appropriate notion for our purpose.

\subsection{Open Problems}
Our research raises several important open questions that remain to be explored.
Below, we highlight a particularly interesting one.

In \cref{inf:q_verify}, we have shown that, assuming the existence of QEFID, there exists a QPT-samplable distribution that is hard to verify.
This result does not exclude the possibility that certain QPT-samplable distributions with high entropy might still be efficiently verifiable.
Therefore, a natural question arises: is every QPT-samplable distribution with high entropy hard to verify?

More concretely, we ask:
\begin{center}
\emph{For any QPT-samplable distribution $\cQ$ with sufficiently large entropy, does there exist a PPT-samplable distribution $\cC$ such that no PPT algorithm can distinguish whether samples came from $\cQ$ or $\cC$?}
\end{center}
An affirmative answer would indicate that no PPT algorithm can detect sampling-based quantum advantage when the entropy of the quantum distribution is sufficiently high.
% \mor{why? there could be a QPT-sampleable low-entropy distribution that cannot be PPT sampled and can be distinguished with PPT distinguisher.}
% \taiga{Added``when the entropy of the quantum distribution is sufficiently high''.}
Conversely, a negative answer would indicate that some form of quantum sampling-based advantage can, in fact, be detected by classical verification.

In the classical case, we already establish an affirmative answer in \Cref{inf:hard_owf}.
Specifically, we show that for any PPT-samplable distribution $\cQ$ with sufficiently large entropy, there exists another PPT-samplable distribution $\cC$ that is statistically far from $\cQ$ yet indistinguishable by any PPT algorithm.
However, our proof technique there crucially relies on the fact that the output of a PPT algorithm can be described by its internal randomness, making it unclear how to extend the argument to the quantum setting.

\subsection{Paper organization}
In \Cref{sec:prelims}, we review standard results and definitions;
we introduce our definitions and basic results about quantum
meta-complexity in \Cref{sec:qmeta}.
In \Cref{sec:def-ver}, we formally introduce our definitions of distribution
verification (\Cref{def:selective,def:adaptive}) and state
verification~(\Cref{def:verify_state}) and prove simple and basic results 
about these definitions.
Next, we prove our main results about classical distribution verification in
\Cref{sec:ver-with-K},
and state verification in \Cref{sec:state-ver}.
We then discuss 
in
\Cref{sec:ver-without-K}
how to achieve quantum distribution verification without
Kolmogorov complexity, employing the nonexistence of one-way puzzles,
and also show the cryptographic hardness of quantum distribution verification.
In \Cref{sec:certified-randomness} we obtain our application to certified
randomness, and in \Cref{sec:universal-ver} we show how to construct a
universal verifier for quantum advantage using 
Kolmogorov complexity ($qcd^t$)
as a benchmark.

\section{Preliminaries}
\label{sec:prelims}

\subsection{Notations}
Here we introduce basic notations and mathematical tools used in this paper.
We denote by $x \chosen X$ an element from a finite set $X$ chosen uniformly at random,
and $y \gets \algo{A}(x)$ denotes assigning to $y$ the output of a
probabilistic or deterministic algorithm $\algo{A}$ on an input $x$. When
we want to explicitly denote that $\algo{A}$ uses randomness $r$, we write $y \gets
\algo{A}(x;r)$. When $D$ is a distribution, $x \chosen D$ denotes sampling
an element from $D$. The notation $\{y_i\}_{i\in[N]}\la\cA(x)^{\otimes N}$
means that $\cA$ is run on input $x$ independently $N$ times, and $y_i$ is
the $i$th result. Let $[\ell]$ denote the set of integers $\{1, \cdots,
\ell \}$, and $y \seteq z$ denote that $y$ is set, defined, or substituted
by $z$. For a string $s \in \zo{\ell}$, $s[i]$ and $s_i$ denotes $i$-th bit
of $s$.
QPT stands for quantum polynomial time. 
PPT stands for (classical) probabilistic polynomial time.
A function $f: \N \ra \R$ is a negligible function if for any constant $c$, there exists $n_0 \in \N$ such that for any $n>n_0$, $f(n) < n^{-c}$. We write $f(n) \leq \negl(n)$ to denote $f(n)$ being a negligible function.
The statistical distance between two distributions $\cD$ and $\cC$ is given by $\Delta(\cD,\cC)\seteq \frac{1}{2}\sum_{x}\abs{\Pr[x\la\cD]-\Pr[x\la\cC]}$ or $\mathsf{SD}(\cD,\cC)\seteq \frac{1}{2}\sum_{x}\abs{\Pr[x\la\cD]-\Pr[x\la\cC]}$.
The trace distance between two quantum states $\rho$ and $\sigma$ is given by $\mathsf{TD}(\rho,\sigma)=\frac{1}{2}\norm{\rho-\sigma}_1$
The Shannon entropy of a distribution $\cD$ is given by $H(\cD)\seteq \sum_{x}\Pr[x\la\cD]\log\left(\frac{1}{\Pr[x\la\cD]}\right)$.

In this work, we use the triangle inequality, and the following concentration inequalities.
\begin{lemma}[Triangle Inequality]\label{triangle}
Let $b_1,...,b_s \in\R$.
Then, we have
\begin{align*}
    \abs{\sum_{i\in[s]}b_i}\leq \sum_{i\in[s]}\abs{b_i}.
\end{align*}
\end{lemma}
\begin{lemma}[Hoeffding's inequality]\label{Hoeffding}
Let $s\in\N$, and $t\in \R$.
Let $X_1,X_2,...,X_s$ be an independent random variable over $[0,1]$.
Then, we have
\begin{align*}
\Pr_{\{b_i\la X_i\}_{i\in[s]} }\left[ \left(\frac{1}{s}\sum_{i\in[s]} b_i-\frac{1}{s}\sum_{i\in[s]}\mathbb{E}[X_i]\right) > t \right]\leq \exp(-2st^2)
\end{align*}
and
\begin{align*}
\Pr_{\{b_i\la X_i\}_{i\in[s]} }\left[ \left(\frac{1}{s}\sum_{i\in[s]}\mathbb{E}[X_i]-\frac{1}{s}\sum_{i\in[s]} b_i\right) > t \right]\leq \exp(-2st^2)
\end{align*}
\end{lemma}

\begin{lemma}[Chernoff bound]\label{Chernoff}
Let $s\in\N$.  
Let $X_1,...,X_s$ be an independent random variable over $\bit$.
Then, for any $\delta>0$, we have
\begin{align*}
\Pr_{\{b_i\la X_i\}_{i\in[s]} }\left[ \abs{\sum_{i\in[s]} b_i-\sum_{i\in[s]}\mathbb{E}[X_i]} \leq \delta\sum_{i\in[s]}\mathbb{E}[X_i] \right]\geq 2\exp\left(-\frac{\delta^2 \sum_{i\in[s(n)]}\mathbb{E}[X_i]}{3}\right).
\end{align*}
\end{lemma}

\subsection{Cryptography}

\begin{definition}[Infinitely-Often One-Way Functions (OWFs)]\label{def:OWFs}
    A function $f:\bit^*\to\bit^*$ that is computable in classical
    deterministic polynomial-time is an infinitely-often one-way function
    (OWF) if, for any PPT adversary $\cA$ and for any polynomial $p$, 
    we have
    \begin{align*}
    \Pr[f(x')=f(x): x\gets\bit^{n}, x'\gets\cA(1^{n},f(x))] \le\frac{1}{p(n)}
    \end{align*}
    for infinitely many $n\in\N$.
\end{definition}

\begin{definition}[One-Way Puzzles (OWPuzzs) \cite{STOC:KhuTom24}]
\label{def:OWPuzz}
A one-way puzzle is a pair $(\Samp, \Vrfy)$ of algorithms with the following syntax:
\begin{itemize}
    \item $\Samp(1^n) \rightarrow (\ans,\puzz)$: It is a QPT algorithm that, on input $1^n$, outputs two classical bit strings $(\ans,\puzz)$. 
    \item $\Vrfy(\ans',\puzz) \rightarrow \top/\bot$: It is an unbounded algorithm that, on input $(\ans',\puzz)$, outputs $\top/\bot$.
\end{itemize}
We require the following propperties.

\paragraph{Correctness:}
\begin{align*}
\Pr_{(\ans,\puzz)\leftarrow \Samp(1^n)} [\top\gets\Vrfy(\ans,\puzz)] \ge 1-\negl(n).
\end{align*}

\paragraph{Security:}
For any \textbf{uniform} QPT adversary $\cA$ and for any polynomial $p$,
\begin{align*}
\Pr_{(\ans,\puzz) \leftarrow \Samp(1^n)} [\top\gets\Vrfy(\cA(1^n,\puzz),\puzz)] \le\frac{1}{p(n)}
\end{align*}
for infinitely many $n\in\N$
\end{definition}

\paragraph{Variants of QEFID.}

\begin{definition}[Infinitely-Often Weak Non-Uniform QEFID]    
An infinitely-often weak non-uniform QEFID with short advice is a QPT algorithm $\Gen$, which takes $1^n$, $b\in \bit$ and $\mu\in[p(n)]$ for an polynomial $p$.
We require that there exists a $\mu^*\in[p(n)]$, a constant $c>0$ and a function $\alpha$ with $\alpha(n)\geq 2n^{-c}$ for all $n\in\N$ such that the following two conditions hold:
\begin{description}
\item[Statistically Far:] 
\begin{align*}
\mathsf{SD}(\Gen(1^n,\mu^*,0),\Gen(1^n,\mu^*,1))\geq \alpha(n)
\end{align*}
for all sufficiently large $n\in\N$.
\item[Computationally Indistinguishability with short advice:]
For any QPT algorithm $\cA$,
\begin{align*}
    &\Bigg|\Pr[1\la\cA(1^n,\mu^*,x):x\la\Gen(1^n,\mu^*,0)]-\Pr[1\la\cA(1^n,\mu^*,x):x\la \Gen(1^n,\mu^*,1)]\Bigg|\\
    &\leq \mathsf{SD}(\Gen(1^n,\mu^*,0),\Gen(1^n,\mu^*,1))-n^{-c}
\end{align*}
for infinitely many $n\in\N$.
\end{description}
\end{definition}

\begin{observation}\label{cor:non_uniform}
Suppose that infinitely-often weak non-uniform QEFIDs with short advice do not exist.
Then, for any constant $c>0$, for any QPT algorithm $\cD$, which takes $1^n$ and outputs $x\in\bit^{m(n)}$, and for any QPT algorithm $\cS$, which takes $1^n$ and $i\in[n]$, and outputs $x\in\bit^{m(n)}$, there exists a QPT algorithm $\mathsf{Dis}$ such that the following holds for all $i\in[n]$:

If 
\begin{align*}
\mathsf{SD}(\cD(1^n),\cS(1^n,i))\geq \alpha(n)    
\end{align*}
for some function $\alpha$ with $\alpha(n)>2n^{-c}$,
then we have
\begin{align*}
\abs{\Pr_{x\la\cS(1^n,i)}[1\la\mathsf{Dis}(1^n,x,i)  ] -\Pr_{x\la\cD(1^n)}[1\la\mathsf{Dis}(1^n,x,i)] }\geq \mathsf{SD}(\cD(1^n),\cS(1^n,i)) -n^{-c}
\end{align*}
for all sufficiently large $n\in\N$.
\end{observation}

\paragraph{Variants of EFI.}

\begin{definition}[Infinitely-Often Weak Non-Uniform EFI]
An infinitely-often weak non-uniform EFI with short advice is a QPT algorithm $\Gen$, which takes $1^n$, $a\in[p(n)]$, and $b\in \bit$ for a polynomial $p$, and outputs an $m(n)$-qubit mixed state.
We require that there exists a $a^*\in[p(n)]$, a constant $c>0$, and a function $\alpha$ with $\alpha(n)\geq 2n^{-c}$ for all $n\in\N$ such that the following two conditions hold:
\begin{description}
\item[Statistically Far:] 
\begin{align*}
\mathsf{TD}(\Gen(1^n,a^*,0),\Gen(1^n,a^*,1))\geq \alpha(n)
\end{align*}
for all sufficiently large $n\in\N$.
\item[Computationally Indistinguishability with short advice:]
For any QPT algorithm $\cA$,
\begin{align*}
    &\Bigg|\Pr[1\la\cA(1^n,a^*,\rho_0):\rho_0\la\Gen(1^n,a^*,0)]-\Pr[1\la\cA(1^n,a^*,\rho_1):\rho_1\la \Gen(1^n,a^*,1)]\Bigg|\\
    &\leq \mathsf{SD}(\Gen(1^n,a^*,0),\Gen(1^n,a^*,1))-n^{-c}.
\end{align*}
for infinitely many $n\in\N$.
\end{description}
\end{definition}

\begin{observation}\label{ob:EFI}
Suppose that infinitely-often weak non-uniform EFI with short advice do not exist.
Then, for any constant $c>0$, for any QPT algorithm $\cD$, which takes $1^n$ as input, and outputs $m(n)$-qubits $\rho$, and for any QPT algorithm $\cS$, which takes $1^n$ and $a\in[n]$, and outputs $m(n)$-qubits, there exists a QPT algorithm $\mathsf{Dis}$ such that the following holds for all $a\in[n]$:

If 
\begin{align*}
    \mathsf{SD}(\cD(1^n),\cS(1^n,a))\geq \alpha(n)
\end{align*}
for some function $\alpha$ with $\alpha(n)\geq 2n^{-c}$, then we have
\begin{align*}
    \abs{\Pr_{\rho(\cS)_i\la\cS(1^n,a)}[1\la\mathsf{Dis}(1^n,a,\rho(\cS)_i)]-\Pr_{\rho(\cD)\la\cD(1^n)}[1\la\mathsf{Dis}(1^n,a,\rho(\cD))] }\geq \mathsf{TD}(\cD(1^n),\cS(1^n,a))-n^{-c}
\end{align*}
for all sufficiently large $n\in\N$.

In particular, this can be achieved by breaking the non-uniform EFI defined by $\Gen(1^n,a^*,0)=\cD(1^n)$ and $\Gen(1^n,a^*,1)=\mathcal{S}(1^n,a^*)$.
\end{observation}

\subsection{Classical Meta-Complexity}\label{sec:meta-for-classic}
Throughout this work, we consider a fixed universal Turing machine $\mathcal{U}$.
Let $\mathcal{U}^t$ refer to the execution of $\mathcal{U}$ for $t$ steps. 
When $r$ is a bit string sampled uniformly at random,
we write $\cU^t(y;r)$ to mean that we run a universal Turing machine $\cU$
on input $y$ and $r$, halting after $t$ steps.
We will call the
output distribution of $\cU^t(y;r)$ the \emph{universal (classical) time
bounded distribution}.

\begin{definition}[Universal time bounded complexity \cite{LiVitanyi93,FOCS:HirNan23}]
     For any bit strings $x,y\in\bit^*$ and for any $t\in\N$, we define
    \begin{align*}uK^t(x|y) \coloneqq -\log_2\Pr_{r \gets \{0,1\}^t}\left[x\la \mathcal{U}^t(y;r)\right].
    \end{align*}
\end{definition}
\begin{remark}
\cite{FOCS:HirNan23} uses $q^t(x)$ to mean $uK^t(x)$.
To avoid the possibility for confusion that the $q$ in $q^t$
\cite{FOCS:HirNan23} might stand for ``quantum'' we relabel the notion as
$uK^t$ with the $u$ standing for ``universal''. 
\end{remark}

% \matthew{As written the theorem below does not hold. The constant $c$ needs to depend on $\cD$. For it to be universal we would need to have a constant length "tag" denoting that the Turing machine encoding has ended. But then that tag cannot be used except for this purpose which means that the overall expressiveness of the language is slightly diminished increasing the K complexity of all $x$. This is why prefix free and non-prefix free K complexities are a bit different. I don't think this should be too much of a problem though. I'm happy to correct the theorem and write a corrected proof.}
%\taiga{I changed $\abs{\cD}$ to $2\abs{\cD}$. }\taiga{I also take the $t$ of $uK^t$ larger than the running time of $\cD$ just in case.}\matthew{Thank you, these two changes make this work. I would prefer to replace $K$ with $\Gamma$ or $F$ to avoid overloading $K$, but it's fine.}\matthew{Should we still be including a citation to ~\cite{FOCS:HirNan23}? }\taiga{I think so. The proof is somewhat similar to \cite{FOCS:HirNan23}, and they give the statement for the first time.}
The following \cref{lem:coding} is a small modification of a theorem shown in \cite{FOCS:HirNan23}, 
where we use $uK^t(x|1^n)$ instead of $uK^t(x)$. For completeness, we provide a proof.
\begin{theorem}[Coding Theorem for $uK^t(x|1^n)$~\cite{FOCS:HirNan23}]\label{lem:coding}
    There exist a universal constant $c$ and a polynomial $t_0$ such that, for any $t$ and for every $t(n)$-time probabilistic algorithm $\cD$, which takes $1^n$ as input and outputs $x\in\bit^{m(n)}$ for any polynomial $m$, we have
    \begin{align*}
        uK^{T}(x|1^n) \leq \log_2 \frac{1}{\Pr[x\la\cD(1^n)]} + 2|\cD| + c
    \end{align*}
    for every $T>t_0(t(n))$.
    Here, $\abs{\cD}$ is the description length of  the algorithm $\cD$.
\end{theorem}

\if0
\begin{proof}[Proof of \cref{lem:coding}]
    % Let $\cR_x\coloneqq \{r : \cD_n(r) = x\}$. $\left[\cU^t(d,1^n)\to x\right]$ is at least the probability of selecting a program implementing $\cD$, run on an $r \in \mathcal{R}_x$. For each $\cD$ there exists some constant $c$ and prefix of that constant length which will specify that $\cU$ should interpret the program 

    Now, consider the probability that for some $r\in \mathcal{R}_x$, the uniform distribution over Turing machines samples the Turing machine $\cG(\cdot,r)$. This is at least $O\left(\frac{p_x}{2^{|\cD|}}\right)$ since it is at least the probability that $\cD$ is chosen at random, then some randomness $r$ in $\mathcal{R}_x$, as well as the constant length description of the rest of the program (that is, the code stating output the first input and the finishing symbol denoting that the Turing machine encoding has ended).

    Then the definition of $uK^t$ gives us
    $$uK^t(x | 1^n) \leq -\log_2 \left(O\left(\frac{p_x}{2^{|\cD|}}\right)\right) = |\cD| +\log_2 \frac{1}{p_x} + c$$.
\end{proof}
\fi

\begin{proof}[Proof of \cref{lem:coding}]
Let $\cM\in\bit^{\abs{\cD}}$ be an encoding of $\cD$.
Let $K\in\bit^{\abs{K}}$ be an encoding of algorithm that represents that ``the first $\abs{\cD}$ bits are regarded as the description of the program, and the remaining bits are regarded as the program’s $t(n)$ bits of randomness followed by $1^n$''.
There exists a constant $c$ such that $\abs{K}\leq \abs{\cD}+c$.
Conditioned on that the first $\abs{\cD}+\abs{K}$ bits of $r$ is equivalent to $K\|\cM$, $U^{T}(r,1^n)$ behaves in the same way as $\cD(1^n)$.
This event occurs with probability at least $2^{-\abs{K}+\abs{\cD}}$.
Hence, we have
\begin{align*}
\Pr_{r\la\bit^T}[x\la\cU^{T}(r,1^n)]
&\geq 2^{-\abs{K}+\abs{\cD}} \Pr[x\la\cD(1^n)]\\
&\geq 2^{-2\abs{\cD}+c}\Pr[x\la\cD(1^n)].
\end{align*}
This implies that
\begin{equation*}
    uK^{T}(x|1^n)\leq \log_2\frac{1}{\Pr[x\la\cD(1^n) ]}+2\abs{\cD}+c.
    \qedhere
\end{equation*}
\end{proof}

\if0
We also use the following \cref{thm:kraft}, which roughly states that $uK^t(x|1^n)$ is larger than $\log_2\left(\frac{1}{\Pr[x\la\cD_n]}\right)-\alpha$ with high probability over the distribution $x\la\cD_n$.
Note that if we consider the prefix Kolmogorov complexity, a similar result is known to hold~\cite{LiVitanyi93}.\mor{cite a reference} In the \cref{thm:kraft}, instead of prefix Kolmogorov complexity, we consider $uK^t(x|1^n)$ for our purpose with a slight modification of the proof.
\fi

We also use the following \cref{thm:kraft}. 
It was shown 
in \cite{LiVitanyi93} for the prefix Kolmogorov complexity, but here 
we consider $uK^t(x|1^n)$ for our purpose, which needs some modifications of a proof. For clarity, we provide a proof.

\begin{theorem}[Incompressibility for $uK^t(x)$]\label{thm:kraft}
There exists a constant $c$ such that, 
for any $n\in\N$, any algorithm distribution $\cD$, which takes $1^n$ as input and outputs $x\in\bit^{m(n)}$ for any polynomial $m$, any $t>n$, and for any $\alpha>0$
we have
\begin{align*}
\Pr[ uK^{t}(x|1^n)\leq -\log_{2} \Pr[x\la\cD(1^n)]-\alpha :x\la\cD(1^n)]\leq (m(n)+c)\cdot 2^{-\alpha+1}.
\end{align*}
\end{theorem}
\if0
\taiga{Elaborate the remark.}
\begin{remark}
A reader familiar with Kolmogorov complexity might wonder why we do not use Kraft's inequality.
Kraft's inequality holds for prefix Kolmogorov complexity.
Therefore, if we use Kraft's inequality, we need to introduce the prefix Kolmogorov complexity which adds unnecessary complexity to our definitions and proofs.
% The definition of prefix Kolmogorov complexity is a little bit complex and we need to introduce additional non-standard notions such as self-deliminating Turing machine.
% We believe that this approach is annoying.
Therefore, we choose to show \cref{thm:kraft}.
\end{remark}
\fi

\begin{proof}[Proof of \cref{thm:kraft}]
There exists a constant $c$ such that for any $t>n$, and $x\in \bit^{m(n)}$
\begin{align*}
    uK^{t}(x|1^n)\leq m(n)+c 
\end{align*}
for all $n\in\N$.
This directly follows from \cref{lem:coding} by considering a probabilistic algorithm $\cD$, which takes $1^n$ and uniformly randomly output $x\in\bit^{m(n)}$.

For any $n\in\N$, any $x\in\bit^{m(n)}$, and any $s\in[m(n)+c]$, we define
\begin{align*}
\cS_{n,t}(s)\seteq \{x\in\bit^{m(n)}: s-1< uK^t(x|1^n)\leq s\}
\end{align*}
and
\begin{align*}
\cH_{n,t}\seteq \{x\in\bit^{m(n)}: uK^{t}(x|1^n)\leq -\log \left(\Pr[x\la\cD(1^n)] \right)-\alpha\}.
\end{align*}
 
Note that, from the definition of $\cS_{n,t}(s)$ and $\cH_{n,t}$, for all $x\in \cS_{n,t}(s)\cap \cH_{n,t}$, we have
\begin{align*}
\Pr[x\la\cD(1^n)]\leq 2^{-s-\alpha+1}.
\end{align*}
Furthermore, because
$
\Pr_{r\la\bit^t}[x\la U^t(r|1^n)]  \geq 2^{-s}
$ for all $x\in\cS_{n,t}(s)$
and 
$ 
\sum_{x\in\bit^n}\Pr_{r\la\bit^t}[x\la U^t(r)]\leq 1
$, we have
\begin{align*}
\abs{\cS_{n,t}(s)}\leq 2^{s}.
\end{align*}

Therefore, we have
\begin{align*}
&\Pr_{x \la \cD(1^n)}[ uK^{t}(x|1^n)\leq -\log (\Pr[x\la\cD(1^n)])-\alpha]\\
&= \sum_{i\in [m(n)+c]}\sum_{x\in \cS_{n,t}(i)\cap \cH_{n,t}} \Pr[x\la \cD(1^n)]\\
&\leq \sum_{i\in [m(n)+c]}\sum_{x\in \cS_{n,t}(i)\cap \cH_{n,t}} 2^{-i-\alpha+1}\\
&\leq \sum_{i\in[m(n)+c]}\abs{\cS_{n,t}(i)} 2^{-i-\alpha+1}\leq\sum_{i\in[m(n)+c]} 2^{-\alpha+1}\leq  (m(n)+c)\cdot 2^{-\alpha+1}.
\end{align*}
\end{proof}

%\matthew{Verified the above proof, and it will still work after change to thm 2.4}

We use the following \cref{lem:modaaronson} for showing the security of distribution verification.
\cref{lem:modaaronson} is based on Lemma 19 of \cite{Aaronson14},
adapted to use $uK^t(x)$ instead of prefix-free Kolmogorov complexity.
For completeness, we give a proof of \cref{lem:modaaronson} in \cref{appendix}.
\begin{theorem}[\cite{Aaronson14} Adapted to $uK^t$]\label{lem:modaaronson}
Let $\epsilon$ be any function and let $s$, $m$, $t$ be any polynomials, and let $t_0$ be a polynomial given in \cref{lem:coding}.
    Let $\cD$ be any algorithm that takes $1^n$ as input and outputs $x\in\bit^{m(n)}$.
    Let $\cG$ be a classical algorithm that takes $1^n$ and outputs 
    $\left(x_1,\dots,x_{s(n)}\right)\in\bit^{m(n)\times s(n)}$ 
    running in time $t(n)$.
    Define $\mathsf{Marginal}_{\cG}(1^n)$ to be the distribution defined as follows:
    \begin{enumerate}
        \item Sample $\left(y_1,\dots,y_{s(n)}\right)\la\cG(1^n)$.
        \item Sample $i\gets [s(n)]$ uniformly at random.
        \item Output $y_i$
    \end{enumerate}
    Define $p_y \coloneqq \Pr[y\la\cD(1^n)]$.
    Define $A^{\cD}_{n,s,\alpha}$ to be the set
    \begin{align*}
    A^{\cD}_{n,s,\alpha}\coloneqq \left\{\left(y_1,\dots,y_{s(n)}\right) : \log_2\frac{1}{p_{y_1}\dots p_{y_{s(n)}}} \leq uK^{t_0(t(n))}\left(y_1,\dots,y_{s(n)}|1^n\right) + \alpha(n)\right\}.
    \end{align*}
    There exists a constant $C$ such that for all sufficiently large $n\in\N$, if
    \begin{align*}\Pr_{y_1,\dots,y_{s(n)}\la\cG(1^n)}\left[\left(y_1,\dots,y_{s(n)}\right)\in A^{\cD}_{n,s,\alpha}\right] \geq 1-\epsilon(n),
    \end{align*}
    then
    we have
    \begin{align*}
        \SD\left(\cD(1^n), \mathsf{Marginal}_{\cG}(1^n)\right) \leq \epsilon(n) + \sqrt{\frac{\log_2\left(\frac{1}{1-\epsilon(n)}\right)+\alpha(n)+C}{s(n)}}.
    \end{align*}
\end{theorem}

\section{Quantum Meta-Complexity}
\label{sec:qmeta}

In this section, we introduce $quK^{t}(x)$, which is a quantum analog of $uK^{t}(x)$, and show several properties of $quK^{t}(x)$.

% \matthew{Universal quantum circuit? This should only be possible for fixed run-times.}
% \taiga{I changed $\mathsf{QU}$ to $\mathsf{QU}^t$ and add explanation, but we might have a better explanation.}
% \matthew{Actually now I'm less sure whether this was previously a problem. I guess it depends on how $c$ is encoded. If it's encoded in a straightforward way i.e. the number of gates and wires are both $O(m(n))$ then this should be fine.}
% \taiga{I understood, but, just in case, I stick to $\mathsf{QU}^t$.}
%Let $\mathsf{QU}^t$ be a universal quantum circuit that takes $1^n$ and $c\in\bit^{m(n)}$ as input, where $m$ is a polynomial. Then, it considers $c$ as an encoding of some $t(n)$-depth quantum circuit $C$, and outputs the output of $C$.
\subsection{Definition of $quK^t(x)$}
Let $\mathsf{QU}^t$ be a quantum algorithm that takes $1^n$ and $c\in\bit^{*}$ as input.
The algorithm considers $c$ as an encoding of an output length $m$, an input size $s$, a $t(n)$-depth quantum circuit $C$.
Then, it runs $C\ket{0^s}$, measures the first $m$-bits, and outputs the resulting $m$ bits. 
If $c$ is not a valid encoding of a quantum circuit, then we simply output $\bot$.

Let $\cA$ be a quantum algorithm that takes $1^n$ as input and outputs $x\in\bit^{m(n)}$, where $m$ is a polynomial.
We say that $\cA$ is a $t$-time quantum algorithm if it can be run by first running a classical Turing machine 
$c\la \cM^t(1^n)$ 
running in $t$-steps, then running $x\la\mathsf{QU}^t(1^n,c)$,
and outputting $x$.

We define $quK^t(x|1^n)$ as follows.
\begin{definition}[$quK^t(x|1^n)$]\label{def:quK}
For any $n,t\in\N$, let us define a quantum algorithm $\cQ^{t}(1^n)$ as follows:
\begin{description}
\item[$\cQ^{t}(1^n)$:]$ $
\begin{enumerate}
\item Sample $\Pi\la\bit^{t}$.
\item Run $c\la\cU^{t}(\Pi,1^n)$.
\item Run $x\la\mathsf{QU}^t(1^n,c)$.
\item Output $x$.
\end{enumerate}
\end{description}
For any $t\in\N$, $x\in\bit^*$ and any $n\in\N$,
let us define $quK^{t}(x|1^n)$ as follows:
\begin{align*}
quK^{t}(x|1^n)\seteq -\log_2 \Pr[x\la\cQ^t(1^n)].
\end{align*}
\end{definition}

% \matthew{Two questions about $quK^t$. 1) Do we want to acknowledge that this is essentially a time bounded version of Gács' notion? 2) given that we are introducing this here, do we need to / want to give an invariance theorem? If so we'd need to argue approximate invariance both with respect to time (which should hit us with a $t\log t$ in the invariance [LV 19 thm 7.1.1]), and an invariance over our choice of gates (which I don't know what we'd get here, probably want to check Bernstein Vazerani and how time efficient their $\epsilon$ approximation results are).}
% \taiga{1)I want to acknowledge Gacs notion, but I do not carefully follow the literature, so I appreciate if you could add the remark after the definition. 
% 2) I do not want to include invariance theorem because we do not need to use0 it.}
\begin{remark}
This notion is both a quantum generalization of $uK^t$~\cite{LiVitanyi93,FOCS:HirNan23} and a restriction of the notion of quantum Kolmogorov complexity introduced by Gács \cite{Gac01} to time bounded algorithms outputting classical states. While Gács' notion is defined slightly differently, the notions are equivalent and the invariance of Gács' notion applies to ours. 
\end{remark}

\subsection{Properties of $quK^t(x)$}

To prove our main results, we need quantum analogues of many of the theorems and lemmas in \cref{sec:meta-for-classic}. The next three theorems are specifically quantum analogues of \cref{lem:coding}, \cref{thm:kraft}, and \cref{lem:modaaronson}. Conceptually, these analogues are obtained by a straightforward substitution—namely, instead of considering $uK^t(x|1^n)$ and $U^t(x|1^n)$, we consider $quK^t(x|1^n)$ and $Q^t(x|1^n)$, respectively. In all three cases, the proofs go identically to their classical counterparts, and hence we omit them.

\if0
To prove our main results, we need quantum analogues of many of the theorems and lemmas in \cref{sec:meta-for-classic}. The next three theorems are specifically quantum analogues of \cref{lem:coding}, \cref{thm:kraft}, and \cref{lem:modaaronson}. In all three cases their proofs go identically to their classical analogues so we choose to omit them.
\fi

\begin{theorem}[Coding Theorem for $quK^t(x)$ - quantum version of~\Cref{lem:coding}]\label{lem:q_coding}
There exist a universal constant $c$ and polynomial $t_0$ such that, for every $t(n)$-time quantum algorithm $\cD$, which takes $1^n$ as input and outputs $x\in\bit^{m(n)}$ for any polynomial $m$, we have
\begin{align*}
    quK^{T}(x|1^n) \leq \log_2 \frac{1}{\Pr[x\la\cD(1^n)]} + 2|\cD| + c
\end{align*}
for every $T>t_0(t(n))$.
Here, $\abs{\cD}$ is the description length of  the algorithm $\cD$.
\end{theorem}

\if0
\begin{proof}[Proof of \cref{lem:q_coding}]
Let us define the sets of $\cS$

With probability $\frac{1}{n}$ over $\ell\la[n]$, we have $\ell=\abs{\cD}$.
Furthermore, with probability $2^{-\abs{\cD}}$ over $\Pi\la \bit^{\ell}$, we have $\cD=\Pi$.
Hence, we have
\begin{align*}
\Pr[\cQ^{t(n)}(1^n)\ra x]&\geq \frac{1}{n} 2^{-\abs{\cD}}\cdot \Pr[\cD(1^n)\ra x]\\
quK^{t(n)}(x|1^n)&\leq \log_2\frac{1}{\Pr[\cD(1^n)\ra x]}+\abs{\cD}+\log(n).
\end{align*}
\end{proof}
\fi

%We will use \cref{thm:q_kraft}, which is a quantum analog of \cref{thm:kraft}.
%The proof goes almost the same way as that of \cref{thm:kraft}.
\begin{theorem}[Incompressibility for $quK^{t}(x)$ = quantum version of~\Cref{thm:kraft}]\label{thm:q_kraft}
There exists a constant $c$ such that, for any $n\in\N$, any algorithm $\cD$, which takes $1^n$ as input and outputs $x\in\bit^{m(n)}$ for any polynomial $m$, any $t>n$, and for any $\alpha>0$, we have
\begin{align*}
\Pr[ quK^{t}(x|1^n)\leq -\log_2 \Pr[x\la\cD(1^n)]-\alpha :x\la\cD(1^n)]\leq (m(n)+c)\cdot 2^{-\alpha+1}.
\end{align*}
\end{theorem}

%We will use \cref{lem:q_modaaronson}, which is a quantum analog of \cref{lem:modaaronson}.
%The difference is that we consider a QPT algorithm and $quK^{t}(x)$ instead of PPT algorithm and $uK^t(x)$.
%Because the proof is the same as that of \cref{lem:modaaronson},
%we omit it.

\begin{theorem}[Quantum version of~\Cref{lem:modaaronson}]\label{lem:q_modaaronson}
    Let $\epsilon$ be any function and let $s$, $m$ and $t$ be any polynomials.
    Let $\cD$ be any algorithm that takes $1^n$ as input, and outputs $x\in\bit^{m(n)}$.
    Let $\cG$ be a $t(n)$-time quantum algorithm that takes $1^n$ and outputs $x\in\bit^{s(n)\cdot m(n)}$.
    Define $\mathsf{Marginal}_{\cG}(1^n)$ to be the distribution defined as follows:
    \begin{enumerate}
        \item Sample $(y_1,\dots,y_s)\la\cG(1^n)$.
        \item Sample $i\gets [s]$ uniformly at random.
        \item Output $y_i$.
    \end{enumerate}
    Define $p_y \coloneqq \Pr[y\la\cD(1^n)]$.
    Define $A^{\cD}_{n,s,\alpha}$ to be the set
    \begin{align*}
    A^{\cD}_{n,s,\alpha}\coloneqq \left\{(y_1,\dots,y_{s(n)}) : \log_2\frac{1}{p_{y_1}\dots p_{y_{s(n)}}} \leq quK^{t(n)}(y_1,\dots,y_{s(n)}|1^n) + \alpha(n)\right\}
    \end{align*}
    There exists a constant $C$ such that for all sufficiently large $n\in\N$, if
    \begin{align*}
    \Pr_{y_1,\dots,y_{s(n)}\la\cG(1^n)}\left[(y_1,\dots,y_{s(n)})\in A^{\cD}_{n,s,\alpha}\right] \geq 1-\epsilon(n),
    \end{align*}
    then
    we have
    \begin{align*}
    \SD(\cD(1^n), \mathsf{Marginal}_{\cG}(1^n)) \leq \epsilon(n) + \sqrt{\frac{\log_2\frac{1}{1-\epsilon(n)}+\alpha(n)+C}{s(n)}}.
    \end{align*}
\end{theorem}

\section{Definition and Fundamental Results on Distribution and State Verification}
\label{sec:def-ver}

\subsection{Distribution Verification}
%\matthew{This mildly diverges from the definitions above. In particular we had $\cD$ be the distribution and $\cM$ be the algorithm.}
%\mor{I think $\cD$ to be an algorithm is more precise, but in the introduction for simplicity, we set $\cD$ to be distribution.}
In this section,
we introduce two definitions of verification of distributions both for classical and quantum cases.
%The key difference between our definitions and those previously considered \cite{}\matthew{Add citation to the paper Tomoyuki mentioned} is that our definitions consider only efficiently samplable distributions as the adversarial distributions which we must identify if too statistically far. The first definition provides a weaker security guarantee against adaptive adversaries. While the second definition provides a stronger security guarantee against adversaries which must sample their outputs identically and independently.

\begin{definition}[Distribution Verification with Selective Soundness]\label{def:selective}
Let $\cD$ be an algorithm that takes $1^n$ as input, and outputs $x\in\bit^{m(n)}$ where $m$ is an arbitrary polynomial.
Let us denote $X$ to mean that $X\in \{\mbox{PPT, QPT,  determinisitic polynomial-time algorithm querying to $\mathsf{PP}$ oracle}\}$.

We say that the algorithm $\cD$ is selectively-verifiable with an $X$-algorithm if for any polynomial $t$, any function $\epsilon:\N\ra (0,1)$, and any constant $c>0$, there exist a polynomial $s$ and an $X$-algorithm $\Vrfy$ such that the following holds.
\begin{description}
    \item[Correctness:] We have 
    \begin{align*}
        \Pr[\top\la\Vrfy(1^n,x_1,\dots,x_{s}): (x_1,\dots,x_{s})\la\cD(1^n)^{\otimes s}] \geq 1-n^{-c},
    \end{align*}
    for all sufficiently large $n\in\N$. Here, $s = s(n^c,t(n),1/\epsilon(n))$.
    \item[Selective-Soundness:] For any $t(n)$-time uniform quantum adversary $\cA$, which takes $1^n$ as input and outputs $x\in\bit^{m(n)}$, such that
    \begin{align*}
        \Delta(\cA(1^n),\cD(1^n))\geq \epsilon(n),
    \end{align*}
    we have
    \begin{align*}
        \Pr[\top\la\Vrfy(1^n,x_1,\dots,x_{s}):(x_1,\dots,x_{s})\la\cA(1^n)^{\otimes s}] \leq n^{-c},
    \end{align*}
     for all sufficiently large $n\in\N$. Here, $s=s(n^c,t(n),1/\epsilon(n))$.
\end{description}
\end{definition}

\begin{definition}[Distribution Verification with Adaptive Soundness]\label{def:adaptive}
Let $\cD$ be an algorithm that takes $1^n$ as input, and outputs $x\in\bit^{m(n)}$, where $m$ is an arbitrary polynomial.
Let us denote $X$ to mean that $X\in \{\mbox{PPT, QPT,  determinisitic polynomial-time algorithm querying to $\mathsf{PP}$ oracle}\}$.

We say that the algorithm $\cD$ is adaptively-verifiable with an $X$-algorithm if, for any polynomial $t$, any function $\epsilon:\N\ra(0,1)$, and any constant $c>0$, there exist a polynomial $s$ and an $X$-algorithm $\Vrfy$ such that the following hold.
\begin{description}
    \item[Correctness:] We have
    \begin{align*}
    \Pr[\top\la\Vrfy(1^n,x_1,\dots,x_{s}):x_1,\dots,x_{s}\la\cD(1^n)^{\otimes s}] \geq 1-n^{-c}
    \end{align*}
    for all sufficiently large $n\in\N$.
    Here, $s = s(n^c,t(n),1/\epsilon(n))$.
    \item[Adaptive-Soundness:] 
    For any algorithm $\cA$ that takes $1^n$ as input and outputs $(x_1,\dots,x_{s(n)})\in$ \\$\left(\bit^{m(n)}\right)^{s(n)} $, 
    %\matthew{I feel like both here and in section 7 this should be $(x_1,\dots,x_{s(n)})\in\bit^{m(n)\times s(n)} $}\mor{you can modify in that way.} 
    we define an algorithm $\mathsf{Marginal}_{\cA}(1^n)$ as follows:
    \begin{description}
    \item[$\mathsf{Marginal}_{\cA}(1^n)$:] $ $ 
    \begin{enumerate}
        \item Sample $(x_1,\dots,x_s)\la\cA(1^n)$.
        \item Sample $i\gets [s]$.
        \item Output $x_i$.
    \end{enumerate}
    \end{description}
    For any $t(n)$-time uniform quantum adversary $\cA$ with $\Delta(\mathsf{Marginal}_{\cA}(1^n),\cD(1^n))\geq \epsilon(n)$ for all $n\in\N$, we have
    \begin{align*}
    \Pr[\top\gets\Ver(1^n,x_1,\dots,x_s) :(x_1,...,x_s)\gets\cA(1^n)] \leq (1-\epsilon(n)) + n^{-c}
    \end{align*}
    for all sufficiently large $n\in\N$.
    Here, $s = s(n^c,t(n),1/\epsilon(n))$.
\end{description}
\end{definition}

\begin{remark}
Note that in the above two definitions, Definitions~\ref{def:selective} and \ref{def:adaptive}, the adversary $\cA$ is quantum even if $X$ or $\cD$ is classical. 
When we want to consider uniform classical probabilistic $\cA$, 
we explicitly say that $\cD$ is selectively/adaptively-verifiable with an $X$-algorithm with \emph{classical-security}.
\end{remark}

% \matthew{I don't think we should be able to do parallel repetition as hoped for in the not yet written security proof below since there is no guarantee that $\cA$ is sampling things in an i.i.d. fashion given the definitions above.}

We can show that adaptive-verifiability implies selective-verifiability as follows, and hence we focus on adaptive-verifiability throughout this work.
\begin{lemma}
    If an algorithm $\cD$ is adaptively-verifiable, then it is selectively-verifiable.
\end{lemma}
\if0
\taiga{!!Under preparation!! From Here,}
\begin{proof}
    Given a weak verifier $\Ver$ we construct a strong verifier $\Ver'_{n,\ell,t}(x_1,...,x_{s'})$ where $s' = sn^{c+1}$ which runs $\Ver$ on each of the $n^{c+1}$ sets of $x_{is+1},...,x_{(i+1)s}$. $\Ver'$ outputs $\top$ if all of the $\Ver$'s output $\top$ and outputs $\bot$ if any $\Ver$ outputs $\bot$.

    By union bound:
\begin{align*}
    \BigPr{\Ver_{n,\ell,t}(x_1,\dots,x_{sn^{c+1}})\to \top}{\cD_n^{\otimes s}\to x_1,\dots,x_{sn^{c+1}}} &\geq 1-(\negl(n)\cdot n^{c+1})\\
     &\geq 1-\negl(n),
\end{align*}

    so correctness is preserved.

    Each of these sets ($x_{is+1},\dots,x_{(i+1)s}$) is sampled $i.i.d.$ from $\cG^{\otimes s}$ so for each $i$ 

    $$\forall i:\BigPr{\Ver_{n,\ell,t}(x_{is+1},\dots,x_{(i+1)s})\to \bot}{\cG^{\otimes s} \to x_1,\dots,x_{s}} \geq 1-n^{-c}$$

    and

    $$\BigPr{\exists i: \Ver_{n,\ell,t}(x_{is+1},\dots,x_{(i+1)s})\to \bot}{\cG^{\otimes s} \to x_1,\dots,x_{s}} \geq 1-n^{-c}$$

    $$\BigPr{\Ver'_{n,\ell,t}(x_{1},\dots,x_{sn^{c+1}})\to \top}{\cG^{\otimes sn^{c+1}} \to x_1,\dots,x_{sn^{c+1}}} \leq n^{-c}$$
    \eli{Sketch: parallel repetition. By i.i.d.-ness this is easy}
\end{proof}
\taiga{To here.}
\matthew{I'm sure you agree with me, but this proof isn't done yet / doesn't yet prove the desired conclusion.}
\fi

\begin{proof}
In the following, we show that $\cD$ is selectively-verifiable with a QPT algorithm if $\cD$ is adaptively-verifiable with a QPT algorithm.
In a similar way, we can show that $\mathcal{D}$ is selectively-verifiable with an $X$ algorithm if $\mathcal{D}$ is selectively-verifiable with an $X$ algorithm for
\[
X\in\{\text{PPT},\allowbreak\ \text{QPT},\allowbreak\ \text{deterministic polynomial-time algorithm querying to }\mathsf{PP}\allowbreak\ \text{oracle}\}.
\]

Let $s$ be a polynomial and $\Vrfy$ be a QPT algorithm such that 
\begin{align*}
\Pr[\top\la\Vrfy(1^n,x_1,\dots,x_{s(n)}):x_1,\dots,x_s\la\cD(1^n)^{\otimes s}]\geq 1-\frac{\epsilon(n)}{n}
\end{align*}
and
\begin{align*}
\Pr[\top\la\Vrfy(1^n,x_1,\dots,x_{s(n)}):x_1,\dots,x_s\la\cB(1^n)]\leq 1-\left(1-\frac{1}{n}\right)\epsilon(n)
\end{align*}
for any QPT algorithm $\cB$ such that $\Delta(\mathsf{Marginal}_{\cB}(1^n),\cD(1^n))\geq \epsilon(n)$ for all sufficiently large $n\in\N$.

Let $s^*(n)\seteq s(n)\cdot\frac{4n^2}{\epsilon(n)^2}$.
We consider the following QPT algorithm $\Vrfy^*$.
\begin{description}
\item[$\Vrfy^*$:]$ $
\begin{enumerate}
\item Take $x_1,\dots,x_{s^*(n)}$ as input.
\item For all $i\in[\frac{4n^2}{\epsilon(n)^2}]$, run $b_i\la\Vrfy(x_{is+1},\dots,x_{(i+1)s})$.
Let $\mathsf{Count}_{x_1,\dots,x_{s^*(n)}}$ be the number of times such that $b_i=\top$.
\item Output $\top$ if $\frac{\epsilon(n)^2}{4n^2}\mathsf{Count}_{x_1,\dots,x_{s^*(n)}}\geq \left(1-\frac{\epsilon(n)}{2}\right)$.
Otherwise, output $\bot$.
\end{enumerate}
\end{description}

\paragraph{Correctness:}
From Hoeffding's inequality, 
\begin{align*}
\Pr_{x_1,\dots,x_{s^*(n)}\la\cD(1^n)^{\otimes s^*(n)}}\left[\frac{\epsilon(n)^2}{4n^2}\mathsf{Count}_{x_1,\dots,x_{s^*(n)}}-\left(1-\frac{\epsilon(n)}{n}\right)\leq -\frac{\epsilon(n)}{4}\right]\leq \exp\left(-\frac{n^2}{2}\right).
\end{align*}
This implies that
\begin{align*}
\Pr_{x_1,\dots,x_{s^*(n)}\la\cD(1^n)^{\otimes s^*(n)}}\left[\frac{\epsilon(n)^2}{4n^2}\mathsf{Count}_{x_1,\dots,x_{s^*(n)}} \leq \left(1-\frac{\epsilon(n)}{2}\right)\right]\leq \exp\left(-\frac{n^2}{2}\right).
\end{align*}

\paragraph{Selective-Soundness:}
Let $\cA$ be an algorithm which takes $1^n$ as input and outputs $x\in\bit^{m(n)}$ such that 
\begin{align*}
\Delta(\cA(1^n),\cD(1^n))\geq \epsilon(n).
\end{align*}
The definition of $\Vrfy$ guarantees that 
\begin{align*}
\Pr[\top\la\Vrfy(1^n,x_1,\dots,x_{s(n)}):x_1,\dots,x_s\la\cA(1^n)^{\otimes s(n)}]\leq 1-\left(1-\frac{1}{n}\right)\epsilon(n).
\end{align*}

From Hoeffding's inequality,
for any QPT algorithm 
such that 
$
\Delta(\cA(1^n),\cD(1^n))\geq \epsilon(n),
$
we have
\begin{align*}
\Pr_{x_1,\dots,x_{s^*(n)}\la\cA(1^n)^{\otimes s^*(n)}}\left[\frac{\epsilon(n)^2}{4n^2}\mathsf{Count}_{x_1,\dots,x_{s^*(n)}}-\left(1-\left(1-\frac{1}{n}\right)\epsilon(n)\right)\geq \frac{\epsilon(n)}{4}\right]\leq \exp\left(-\frac{n^2}{2}\right).
\end{align*}
This implies that
\begin{align*}
\Pr_{x_1,\dots,x_{s^*(n)}\la\cA(1^n)^{\otimes s^*(n)}}\left[\frac{\epsilon(n)^2}{4n^2}\mathsf{Count}_{x_1,\dots,x_{s^*(n)}}\geq\left(1-\frac{\epsilon(n)}{2}\right)\right]\leq \exp\left(-\frac{n^2}{2}\right)
\end{align*}
for all sufficiently large $n\in\N$.
\end{proof}

\subsection{Quantum State Verification}
In this section, we introduce definition of quantum state verification.

\begin{definition}[Quantum State Verification with Adaptive Soundness]\label{def:verify_state}
Let $\cD$ be an algorithm that takes $1^n$ as input, and outputs  $m(n)$-qubit state $\rho$, where $m$ is an arbitrary polynomial.

We say that the algorithm $\cD$ is adaptively-verifiable if for any polynomial $t$, any inverse polynomial $\epsilon$, and any constant $c$, there exists a polynomial $s$ and a QPT algorithm $\Vrfy$ such that the following holds:
\begin{description}
\item[Correctness:]
We have
\begin{align*}
\Pr\left[\top\la\Vrfy\left(\rho^{\otimes s(n)} \right): \rho^{\otimes s(n)}\la\cD(1^n)^{\otimes s(n)} \right]\geq 1-n^{-c}
\end{align*}
for all sufficiently large $n\in\N$.
\item[Adaptive-Soundness:]
For any algorithm $\cA$ which takes $1^n$ as input and outputs $m(n)\cdot s(n)$-qubit state, we define the algorithm $\mathsf{Marginal}_{\cA}(1^n)$ as follows:
\begin{enumerate}
\item Run $\rho_{\cR[1],...,\cR[s(n)]}\la\cA(1^n)$.
\item Sample $i\la[s(n)]$.
\item Output the $\cR[i]$ register.
\end{enumerate}
For any $t(n)$-time uniform quantum adversary $\cA$ such that
\begin{align*}
\mathsf{TD}(\mathsf{Marginal}_{\cA}(1^n), \cD(1^n))\geq \epsilon(n)
\end{align*}
for all $n\in\N$,
we have
\begin{align*}
\Pr\left[\top\la\Vrfy\left(\rho_{\cR[1],...,\cR[s(n)]} \right): \rho_{\cR[1],...,\cR[s(n)]}\la\cA(1^n) \right]\leq (1-\epsilon(n))+n^{-c}
\end{align*}
for all sufficiently large $n\in\N$.
\end{description}
\end{definition}

\subsection{Impossible Parameter Regime}

In the definitions of adaptive-soundness, we require that the probability that $\Vrfy\left(x_1,...,x_{s(n)}\right)$ outputs $\top$ is at most $(1-\epsilon(n))$ for any adversary $\cA$ such that $\Delta(\mathsf{Marginal}_{\cA}(1^n), \cD(1^n))\geq\epsilon(n)$.
One might think that this requirement is too weak, and we should construct a verifier which outputs $\top$ with negligible probability for any adversary $\cA$ with $\Delta(\mathsf{Marginal}_{\cA}(1^n), \cD(1^n))\geq\epsilon(n)$.
However, let us stress that, in general, such a parameter regime seems impossible to achieve as we suggest in \cref{thm:impossibility}.
Therefore, in our work, we focus on the parameter regime defined in Definitions \ref{def:adaptive} and \ref{def:verify_state}.

\begin{proposition}[Restatement of \cref{informal:impossibility}]\label{thm:impossibility}
Let $\cD:1^n\ra \bit^{m(n)}$ be an algorithm such that there exists another algorithm $\cD^*$ such that $\mathsf{SD}(\cD(1^n),\cD^*(1^n))\geq 1-\alpha(n)$, where $\alpha$ and $m$ are an arbitrary function.
Then, for any $\delta>0$, any function $\epsilon,\beta,s$, no algorithm $\Vrfy$ satisfies the following at the same time:
\begin{description}
\item[Correctness:]
\begin{align*}
\Pr[\top\la \Vrfy(1^n,x_1,...,x_{s(n)}): x_1,...,x_{s(n)}\la\cD(1^n)^{\otimes s(n)}]\geq (1-\beta(n)).
\end{align*}
\item[Adaptive-Soundness:]
For any algorithm $\cA$, which takes $1^n$ as input and outputs $(x_1,...x_{s(n)})\in\bit^{m(n)\cdot s(n)}$, we define the algorithm $\mathsf{Marginal}_{\cA}(1^n)$ as follows:
\begin{enumerate}
    \item Run $x_1,...,x_{s(n)}\la \cA(1^n)$.
    \item Sample $i\in[s(n)]$.
    \item Output the $x_{i}$.
\end{enumerate}
For any uniform adversary $\cA$ such that
\begin{align*}
    \mathsf{SD}(\mathsf{Marginal}_{\cA}(1^n),\cD(1^n))\geq \epsilon(n)
\end{align*}
for all $n\in\N$, we have
\begin{align*}
\Pr[\top\la \Vrfy(1^n,x_1,...,x_s):(x_1,...,x_s)\la\cA(1^n)]\leq \left(1-\frac{\epsilon(n)}{1-\alpha(n)}\right)(1-\beta(n))-\delta
\end{align*}
for some $n\in\N$.
\end{description}
\end{proposition}
\begin{remark}
Note that in \cref{thm:impossibility}, we allow that $\cA$ is any algorithm (including non-uniform algorithm) instead of time-bounded uniform algorithm. 
Therefore, \cref{thm:impossibility} does not explicitly rule out the possibility that we can construct verification algorithm, which is secure against only time-bounded uniform algorithm.
However, in the statement if we consider a uniform QPT algorithm $\cD$ such that there exists another QPT algorithm $\cD^*$ with $\mathsf{SD}(\cD(1^n),\cD^*(1^n))\geq 1-\alpha(n)$, then we can similarly claim that no $\Vrfy$ algorithm satisfy the correctness and soundness with the same parameter secure against sufficiently large time-bounded uniform algorithm.
\end{remark}

\begin{proof}[Proof of \cref{thm:impossibility}]
Let us consider the following algorithm $\cA$.
\begin{description}
    \item[$\cA(1^n)$:]$ $
    \begin{enumerate}
        \item Set $b=1$ with probability $1-\frac{\epsilon(n)}{1-\alpha(n)}$ and set $b=0$ with probability $\frac{\epsilon(n)}{1-\alpha(n)}$.
        \item If $b=1$, run $x_1,...,x_{s(n)}\la\cD(1^n)^{\otimes s(n)}$, and output $x_1,...,x_{s(n)}$.
        \item If $b=0$, run $x_1,...,x_{s(n)}\la\cD^*(1^n)^{\otimes s(n)}$, and output $x_1,...,x_{s(n)}$.
    \end{enumerate}
\end{description}

From the construction of $\cA$, we have
\begin{align*}
\mathsf{SD}(\mathsf{Marginal}_{\cA}(1^n),\cD(1^n))\geq \frac{\epsilon(n)}{1-\alpha(n)}\mathsf{SD}(\cD^*(1^n),\cD(1^n))\geq \epsilon(n).
\end{align*}
Let $\Vrfy$ be an arbitrary algorithm such that
\begin{align*}
\Pr[\top\la\Vrfy(1^n,x_1,...,x_{s(n)}):x_1,...,x_{s(n)}\la\cD(1^n)^{\otimes s(n)} ]\geq 1-\beta(n).
\end{align*}
Then, from the construction of $\cA$, we have
\begin{align*}
&\Pr[\top\la\Vrfy(1^n,x_1...,x_{s(n)}):x_1,...,x_{s(n)}\la\cA(1^n)]\\
&\geq(1-\epsilon(n))\Pr[\top\la\Vrfy(1^n,x_1...,x_{s(n)}):x_1,...,x_{s(n)}\la\cD(1^n)^{\otimes s(n)}]\geq \left(1-\frac{\epsilon(n)}{1-\alpha(n)}\right)(1-\beta(n)).
\end{align*}
This completes the proof.
\end{proof}

\subsection{Impossibility of Learning in Non-iid Setting}

Given that we study the feasibility of verifying marginal distributions, one might wonder whether we can learn a marginal distribution of any unknown algorithm.
Remark that, for useful parameter regimes, it seems impossible to learn the marginal distribution of an unknown algorithm $\cA$ even if we allow $s$ to be an arbitrary function.
This justifies studying verification instead of learning marginal distribution.
Our formal statement of the impossibility is as follows:
\begin{proposition}[Restatement of \cref{informal:impossibility_learn}]\label{thm:impossibility_learn}
For any $\frac{1}{2}>\alpha(n),\beta(n)>0$ and any function $s$, no algorithms $\mathsf{Learn}$ satisfy the following:
For any uniform algorithm $\cA$, which takes $1^n$ and outputs $x_1,...,x_{s(n)}\in\bit^{m(n)}$, we have
\begin{align*}
\Pr\left[\mathsf{SD}(h(1^n),\mathsf{Marginal}_{\cA}(1^n))\leq \alpha(n):
\begin{array}{ll}
x_1,...,x_{s(n)}\la\cA(1^n)\\
h\la\mathsf{Learn}(1^n,x_1,...,x_{s(n)})
\end{array}
\right]\leq 1-\beta(n)
\end{align*}
for some $n\in\N$.
\end{proposition}

\begin{proof}[Proof of \cref{thm:impossibility_learn}]
In the following, suppose that there exists an algorithm $\mathsf{Learn}$, and then, we show that this leads to a contradiction to \cref{thm:impossibility}.

Let $\cD$ be an algorithm such that there exists $\cD^*$ with $\mathsf{SD}(\cD^*(1^n),\cD(1^n))=1$ for all $n\in\N$.
If there exists a $\mathsf{Learn}$ and a function $s$ presented in \cref{thm:impossibility_learn}, we can construct $\Vrfy$ such that
\begin{align*}
\Pr[\top\la\Vrfy\left(x_1,...,x_{s(n)}\right):x_1,...,x_{s(n)}\la \cD(1^n)^{\otimes s(n)} ]\geq 1-\beta(n)
\end{align*}
and
\begin{align*}
\Pr[\top\la\Vrfy\left(x_1,...,x_{s(n)}\right):x_1,...,x_{s(n)}\la \cA(1^n)]\leq \beta(n)
\end{align*} 
for any $\cA$ algorithm and any $\epsilon(n)$ with $\epsilon(n)>\min\left\{2\alpha(n),\frac{1-2\beta(n)}{1-\beta(n)}\right\}$ and $\mathsf{SD}(\mathsf{Marginal}_{\cA}(1^n),\cD(1^n))\geq \epsilon(n)$.
This is a contradiction to \cref{thm:impossibility} because, with the choice of $\epsilon$, we have
\begin{align*}
\Pr[\top\la\Vrfy\left(x_1,...,x_{s(n)}\right):x_1,...,x_{s(n)}\la \cA(1^n)]\leq \beta(n)<(1-\epsilon(n))(1-\beta(n)).
\end{align*}

Now, we give the construction of $\Vrfy$.
\begin{description}
\item[$\Vrfy\left(x_1,...,x_{s(n)}\right)$:]$ $ 
\begin{enumerate}
\item Receive $x_1,...,x_{s(n)}$, which is supposed to be generated by some unknown algorithm $\cA(1^n)$.
\item Run $h\la\mathsf{Learn}(x_1,...,x_{s(n)})$.
\item Output $\top$ if $\mathsf{SD}(h(1^n), \cD(1^n))\leq \epsilon(n)-\alpha(n)$.
Otherwise, output $\bot$.
\end{enumerate}
\end{description}
\paragraph{Correctness.}
With probability at least $\beta(n)$ over $x_1,...,x_{s(n)}\la\cA(1^n)$ and $h\la\mathsf{Learn}(x_1,...,x_{s(n)})$, $h$ satisfies
\begin{align*}
\mathsf{SD}(h(1^n),\cD(1^n))\leq \alpha(n)\leq \epsilon(n)-\alpha(n).
\end{align*}
Hence, we have
\begin{align*}
\Pr[\top\la\Vrfy\left(x_1,...,x_{s(n)}\right):x_1,...,x_{s(n)}\la\cD(1^n)^{\otimes s(n)}]\geq 1-\beta(n).
\end{align*}

\paragraph{Adaptive-Soundness.}
With probability at least $\beta(n)$, $h$ satisfies
\begin{align*}
\mathsf{SD}(h(1^n),\mathsf{Marginal}_{\cA}(1^n))\leq \alpha(n).
\end{align*}

From the triangle inequality, for any $\cA$ such that 
\begin{align*}
\mathsf{SD}(\mathsf{Marginal}_{\cA}(1^n),\cD(1^n))\geq \epsilon(n),
\end{align*}
we have
\begin{align*}
\mathsf{SD}(h(1^n),\cD(1^n))\geq \epsilon(n)-\alpha(n)
\end{align*}
with probability at least $B$.
Therefore, we have
\begin{align*}
\Pr[\top\la\Vrfy\left(x_1,...,x_{s(n)}\right):x_1,...,x_{s(n)}\la\cA(1^n) ]\leq \beta(n).
\end{align*}
\end{proof}

\subsection{Unconditional Efficient Verification Algorithm for Distributions with Small Entropy}

In this section, we prove \cref{thm:small_entropy}, which guarantees that every PPT algorithm with small entropy can be verified by a PPT algorithm.
Note that, by the same argument, we can show that every QPT algorithm with small entropy can be verified by a QPT algorithm.

The soundness requirement of \cref{thm:small_entropy} is slightly stronger than the original definition of distribution verification with adaptive-soundness (Definition~\ref{def:adaptive}).
In Definition~\ref{def:adaptive}, we require that, for every time-bounded adversary $\cA$ with $\Delta(\mathsf{Marginal}_\cA(1^n),\cD(1^n))\geq \epsilon(n)$, the verifier must reject.
In other words, there might exist some unbounded algorithm $\cA$ that makes $\Vrfy$ accept even though $\mathsf{Marginal}_\cA(1^n)$ is statistically far from $\cD(1^n)$.
On the other hand, in \cref{thm:small_entropy}, $\Vrfy$ must reject any unbounded algorithm $\cA$ as long as $\Delta(\mathsf{Marginal}_\cA(1^n),\cD(1^n))\geq \epsilon(n)$.

\begin{proposition}[Adaptive-Verification of Classical Distributions with small entropy by PPT Algorithms (Restatement of \cref{thm:unconditional_verifiaction})]\label{thm:small_entropy}
For any constant $c,d,\alpha>0$, and for any PPT algorithms $\cD$, which takes $1^n$ as input, and outputs $x\in\bit^{m(n)}$ for an arbitrary polynomial $m$, such that $H(\cD(1^n))\geq \alpha\log(n)$, there exists a PPT algorithm $\Vrfy$ and a polynomial $s$ such that the following holds.
\begin{description}
    \item[Correctness:] We have
    \begin{align*}
    \Pr[\top\la\Vrfy(1^n,x_1,\dots,x_{s}):x_1,\dots,x_{s}\la\cD(1^n)^{\otimes s}] \geq 1-n^{-c}
    \end{align*}
    for all sufficiently large $n\in\N$.
    Here, $s = s(n^c,d,1/\epsilon(n))$.
    \item[Adaptive-Soundness:] 
    For any algorithm $\cA$ that takes $1^n$ as input and outputs $\left(x_1,\dots,x_{s(n)}\right)\in\left(\bit^{m(n)}\right)^{s(n)} $, 
    %\matthew{I feel like both here and in section 7 this should be $(x_1,\dots,x_{s(n)})\in\bit^{m(n)\times s(n)} $}\mor{you can modify in that way.} 
    we define an algorithm $\mathsf{Marginal}_{\cA}(1^n)$ as follows:
    \begin{description}
    \item[$\mathsf{Marginal}_{\cA}(1^n)$:] $ $ 
    \begin{enumerate}
        \item Sample $(x_1,\dots,x_s)\la\cA(1^n)$.
        \item Sample $i\gets [s]$.
        \item Output $x_i$.
    \end{enumerate}
    \end{description}
    For any adversary $\cA$ with $\Delta(\mathsf{Marginal}_{\cA}(1^n),\cD(1^n))\geq \epsilon(n)$ for all $n\in\N$, we have
    \begin{align*}
    \Pr[\top\gets\Ver(1^n,x_1,\dots,x_s) :(x_1,...,x_s)\gets\cA(1^n)] \leq (1-\epsilon(n)) + \frac{1}{d}
    \end{align*}
    for all sufficiently large $n\in\N$.
    Here, $s = s(n^c,d,1/\epsilon(n))$.
\end{description}
\end{proposition}

\begin{proof}[Proof of \cref{thm:small_entropy}]
For any constant $c,d>0$, and any PPT algorithm $\cD$ with $H(\cD(1^n))\leq c\log(n)$, we construct a PPT algorithm $\Vrfy$ such that the following is satisfied for some polynomial $s$. 
\begin{itemize}
\item 
\begin{align*}
\Pr[\top\la\Vrfy(x_1,\dots x_{s(n)}): (x_1,\dots,x_{s(n)})\la\cD(1^n)^{\otimes s(n)}]\geq 1-\negl(n)
\end{align*}
for all sufficiently large $n\in\N$.
\item For any PPT algorithm $\cA$ with $\Delta(\mathsf{Marginal}_{\cA}(1^n),\cD(1^n))\geq \epsilon(n)$, we have
\begin{align*}
\Pr[\top\la\Vrfy(x_1,\dots x_{s(n)}): (x_1,\dots,x_{s(n)})\la\cA(1^n)] \leq 1-\epsilon(n)+\frac{1}{d}
\end{align*}
for all sufficiently large $n\in\N$.
\end{itemize}

Let $m$ be a polynomial such that $\cD(1^n)$ outputs $m(n)$-bit strings.
Let us describe our $\Vrfy$ algorithm.
\begin{description}
\item[$\Vrfy$:]$ $
\begin{enumerate}
\item Receive $x_1,\dots,x_{n^{1000cd}}$.
\item Let $A^*_{x_1,\dots,x_{n^{1000cd}}}$ be a distribution that samples $i\la[n^{1000cd}]$, and outputs $x_i$.
\item Run $X_i\la\cD(1^n)$ for all $i\in[m(n)\cdot n^{1000cd}]$. For each $X_i$, let $\mathsf{Count}_{X_1,\dots,X_{m(n)\cdot n^{1000cd}}}(X_i)$ be the number of $j\in[m(n)\cdot n^{1000cd}]$ such that $X_i= X_j$.
\item Let $\mathsf{List}_{n}$ be a set of $X\in\{X_i\}_{i\in[m(n)\cdot n^{1000cd}]} $ such that $ \frac{\mathsf{Count}_{X_1,\dots,X_{m(n)\cdot n^{1000cd}}}(X)}{m(n)\cdot n^{1000cd}}\geq n^{-200cd}$.
\item Output $\top$ if $\abs{\frac{\mathsf{Count}_{X_1,\dots,X_{m(n)\cdot n^{1000cd}}}(X_i)}{m(n)\cdot n^{1000cd}}-\Pr[X_i\la A^*_{x_1,\dots,x_{n^{1000cd}}}]}\leq n^{-100cd}$ for all $X_i\in\mathsf{List}_n$.
Otherwise, output $\bot$.
\end{enumerate}
\end{description}
\paragraph{Correctness.}
Let 
\begin{align*}
\cS_{\cD}(1^n,250cd)\seteq \left\{\Pr[x\la\cD(1^n)]\geq n^{-250cd} \right\}.    
\end{align*}
We use the following Claim~\ref{claim:list}.
\begin{claim}\label{claim:list}
With probability at least $1-\negl(n)$ over $\{X_i\}_{i\in[m(n)\cdot n^{1000cd}]}\la\cD(1^n)^{\otimes m(n)\cdot n^{1000cd}}$, we have
\begin{align*}
x\notin\mathsf{List}_n 
\end{align*}
for all $x\notin\cS_{\cD}(1^n,250cd)$,
and
\begin{align*}
\abs{\frac{\mathsf{Count}_{X_1,\dots,X_{m(n)\cdot n^{1000cd}}}(X)}{m(n)\cdot n^{1000cd}}-\Pr[X\la\cD(1^n) ]}\leq n^{-250cd}
\end{align*}
for all $X\in\mathsf{List}_n$.
\end{claim}

From Hoeffding's inequality, for each $X\in\cS_{\cD}(1^n,250cd)$, we have
\begin{align*}
\Pr_{x_1,\dots,x_{n^{1000cd}}\la\cD(1^n)^{\otimes n^{1000cd}} }\left[\abs{\Pr[X\la\cA^*_{x_1,\dots,x_{n^{1000cd}}}]-\Pr[X\la\cD(1^n) ]}\leq \frac{1}{n^{250cd}}\right]
&\geq 1-2\exp(-\frac{n^{1000cd}}{n^{500cd}})\\
&\geq 1-2\exp(-n^{500cd}).
\end{align*}
This implies that 
\begin{align*}
&\Pr_{x_1,\dots,x_{n^{1000cd}}\la\cD(1^n)^{\otimes n^{1000cd}}}\left[\abs{\Pr[X\la\cA^*_{x_1,\dots,x_{n^{1000cd}}}]-\Pr[X\la\cD(1^n) ]}\leq \frac{1}{n^{250cd}} \mbox{\,\,for all $X\in\cS_{\cD}(1^n,250cd)$} \right]\\
&\geq \left(1-2\exp(-n^{500cd})\right)^{n^{250cd}}\geq 1-\negl(n).
\end{align*}

From Claim~\ref{claim:list}, these imply that with probability at least $1-\negl(n)$, for all $X\in\mathsf{List}_n$, we have
\begin{align*}
&\abs{\frac{\mathsf{Count}_{X_1,\dots,X_{m(n)\cdot n^{1000cd}}}(X)}{m(n)\cdot n^{1000cd}}-\Pr[X\la\cA^*_{x_1,\dots,x^{n^{1000cd}}}]}\\
&\leq \abs{\frac{\mathsf{Count}_{X_1,\dots,X_{m(n)\cdot n^{1000cd}}}(X)}{m(n)\cdot n^{1000cd}}-\Pr[X\la\cD(1^n)]}+\abs{\Pr[X\la\cA^*_{x_1,\dots,x^{n^{1000cd}}}]-\Pr[X\la\cD(1^n)]}\leq 2n^{-250cd}.
\end{align*}
Therefore, we have
\begin{align*}
\Pr_{x_1,\dots,x_{n^{1000cd}}\la\cD(1^n)^{\otimes n^{1000cd}}}[\top\la\Vrfy(x_1,\dots,x_{n^{1000cd}})]\geq 1-\negl(n).
\end{align*}

\begin{proof}[Proof of Claim~\ref{claim:list}]
For each $X\in\bit^{m(n)}$, we have
\begin{align*}
&\Pr_{X_1,\dots,X_{m(n)\cdot n^{1000cd}}\la\cD(1^n)^{\otimes m(n)\cdot n^{1000cd}}}\left[\abs{\frac{\mathsf{Count}_{X_1,\dots,X_{m(n)\cdot n^{1000cd}}}(X)}{n^{1000cd}}-\Pr[X\la\cD(1^n)]}\leq n^{-250cd}\right]\\
&\geq 1-2\exp(-m(n)\cdot n^{500cd}).
\end{align*}
From union bound, with probability at least $1-\negl(n)$ over $X_1,\dots,X_{m(n)\cdot n^{1000cd}}\la\cD(1^n)^{\otimes m(n)\cdot n^{1000cd}}$, we have
\begin{align*}
\abs{\frac{\mathsf{Count}_{X_1,\dots,X_{m(n)\cdot n^{1000cd}}}(X)}{m(n)\cdot n^{1000cd}}-\Pr[X\la\cD(1^n)]}\leq n^{-250cd} 
\end{align*}
for all $X\notin\bit^{m(n)}$.
\end{proof}

\paragraph{Adaptive-Soundness.}

\if0
\taiga{!!Intuition!! From here}
Suppose that $\Vrfy$ outputs $\top$ with high probability.
This roughly means that \cref{ineq:condition} holds with high probability.
From \cref{claim:test}, this means that $\Delta(A^*_{x_1,\dots,x_{100n^{cd}}},\cD(1^n))\leq small$ with high probability over $x_1,\dots,x_{100n^{cd}}\la\cA(1^n)$.
From probabilistic argument, this implies that $\Delta(\mathsf{Marginal}_{\cA}(1^n),\cD(1^n))\leq small$.
\taiga{To here}
\fi

For showing the soundness, we use the following Claims~\ref{claim:test} and \ref{claim:prob}.
For showing Claim~\ref{claim:test}, we use the following Claim~\ref{claim:entropy}.
We defer the proof of them.
\begin{claim}\label{claim:test}
Let $\cD$ be an arbitrary PPT algorithm, which takes $1^n$ as input, outputs $x\in\bit^{m(n)}$, and satisfies $H(\cD(1^n))\leq c\log(n)$.
Let 
\begin{align*}
\cS_{\cD}(1^n,50cd)\seteq\left\{x\in\bit^*:\Pr[x\la\cD(1^n)]\geq \frac{1}{n^{50cd}}\right\}.
\end{align*}
Suppose that 
\begin{align}\label{ineq:condition}
\abs{\Pr[x\la\cA^*_{x_1,\dots,x_{n^{1000cd}}}]-\Pr[x\la \cD(1^n)]}\leq 2n^{-100cd} \mbox{\,\,for all\,\,} x\in \cS_{\cD}(1^n,50cd).
\end{align}
Then, we have
\begin{align*}
\Delta(A^*_{x_1,\dots,x_{n^{1000cd}}},\cD(1^n))\leq\frac{1}{49d}
\end{align*}
for all sufficiently large $n\in\N$.
\end{claim}

\begin{claim}\label{claim:prob}
Let 
\begin{align*}
\cS_{\cD}(1^n,50cd)\seteq\left\{x\in\bit^*:\Pr[x\la\cD(1^n)]\geq \frac{1}{n^{50cd}}\right\}.
\end{align*}
With probability at least $1-\negl(n)$ over $\{X_i\}_{i\in[m(n)\cdot n^{1000cd}]}\la\cD(1^n)^{\otimes m(n)\cdot n^{1000cd}}$, for any $X\in\cS_{\cD}(1^n,50cd)$, we have 
\begin{align*}
\abs{\frac{\mathsf{Count}_{X_1,\dots,X_{m(n)\cdot n^{1000cd}}}(X)}{m(n)\cdot n^{1000cd}}-\Pr[X\la\cD(1^n)]}\leq n^{-250cd}.
\end{align*}
\end{claim}

\begin{claim}\label{claim:entropy}
Let $\cD$ be an arbitrary PPT algorithm, which takes $1^n$ as input, outputs $x\in\bit^{m(n)}$, and satisfies $H(\cD(1^n))\leq c\log(n)$.
We define $\cS_\cD(1^n,d)$ as follows:
\begin{align*}
\cS_\cD(1^n,d)\seteq \left\{x\in\bit^{m(n)}: \Pr[x\la\cD(1^n)]\geq \frac{1}{n^d} \right\}.
\end{align*}
Then, we have
\begin{align*}
\Pr_{x\la \cD(1^n)}[x\in\cS_{\cD}(1^n,d)]\geq 1-\frac{c}{d}.
\end{align*}
\end{claim}

For contradiction, suppose that 
\begin{align*}
\Pr[\top\la\Vrfy(x_1,\dots,x_{n^{1000cd}}):x_1,\dots,x_{n^{1000cd}}\la\cA(1^n)]\geq 1-\epsilon(n)+\frac{1}{d}.
\end{align*}
From Claim~\ref{claim:prob}, with probability at least $1-\negl(n)$ over $\{X_i\}_{i\in[m(n)\cdot n^{1000cd}]}\la\cD(1^n)^{\otimes m(n)\cdot n^{1000cd}}$, for any $X\in\cS_{\cD}(1^n,50cd)$,
\begin{align*}
\abs{\frac{\mathsf{Count}_{X_1,\dots,X_{m(n)\cdot n^{1000cd}}}(X)}{m(n)\cdot n^{1000cd}}-\Pr[X\la\cD(1^n) ]}\leq n^{-250cd} 
\end{align*}
and 
\begin{align*}
\cS_{\cD}(1^n,50cd)\subseteq \mathsf{List}_n.
\end{align*}
This and \cref{ineq:condition} imply that
\begin{align*}
&\Pr[\abs{\Pr[x\la\cA^*_{x_1,\dots,x_{n^{1000cd}}}]-\Pr[x\la \cD(1^n)]}\leq 2n^{-100cd} \mbox{\,\,for all $x\in\cS_{\cD}(1^n,50cd)$}:x_1,\dots,x_{n^{1000cd}}\la\cA(1^n)]\\
&\geq 1-\epsilon(n)-\negl(n)+\frac{1}{d}.
\end{align*}
From Claim~\ref{claim:test}, this implies that
\begin{align}\label{ineq:consequence}
\Pr[\Delta(\cA^*_{x_1,\dots,x_{n^{1000cd}}},\cD(1^n))\leq \frac{1}{49d}:x_1,\dots,x_{n^{1000cd}}\la\cA(1^n)] \geq 1-\epsilon(n)-\negl(n)+\frac{1}{d}.
\end{align}
Let
\begin{align*}
\cT_{n,d}\seteq \left\{x_1,\dots,x_{n^{1000cd}}\in\bit^*:\Delta(\cA^*_{x_1,\dots,x_{n^{1000cd}}},\cD(1^n))\leq \frac{1}{49d} \right\}.
\end{align*}

Then, \cref{ineq:consequence} implies that
\begin{align*}
&\Delta\left(\mathsf{Marginal}_{\cA}(1^n),\cD(1^n)\right)\\
&\leq\mathbb{E}_{x_1,\dots,x_{n^{1000cd}}\la\cA(1^n)}\left[\Delta(\cA^*_{x_1,\dots,x_{n^{1000cd}}},\cD(1^n))\right]\\
&=\sum_{x_1,\dots,x_{n^{1000cd}}}\Pr[x_1,\dots,x_{n^{1000cd}}\la\cA(1^n)]\Delta(\cA^*_{x_1,\dots,x_{n^{1000cd}}},\cD(1^n))\\
&=\sum_{x_1,\dots,x_{n^{1000cd}}\in \cT_{n,d}} \Pr[x_1,\dots,x_{n^{1000cd}}\la\cA(1^n)]\Delta(\cA^*_{x_1,\dots,x_{n^{1000cd}}},\cD(1^n))\\
&+\sum_{x_1,\dots x_{n^{1000cd}}\notin \cT_{n,d}} \Pr[x_1,\dots,x_{n^{1000cd}}\la\cA(1^n)]\Delta(\cA^*_{x_1,\dots,x_{n^{1000cd}}},\cD(1^n))\\
&\leq \sum_{x_1,\dots,x_{n^{1000cd}}\in \cT_{n,d}}\Pr[x_1,\dots,x_{n^{1000cd}}\la\cA(1^n)]\frac{1}{49d}+\sum_{x_1,\dots,x_{n^{1000cd}}\notin \cT_{n,d}}\Pr[x_1,\dots,x_{n^{1000cd}}\la\cA(1^n)]\\
&\leq \frac{1}{49d}\left(1-\epsilon(n)-\negl(n)+\frac{1}{d}\right) +\epsilon(n)+\negl(n)-\frac{1}{d}\leq \epsilon(n)
\end{align*}
for all sufficiently large $n\in\N$.
Here, in the second inequality, we have used \cref{ineq:consequence}.
This is a contradiction.

\begin{proof}[Proof of Claim~\ref{claim:test}]
We have
\begin{align*}
&\Delta(A^*_{x_1,\dots,x_{n^{1000cd}}},\cD(1^n))\\
&=\frac{1}{2}\Bigg(\sum_{x\in\cS_{\cD}(1^n,50cd )}\abs{ \Pr[x\la\cA^*_{x_1,\dots,x_{n^{1000cd}}}] -\Pr[x\la\cD(1^n)]}\\
&+\sum_{x\notin\cS_{\cD}(1^n,50cd) }\abs{ \Pr[x\la\cA^*_{x_1,\dots,x_{n^{1000cd}}}] -\Pr[x\la\cD(1^n)]}
\Bigg)\\
&\leq \frac{1}{2}\Bigg(\abs{\cS_{\cD}(1^n,50cd )}2n^{-100cd}\\
&+\sum_{x\notin\cS_{\cD}(1^n,50cd) }\abs{\Pr[x\la\cA^*_{x_1,\dots,x_{n^{1000cd}}}] -\Pr[x\la\cD(1^n)]}
\Bigg)\\
&\leq n^{-50cd}+\frac{1}{2}\sum_{x\notin\cS_{\cD}(1^n,50cd)}\Pr[x\la\cA^*_{x_1,\dots,x_{n^{1000cd}}}]+\frac{1}{2}\sum_{x\notin\cS_{\cD}(1^n,50cd)}\Pr[x\la\cD(1^n)]\\
&\leq n^{-50cd}+\frac{1}{100d}+\frac{1}{100d}+n^{-50cd}\leq \frac{1}{49d}
\end{align*}
for all sufficiently large $n\in\N$.
Here, in the third inequality, we have used \cref{claim:entropy}, and
\begin{align*}
\sum_{x\in \cS_{\cD}(1^n,50cd)}\Pr[x\la\cA^*_{x_1,\dots,x_{n^{1000cd}}}]
&\geq \sum_{x\in \cS_{\cD}(1^n,50cd)}(\Pr[x\la\cD(1^n)] -2n^{-100cd})\\
&\geq 1-\frac{1}{50d}-2n^{-50cd}.
\end{align*}
\end{proof}

\begin{proof}[Proof of Claim~\ref{claim:prob}]
From Hoeffding's inequality, for each $X\in\cS_{\cD}(1^n,50cd)$, we have
\begin{align*}
\Pr_{x_1,\dots,x_{m(n)\cdot n^{1000cd}}\la\cD(1^n)^{\otimes m(n)\cdot n^{1000cd}} }\left[\abs{\frac{\mathsf{Count}_{x_1,\dots,x_{m(n)}}(X)}{m(n)\cdot n^{1000cd}}-\Pr[X\la\cD(1^n) ]}\leq \frac{1}{n^{250cd}}\right]
\geq 1-2\exp(-n^{500cd}).
\end{align*}
From union bound, we have
\begin{align*}
&\Pr_{x_1,\dots,x_{m(n)\cdot n^{1000cd}}\la\cD(1^n)^{\otimes m(n)\cdot n^{1000cd}} }\left[\abs{\frac{\mathsf{Count}_{x_1,\dots,x_{m(n)}}(X)}{m(n)\cdot n^{1000cd}}-\Pr[X\la\cD(1^n) ]}\leq \frac{1}{n^{250cd}} \mbox{ for all }X\in\cS_{\cD}(1^n,50cd)\right]\\
&\geq \left(1-2\exp(-n^{500cd})\right)^{n^{50cd}}\geq 1-\negl(n) .
\end{align*}

\end{proof}

\begin{proof}[Proof of Claim~\ref{claim:entropy}]
This immediately holds from the definition of entropy as follows:
\begin{align*}
c\log(n)
&\geq H(D(1^n))\\
&=\sum_{x\in\cS_{n,d}}p_x\log\left(\frac{1}{p_x}\right)+\sum_{x\notin\cS_{n,d}}p_x\log\left(\frac{1}{p_x}\right)\\
&\geq \sum_{x\notin\cS_{n,d}}p_x\log\left(\frac{1}{p_x}\right)\geq d\log(n)\sum_{x\notin\cS_{n,d}}p_x.
\end{align*}
This implies that
\begin{align*}
\Pr_{x\la\cD(1^n)}[x\in\cS_{\cD}(1^n,d)]\geq 1-\frac{c}{d}.
\end{align*}
\end{proof}
\end{proof}

\section{Distribution Verification through Kolmogorov Complexity}
\label{sec:ver-with-K}

In this section, we show how to construct a distribution verification protocol using techniques that leverage Kolmogorov complexity.
\taiga{@Matthew: Is the sentence above grammatically correct? ``a distribution verification'' instead? }
We also show that the existence of OWFs and QEFID implies the hardness of verification of classical distributions and quantum distributions, respectively.

\subsection{Classical Distribution Verification}
\label{sec:OWF}
In this section, we show an equivalence between the existence of OWFs and
hardness of classical distribution verification.
\subsubsection{Efficient Verification of Classical Distributions from the Non-Existence of OWFs}
First, we show how to efficiently verifiy any classical distribution assuming the non-existence of OWFs.

\begin{theorem}[Restatement of \cref{inf:owf}]\label{thm:easiness}
    Assume that infinitely-often OWFs do not exist. 
    Then every uniform PPT algorithm $\cD$, which takes $1^n$ as input and outputs $x\in\bit^{m(n)}$, where $m$ is a polynomial, is adaptively-verifiable with a PPT algorithm with classical-security.
\end{theorem}

For showing \cref{thm:easiness}, we will use the following \cref{thm:probest,thm:ukest,thm:ukest_better}.

\begin{theorem}[No i.o.-OWFs implies one-sided error probability estimation~\cite{FOCS:ImpLev90,STOC:IlaRenSan22,CHK25}]\label{thm:probest}
    If there do not exist infinitely-often OWFs, then for any polynomial $m$, for any PPT algorithm $\cD$, which takes $1^n$ as input and outputs $x\in\bit^{m(n)}$, and for any constant $c\in \N$, there exists a PPT algorithm $\mathsf{Approx}$ such that for all sufficiently large $n$,
    \begin{align*}
    \Pr_{\substack{x \la \cD(1^n) \\ \mathsf{Approx}}}\left[\frac{1}{2}\Pr[x\la\cD(1^n)]\leq \mathsf{Approx}(x,1^{n})\leq \Pr[x\la\cD(1^n)]\right] \geq 1-n^{-c}.
    \end{align*}

    Furthermore, for all sufficiently large $n$, for all $x\in\bit^{m(n)}$,
    \begin{align*}
        \Pr_{\mathsf{Approx}}[\mathsf{Approx}(x,1^n) \leq \Pr[x\la\cD(1^n)]] \geq 1 - n^{-c}.
    \end{align*}
\end{theorem}

\begin{theorem}[No i.o.-OWFs implies $uK^t$ estimation on all PPT distributions (Theorem 7.1 of \cite{FOCS:HirNan23})]\label{thm:ukest}
    If there do not exist infinitely-often OWFs, then there exists a PPT algorithm $\cM$ such that for any polynomial $m$ and for any PPT algorithm $\cD$, which takes $1^n$ as input and outputs $x\in\bit^{m(n)}$, there exists a polynomial $t_0$, such that for every $t > t_0(n)$ and every $\delta^{-1} \in \N$, for all sufficiently large $n$,

\begin{align*}
    \Pr_{\cM,x \la \cD(1^n)}[{uK^t(x)-1\leq \cM(x,1^t,1^{\delta^{-1}})\leq uK^t(x)}+1] \geq 1-\delta.
\end{align*}
\end{theorem}

For our purposes, we need a slightly modified version of \cref{thm:ukest}. We give this modified version below.
While \cref{thm:ukest} asserts that $uK^t(x)$ can be approximated, \cref{thm:ukest_better} guarantees that $uK^t(x \mid 1^n)$ can be approximated.
In short, \cref{thm:ukest} is proven by estimating the probability $\Pr_{r\la \bit^t}[x\la\mathcal{U}^t(r)]$, which is possible assuming the non-existence of OWFs.
\cref{thm:ukest_better} can be proven by estimating $\Pr_{r\la \bit^t}[x\la\mathcal{U}^t(1^n,r)]$ instead.
Because the proofs are otherwise identical, we omit the proof of the length conditional version.

\begin{theorem}[Length conditional version of $uK^t$ estimation]\label{thm:ukest_better}
    If there do not exist infinitely-often OWFs, then there exists a PPT algorithm $\cM$ such that for any PPT algorithm $\cD$, which takes $1^n$ as input and outputs $x\in\bit^{m(n)}$, there exists a polynomial $t_0$, such that for every $t > t_0(n)$ and every $\delta^{-1} \in \N$, for all sufficiently large $n$,
    \begin{align*}
    \Pr_{\cM,x \la \cD(1^n)}[{uK^{t}(x|1^n)-1\leq \cM(x,1^{n,t,\delta^{-1}})\leq uK^t(x|1^n)+1}] \geq 1-\delta.
    \end{align*}
\end{theorem}

\begin{proof}[Proof of \cref{thm:easiness}]

In the following, assume that infinitely-often OWFs do no exist.
Then, we give a polynomial $s$ and a PPT algorithm $\Vrfy$ which satisfies the correctness and adaptive soundness in \Cref{def:adaptive}.
For describing $\Vrfy$, let us introduce several notations below.

\paragraph{Notations:}
We set $s(n,t(n),\epsilon(n)) = n^{4c}\left(\log_2\frac{1}{1-\epsilon(n)} + 2(\log(n))^2 \right) $.
% For simplicity, we often denote $s=n^{4c}\left(\log_2\frac{1}{1-\epsilon(n)} + 2(\log(n))^2 \right)$ below.
Let $t_0$ be a polynomial given in \cref{lem:coding}. We also set $\alpha=(\log(n))^2$.

$\mathsf{Approx}$ is a PPT algorithm given in \Cref{thm:probest} such that 
\begin{align*}
\Pr_{x\la\cD(1^n)^{\otimes s}}\left[\frac{1}{2}\Pr[x\la\cD(1^n)^{\otimes s}]\leq \mathsf{Approx}(x,1^{n})\leq \Pr[x\la\cD(1^n)^{\otimes s}]\right] \geq 1-n^{-4c}
\end{align*}
and
\begin{align*}
\Pr[\mathsf{Approx}(x,1^n) \leq \Pr[x\la\cD(1^n)^{\otimes s}]] \geq 1 - n^{-4c}
\end{align*}
for all $x\in\bit^{s\cdot m(n)}$ and all sufficiently large $n\in\N$.

$\cM$ is a PPT algorithm given in \cref{thm:ukest_better} such that, for any PPT algorithm $\cA$, we have
\begin{align*}
\Pr_{x \la \cA(1^n)}[{uK^{t_0(t(n))}(x|1^n)-1\leq \cM(x,1^{t})\leq uK^{t_0(t(n))}(x|1^n)}+1] \geq 1-n^{-4c}
\end{align*}
for all sufficiently large $n\in\N$.

\paragraph{Construction:}
We give a construction of $\Vrfy$.

\begin{description}
    \item[$\Vrfy$:]$ $
    \begin{enumerate}
        \item Take $1^n$ and $y_1,\dots,y_{s(n)}$ as input.
        \item Compute $\cM(y_1,\dots,y_{s(n)},1^{t}) \to k$
        \item Compute $\mathsf{Approx}(y_1,\dots,y_{s(n)}) \to p$
        \item Output $\top$ if $-\log_2 p \leq k + \alpha$.
        Otherwise, output $\bot$.
    \end{enumerate}
\end{description}
In the following, we show that $\Vrfy$ satisfies correctness and security.

\paragraph{Correctness:}
    In the following, we show that
    \begin{equation*}
        \Pr[\top\la\Ver(1^n,y_1,\dots,y_s) : y_1,\dots,y_{s}\la\cD(1^n)^{\otimes s} ] \geq 1 - \frac{1}{n^c}
    \end{equation*}
    for all sufficiently large $n\in\N$.

In the following, we use the following Claims \ref{claim:est_D} to \ref{claim:kraft}.
These directly follow from probabilistic arguments.
For clarity, we describe the proof in the end of proof of correctness.
\begin{claim}\label{claim:est_D}
Let us denote
\begin{align*}
\cI_n\seteq \left\{y_1,\dots,y_s\in(\bit^{m(n)})^{\times s} : \Pr[\mathsf{Approx}(y_1,\dots,y_s)\geq \frac{\Pr[y_1,\dots,y_s\la\cD(1^n)^{\otimes s}]}{2}]\geq 1-n^{-2c}\right\}.
\end{align*}
Then, we have
\begin{align*}
\Pr_{y\la\cD(1^n)^{\otimes s}}[y_1,\dots,y_s\in\cI_n]\geq 1-n^{-2c}
\end{align*}
for all sufficiently large $n\in\N$.
\end{claim}

\begin{claim}\label{claim:est_uK}
Let us denote
\begin{align*}
\mathcal{J}_n\seteq \left\{(y_1,\dots,y_s)\in(\bit^{m(n)})^{\times s} : \Pr[\cM(y_1,\dots,y_s)\geq uK^{t_0(t(n))}(y_1,\dots,y_s|1^n)-1]\geq 1-n^{-2c}\right\}.
\end{align*}
Then, we have
\begin{align*}
\Pr_{y_1,\dots,y_s\la\cD(1^n)^{\otimes s}}[y_1,\dots,y_s\in \mathcal{J}_n]\geq 1-n^{-2c}
\end{align*}
for all sufficiently large $n\in\N$.
\end{claim}

\begin{claim}\label{claim:kraft}
Let us denote 
\begin{align*}
&\cK_n\seteq\Bigg\{(y_1,\dots,y_{s})\in(\bit^{m(n)})^{\times s}:\\
&\hspace{2cm}-\log_2 (\Pr[y_1,\dots,y_s\la\cD(1^n)^{\otimes s}]) \geq uK^{t_0(t(n))}(y_1,\dots,y_s|1^n) + \alpha - 2\Bigg\}.
\end{align*}
Then, we have
\begin{align*}
\Pr_{y_1,\dots,y_s\la\cD(1^n)^{\otimes s}}[y_1,\dots,y_s\in\cK_n]\leq \frac{1}{n^{2c}}
\end{align*}
for all sufficiently large $n\in\N$.
\end{claim} 
From union bound, we have
\begin{align*}
    &\Pr_{y_1,\dots,y_s\la\cD(1^n)^{\otimes s}}[\bot\la\Ver(1^n,y_1,\dots,y_s)]\\
    &\leq\sum_{y_1,\dots,y_s\notin \cI_n\cap\mathcal{J}_n}\Pr[y_1,\dots,y_s\la\cD(1^n)^{\otimes s(n)}]\\
    &+\sum_{y_1,\dots,y_s\in \cI_n\cap\mathcal{J}_n}
    \Pr[y_1,\dots,y_s\la\cD(1^n)^{\otimes s(n)}]
    \Pr_{\mathsf{Approx},\cM}[-\log_2 \left(\mathsf{Approx}(y_1,\dots,y_s)\right)\geq \cM(y_1,\dots,y_s)+\alpha]\\
    &\leq\frac{2}{n^{2c}}\\
    &+\sum_{y_1,\dots,y_s\in \cI_n\cap\mathcal{J}_n}
    \Pr[y_1,\dots,y_s\la\cD(1^n)^{\otimes s(n)}]
    \Pr_{\mathsf{Approx},\cM}[-\log_2 \left(\mathsf{Approx}(y_1,\dots,y_s)\right)\geq \cM(y_1,\dots,y_s)+\alpha]\\
    &\leq\frac{2}{n^{2c}}+\frac{2}{n^{2c}}\\
    &+ 
    \sum_{y_1,\dots,y_s\in \cI_n\cap\mathcal{J}_n}
    \Pr[y_1,\dots,y_s\la\cD(1^n)^{\otimes s(n)}]\\
    &\hspace{2cm}\Pr[-\log_2 \left(\Pr[y_1,\dots,y_s\la\cD(1^n)^{\otimes s}]\right)\geq uK^{t_0(t(n))}(y_1,\dots,y_s|1^n)+\alpha-2]
    \\
    &\leq\frac{4}{n^{2c}}+\sum_{y_1,\dots,y_s\in\cI_n\cap\mathcal{J}_n\cap\cK_n}\Pr[y_1,\dots,y_s\la\cD(1^n)^{\otimes s}]\\
    &\leq \frac{5}{n^{2c}}
    \end{align*}
for all sufficiently large $n\in\N$, which completes the proof of the correctness. 
Here, in the second inequality, we have used Claims~\ref{claim:est_D} and \ref{claim:est_uK}, in the third inequality, we have used the definition of $\cI_n$ and $\mathcal{J}_n$ and in the final inequality, we have used Claim~\ref{claim:kraft}.

\begin{proof}[Proof of Claim~\ref{claim:est_D}]
We have
\begin{align*}
1-n^{-4c}
&\leq\Pr_{y_1,\dots,y_s\la\cD(1^n)^{\otimes s}}\left[\frac{\Pr[y_1,\dots,y_s\la\cD(1^n)^{\otimes s}]}{2} \leq \mathsf{Approx}(y_1,\dots,y_s)\right]\\
&=\sum_{y_1,\dots,y_s\in\cI_n}\Pr[y_1,\dots,y_s\la\cD(1^n)^{\otimes s}]\Pr[\frac{\Pr[y_1,\dots,y_s\la\cD(1^n)^{\otimes s}]}{2} \leq \mathsf{Approx}(y_1,\dots,y_s)]\\
&+\sum_{y_1,\dots,y_s\notin\cI_n}\Pr[y_1,\dots,y_s\la\cD(1^n)^{\otimes s}]\Pr[\frac{\Pr[y_1,\dots,y_s\la\cD(1^n)^{\otimes s}]}{2} \leq \mathsf{Approx}(y_1,\dots,y_s)]\\
&\leq\sum_{y_1,\dots,y_s\in\cI_n}\Pr[y_1,\dots,y_s\la\cD(1^n)^{\otimes s}]+\sum_{y_1,\dots,y_s\notin\cI_n}\Pr[y_1,\dots,y_s\la\cD(1^n)^{\otimes s}](1-n^{-2c})
\end{align*}
for all sufficiently large $n\in\N$.
This implies that
\begin{align*}
1-n^{-2c}\leq \sum_{y_1,\dots,y_s\in\cI_n}\Pr[y_1,\dots,y_s\la\cD(1^n)^{\otimes s}]
\end{align*}
for all sufficiently large $n\in\N$.
\end{proof}

\begin{proof}[Proof of Claim~\ref{claim:est_uK}]
The proof of Claim~\ref{claim:est_uK} is the same as Claim~\ref{claim:est_D}.
More precisely, Claim~\ref{claim:est_D} analyzes $\mathsf{Approx}$, $\frac{\Pr[y_1,\dots,y_s \leftarrow \cD(1^n)^{\otimes s}]}{2}$, and $\cI_n$.
The proof can be adapted by performing the same analysis with $M$, $uK^t(x)$, and $\mathcal{J}_n$ in place of these quantities.
\end{proof}

\begin{proof}[Proof of Claim~\ref{claim:kraft}]
From Claim~\ref{thm:kraft}, we have
\begin{align*}
&\Pr_{y_1,\dots,y_s\la\cD(1^n)^{\otimes s}}[-\log_2(\Pr[y_1,\dots,y_s\la\cD(1^n)^{\otimes s}])\geq uK^{t_0(t(n))} (y_1,\dots,y_s|1^n)+\alpha-2]\\
&\leq (m(n)\cdot s+C)2^{-\alpha+3}
\leq 8(m(n)\cdot s+C)n^{-\log(n)}\leq n^{-2c}
\end{align*}
for all sufficiently large $n\in\N$.
This implies that
\begin{align*}
\Pr_{y_1,\dots,y_s\la\cD(1^n)^{\otimes s}}[y_1,\dots,y_s\in\cK_n]\leq \frac{1}{n^{2c}}
\end{align*}
for all sufficiently large $n\in\N$.
\end{proof}

\paragraph{Adaptive-Soundness:}

In the following, we prove the soundness.
In other words, for any $t$-time probabilistic adversary $\cA$, which takes $1^n$ as input and outputs $x\in\bit^{m(n)\cdot s}$, and satisfies $\Delta(\mathsf{Marginal}_\cA(1^n),\cD(1^n))\geq \epsilon(n)$, we show that
\begin{align*}
\Pr[\top\la\Vrfy(1^n,y_1,\dots,y_s) :y_1,\dots,y_s\la\cA(1^n)]\leq 1-\epsilon(n)+n^{-c}
\end{align*}
for all sufficiently large $n\in\N$.

We have the following Claims~\ref{claim:uK_est_2} and \ref{claim:coding}.
We can prove Claim~\ref{claim:uK_est_2} in the same way as Claim~\ref{claim:est_uK}.
Claim~\ref{claim:coding} directly follows from \cref{lem:coding} and $s=n^{4c}\left(\log_2\frac{1}{1-\epsilon(n)}+2(\log(n))^2\right)$.

\begin{claim}\label{claim:uK_est_2}
Let us denote $\cL_n$
\begin{align*}
\cL_n\seteq \left\{y_1,\dots,y_s\in\bit^{m(n)\cdot s}: \Pr\left[\cM(y_1,\dots,y_s)\leq uK^{t_0(t(n))}(y_1,\dots,y_s|1^n)+1\right]\geq 1-n^{-2c}\right\}.
\end{align*}
Then, we have
\begin{align*}
\Pr_{y_1,\dots,y_s\la\cA(1^n)^{\otimes s}}[y_1,\dots,y_s\in\cL_n]\geq 1-n^{-2c}
\end{align*}
% \matthew{[PRIORITY] Should be $y_1,\dots,y_s\la\cA(1^n)$ above. I'm basically certain.}
% \taiga{I agree. If I were correct, $y_1,\dots,y_s\la\cD(1^n)^{\otimes s}$ should be $y_1,\dots,y_s\la\cA(1^n)$.}
% for all sufficiently large $n\in\N$.
\end{claim}

\begin{claim}\label{claim:coding}
Let us denote 
\begin{align*}
\cN_n\seteq \left\{y_1,\dots,y_s\in\bit^{m(n)\cdot s}:-\log_2\left(\Pr[y_1,\dots,y_s\la\cD(1^n)^{\otimes s}]\right)\leq uK^{t_0(t(n))}(y_1,\dots,y_s|1^n)+1+\alpha\right\}.
\end{align*}
Then, for any PPT algorithm $\cA$ such that
\begin{align*}
\Delta(\mathsf{Marginal}_{\cA}(1^n),\cD(1^n))\geq \epsilon(n),
\end{align*}
we have
\begin{align*}
\Pr_{y_1,\dots,y_s\la\cA(1^n)}[y_1,\dots,y_s\in\cN_n]\leq 1-\epsilon(n)+n^{-2c}
\end{align*}
for all sufficiently large $n\in\N$.
\end{claim}

From union bound,
\begin{align*}
&\Pr_{y_1,\dots,y_s\la\cA(1^n)}[\top\la\Vrfy(1^n,y_1,\dots,y_s)]\\
&\leq\sum_{y_1,\dots,y_s\notin \cL_n}\Pr[y_1,\dots,y_s\la\cA(1^n)]\\
&+\sum_{y_1,\dots,y_s\in \cL_n}\Pr[y_1,\dots,y_s\la\cA(1^n)]\Pr[-\log_2\left(\mathsf{Approx}(y_1,\dots,y_s)\right)\leq \cM(y_1,\dots,y_s)+\alpha]\\
&\leq n^{-2c}\\
&+\sum_{y_1,\dots,y_s\in \cL_n}\Pr[y_1,\dots,y_s\la\cA(1^n)]\Pr[-\log_2\left(\mathsf{Approx}(y_1,\dots,y_s)\right)\leq \cM(y_1,\dots,y_s)+\alpha]\\
&\leq 3n^{-2c}\\
&+\sum_{y_1,\dots,y_s\in \cL_n}\Pr[y_1,\dots,y_s\la\cA(1^n)]\Pr[-\log_2\left(\Pr[y_1,\dots,y_s\la\cD(1^n)]\right)\leq uK^{t_0(t(n))}(y_1,\dots,y_s|1^n)+\alpha+1]\\
&\leq 4n^{-2c}+\sum_{y_1,\dots,y_s\in \cL_n\cap\cN_n}\Pr[y_1,\dots,y_s\la\cA(1^n)]\\
&\leq 1-\epsilon(n)+4n^{-2c} \leq 1-\epsilon(n)+n^{-c}
\end{align*}
for all sufficiently  large $n\in\N$. 
% \matthew{[PRIORITY] I have two questions about the last four or so lines in the above proof. 1) Why did it jump from $n^{-2c}$ to $3n^{-2c}$ instead of $2n^{-2c}$? and 2) How did it drop back down to $n^{-2c}$? I'm fairly sure that it should have been $n^{-2c}$ to $2n^{-2c}$ to $3n^{-2c}$, which is still fine, because that is still less than $n^{-c}$}
% \taiga{For (1), we need to stick to $3n^{-2c}$. This term comes from error for $\cM$ and $\mathsf{Approx}$. Both evaluation has error at least $n^{-2c}$.}
\taiga{For (2), I think that you are right. Please feel free to rewrite this.}
Here, in the second inequality, we have used Claim~\ref{claim:uK_est_2}, in the third inequality, we have used the definition of $\cL_n$ and the definition of $\mathsf{Approx}$, and in the final inequality, we have used Claim~\ref{claim:coding}.
This completes the proof of security.
\end{proof}

\subsubsection{Hardness of Verification of Classical Distributions from OWFs}

\begin{proposition}[Restatement of \Cref{inf:hard_owf}]\label{thm:hardness}
Consider any PPT algorithm $\cD$ that takes $1^n$ as input and outputs an $x\in\bit^{m(n)}$ with a polynomial $m$.
Suppose that $H(\cD(1^n))\geq n^{1/c}$ for a constant $c>0$.  
If OWFs exist, then $\cD$ is not selectively-verifiable with a PPT algorithm with classical-security.
\end{proposition}

For showing \cref{thm:hardness}, we will use the following Lemma~\ref{lem:entropy-distance}.
\begin{lemma}[Fannes-Inequality~\cite{NielsenChuang}]\label{lem:entropy-distance} For any pair of random variables $X$ and $Y$ over domain $U$
\begin{align*}
|H(X)-H(Y)| \leq \Delta(X,Y)\cdot\log(|U|)+1/e.
\end{align*}
\end{lemma}

\begin{proof}[Proof of \cref{thm:hardness}]
Let $\cD$ be a PPT algorithm, which takes $1^n$ as input and outputs $x\in\bit^{m(n)}$ such that $H(D(1^n))\geq n^{1/c}$ for some constant $c$.

In the following, we construct a PPT algorithm $\cA$ such that
\begin{align*}
\Delta(\cA(1^n),\cD(1^n)) \geq 1/\poly(n)
\end{align*}
and for every PPT algorithm $\Vrfy$, every polynomial $s$ and every polynomial $p$,
\begin{align}\label{eqn:ineq}
    \begin{split}
        &\Big|\Pr[\top\la\Vrfy(x_1,\dots,x_s):x_1,\dots,x_s\la\cD^{\otimes s(n)}(1^n)]\\
        &\hspace{3cm}-\Pr[\top\la\Vrfy(x_1,\dots,x_s):x_1,\dots,x_s\la\cA(1^n)^{\otimes s(n)}]\Big|
        \leq\frac{1}{p(n)}
    \end{split}
\end{align}
for infinitely many $n\in\N$.

Because $\cD$ is a PPT algorithm, $\cD$ internally takes $\ell(n)$ random bits as input for some polynomial $\ell$.
Let $G$ be an infinitely-often PRG, which takes $n^{1/c}/4$ bits of randomness to $\ell(n)$ bits.
Our construction of $\cA$ is as follows:
\begin{description}
\item[$\cA(1^n)$:]$ $
\begin{enumerate}
    \item Sample $r\la\bit^{\frac{n^{1/c}}{4}}$.
    \item Output $\cD(1^n;G(r))$.
\end{enumerate}
\end{description}

\paragraph{Statistically Far:}
From Lemma~\ref{lem:entropy-distance}, we have
\begin{align*}
\abs{H(\cA(1^n))-H(\cD(1^n))}&\leq \Delta(\cA(1^n),\cD(1^n))\cdot m(n)+1/e\\
\frac{3}{4}n^{1/c}-1/e &\leq\Delta(\cA(1^n),\cD(1^n))\cdot m(n).
\end{align*}
This implies that
\begin{align*}
\Delta(\cA(1^n),\cD(1^n) )\geq  \frac{1}{m(n)}\left(\frac{3n^{1/c}}{4}-1/e\right).
\end{align*}

\paragraph{Security:}
For contradiction, suppose that there exists a PPT adversary $\Vrfy$ that breaks \cref{eqn:ineq}.
Then, we construct a PPT adversary $\cB$ that breaks the security of PRGs.
\begin{description}
    \item[$\cB(1^n)$:]$ $
    \begin{enumerate}
        \item Receive $r_1,...,r_s$.
        \item Run $x_i\la \cD(1^n;r_i)$ for all $i\in[s]$.
        \item Run $b\la\Vrfy(x_1,\dots,x_s)$.
    \end{enumerate}
\end{description}
For contradiction, suppose that $\Vrfy$ breaks the \cref{eqn:ineq}.
This directly implies that $\cB$ breaks the security of PRGs.
\end{proof}

\section{Quantum State Verification}
\label{sec:state-ver}

\subsection{Weak Non-uniform EFI from Hardness of Quantum State Verification}

In this section, we prove \cref{thm:state_verification}.
\begin{theorem}[Restatement of \cref{informal:state_verification}]\label{thm:state_verification}
Assume that infinitely-often weak non-uniform EFI do not exist.
Then, every QPT algorithm $\cD$, which outputs a quantum state, is adaptively-verifiable.
\end{theorem}

\begin{proof}[Proof of \cref{thm:state_verification}]

In the following, we show that, for any polynomial $t$, any constant $c$, and any inverse polynomial $\epsilon$, there exists a QPT algorithm $\Vrfy$ and a polynomial $s$ that satisfies correctness and adaptive soundness assuming the non-existence of infinitely-often weak non-uniform EFI.

Before describing our verification protocol, let us introduce the notations to describe the verification protocol.

\paragraph{Notations.}

Let $U$ be a classical Turing machine.

For each $n\in\N$ and each $\Pi\in\bit^*$, we define $\widetilde{\mathsf{QU}}^{t(n)}_{s,m}(1^n,\Pi)$.

\begin{description}
    \item[$\widetilde{\mathsf{QU}}_{s,m}^{t(n)}(1^n,\Pi)$:]$ $
    \begin{itemize}
\item Run a classical universal Turing machine $U$ on $1^n$ and $\Pi$, halt within $t(n)$-steps, and obtain $c\in\bit^*$.
\item Consider $c$ as an encoding of an input size $k(n)$, and quantum circuit $C$.
\item Prepare $C\ket{0^{k(n)}}$, output the first $m(n)\cdot s(n)$-qubits $\rho(\Pi)_{\cR[1],...,\cR[s(n)]}$ of $C\ket{0^{k(n)}}$.\footnote{In quantum computing, the standard model of computation is the quantum circuit model rather than quantum Turing machine. Therefore, we define $\mathsf{QU}^{t(n)}(1^n,\Pi)$ using quantum circuits instead of using a quantum Turing machine.}  
\end{itemize}
\end{description}

We denote $\rho(\cD(1^n))$ to mean the output of $\cD(1^n)$. We often omit the security parameter $n$ of $\rho(\cD(1^n))$, writing it as $\rho(\cD)$.
For each $\Pi\in\bit^*$, we denote $\rho(\Pi)_{\cR[1],...,\cR[s(n)]}$ to mean the output of $\widetilde{\mathsf{QU}}_{s,m}^{t(n)}(1^n,\Pi)$.
We define $\rho(\Pi)_{\cR[i]} $ as follows:
\begin{align*}
\rho(\Pi)_{\cR[i]}\seteq \mathsf{Tr}_{\cR[1],...,\cR[i-1],\cR[i+1],...,\cR[s(n)]}\left( \rho(\Pi)_{\cR[1],...,\cR[s(n)]} \right).
\end{align*}

\matthew{TODO for me: Add some explanation here about what $\cS$ is.} 
Consider the following $\cS(1^n,a)$: interpret advice $a$ as a pair of a sampler $\Pi$, and a register $i$ and output $\rho(\Pi)_{\cR[i]}$. Applying observation~\ref{ob:EFI} to that $\cS$ (and hard-coding in $1^n$ for simplicity), we know there exists a QPT algorithm $\mathsf{Dis}$ such that the followings are satisfied: 
$\mathsf{Dis}$ takes $\Pi\in\bit^{\leq \log(n)}$, $i\in[s(n)]$, and $m(n)$-qubits $\rho$ as input, outputs $b\in\bit$, and satisfies the following properties for all $\Pi\in\bit^{\leq \log(n)}$ and $i\in[s(n)]$:
If 
\begin{align*}
    \mathsf{TD}(\rho(\Pi)_{\cR[i]},\rho(\cD))\geq \alpha(n)
\end{align*}
for some function $\alpha(n)>2\epsilon(n)\cdot n^{-100c}$, then we have
\begin{align*}
\Pr[1\la\mathsf{Dis}(\Pi,i,\rho(\Pi)_{\cR[i]})]-\Pr[1\la\mathsf{Dis}(\Pi,i,\rho(\cD)) ]\geq \mathsf{TD}(\rho(\Pi)_{\cR[i]},\rho(\cD))-\epsilon(n)\cdot n^{-100c}
\end{align*}
for all sufficiently large $n\in\N$.
We set $s(n)=\frac{8n^{200c+10}}{\epsilon(n)}$.

\matthew{[For Camera Ready] It would be great is we could either drop the $n^{-100c}$ or make it something smaller.}

\matthew{[For Camera Ready] It seems like it would make (some of) the analysis simpler if we just had the verification algorithm run each $\Pi$ $s(n)/n$ times. There's plenty of straightforward ways to do this algorithmically. This might make other parts of the analysis harder if we are assuming independence in certain parts of the analysis.}

\paragraph{Description of the Verification Algorithm.}
We consider the following algorithm $\Vrfy$:
\begin{description}
\item[$\Vrfy\left(1^n,\rho_{\cR[1],...,\cR[s(n)]} \right)$:]$ $
\begin{enumerate}
\item For all $i\in[s(n)]$, sample $\Pi_i\la\bit^{\leq \log(n)}$, and run $A_{\Pi,i}\la\mathsf{Dis}(\Pi_i,i,\rho_{\cR[i]})$. For each $ \Pi \in \bit^{\leq \log(n)} $, let $ \mathsf{Count}[\Pi] := |\{ i \in [s(n)] : \Pi_i = \Pi \}| $ denote the number of times $ \Pi $ is selected.
\item For each \( i \in [s(n)] \) and $\Pi\in\bit^{\leq \log(n)}$, run $\rho(\cD)\la\cD(1^n)$, and run \(B_{\Pi,i}\la\mathsf{Dis}(\Pi_i,i,\rho(\cD)) \).
\item If there exists $\Pi\in\bit^{\leq \log(n)}$ such that
\begin{align*}
\frac{1}{\mathsf{Count}[\Pi]}\left(\sum_{i \in [s(n)] : \Pi_i = \Pi} A_{\Pi,i}-\sum_{i\in[s(n)]:\Pi_i=\Pi}B_{\Pi,i}\right)\geq n^{-100c},
\end{align*}
then output $\bot$.
Otherwise, output $\top$.
\end{enumerate}
\end{description}
\end{proof}
\paragraph{Correctness.}
In the following, we show that
\begin{align*}
\Pr[\top\la\Vrfy(1^n,\rho(\cD)^{\otimes s(n)})]\geq 1-\negl(n).
\end{align*}

For showing this, it is sufficient to show the following Claim~\ref{claim:union_s}.
\begin{claim}\label{claim:union_s}
With probability at least $1-\negl(n)$, we have
\begin{align*}
\frac{1}{\mathsf{Count}[\Pi]}\left(\sum_{i:\Pi_i=\Pi}A_{\Pi,i}-\sum_{i:\Pi_i=\Pi}B_{\Pi,i}\right)\leq n^{-100c}
\end{align*}
for all $\Pi\in\bit^{\leq \log(n)}$.
Here, the probability is taken over $\{\Pi_i\la\bit^{\leq \log(n)}\}_{i\in[s(n)]}$, $\{A_{\Pi_i,i}\la \mathsf{Dis}_{\Pi_i,i} (\rho(\cD))\}_{i\in[s(n)]}$ and $\{B_{\Pi_i,i}\la \mathsf{Dis}_{\Pi_i,i} (\rho(\cD))\}_{i\in[s(n)]}$.
\end{claim}
Claim~\ref{claim:union_s} follows from union bound, Chernoff bound (Lemma~\ref{Chernoff}). and Hoeffding's inequality (Lemma~\ref{Hoeffding}).
For clarity, we describe the proof in the following. For showing this, we use Claim~\ref{claim:Chernoff_s}.
\begin{claim}\label{claim:Chernoff_s}
With probability at least $1-\negl(n)$ over $\{\Pi_i\la\bit^{\leq \log(n)}\}_{i\in[s(n)]}$, we have
\begin{align*}
\mathsf{Count}[\Pi]\geq \frac{s(n)}{4n}
\end{align*}
for all $\Pi\in\bit^{\leq \log(n)}$ and all $n\in\N$.
\end{claim}
Claim~\ref{claim:Chernoff_s} follows from Chernoff bound (Lemma~\ref{Chernoff}).
Once we have shown Claim~\ref{claim:Chernoff_s}, Claim~\ref{claim:union_s} follows from union bound and Hoeffding's inequality (Lemma~\ref{Hoeffding}).

\begin{proof}[Proof of Claim~\ref{claim:union_s}]
Suppose that
$
\mathsf{Count}[\Pi]\geq \frac{s(n)}{4n}
$
for all $\Pi\in\bit^{\leq \log(n)}$.

If this occurs, then from Hoeffding's inequality (Lemma~\ref{Hoeffding}), for each $\Pi\in\bit^{\leq \log(n)}$, we have
\begin{align*}
&\Pr_{\{A_{\Pi,i}\la\mathsf{Dis}(\Pi_i,i,\rho(\cD))\}_{i\in[s(n)]}}\left[\frac{1}{\mathsf{Count}[\Pi]}\left(\sum_{i\in[s(n)]:\Pi_i=\Pi}A_{\Pi,i}-\sum_{i\in[s(n)]:\Pi_i=\Pi}\Pr[1\la\mathsf{Dis}(\Pi,i,\rho(\cD))]\right)\geq \frac{n^{-100c}}{2}\right]\\
&\hspace{15cm}\leq 2^{-\frac{s(n)}{8n^{200c+1}}}
\end{align*}
and
\begin{align*}
&\Pr_{\{B_{\Pi,i}\la\mathsf{Dis}(\Pi_i,i,\rho(\cD))\}_{i\in[s(n)]}}\left[\frac{1}{\mathsf{Count}[\Pi]}\left(\sum_{i\in[s(n)]:\Pi_i=\Pi}\Pr[1\la\mathsf{Dis}(\Pi,i,\rho(\cD))]-\sum_{i\in[s(n)]:\Pi_i=\Pi}B_{\Pi,i}\right)\geq \frac{n^{-100c}}{2}\right]\\
&\hspace{15cm}\leq 2^{-\frac{s(n)}{8n^{200c+1}}}
\end{align*}
for all $n\in\N$.

If $\mathsf{Count}[\Pi]\geq \frac{s(n)}{4n}$ for all $\Pi\in\bit^{\leq \log(n)}$, then, for each $\Pi \in\bit^{\leq \log(n)}$, by union bounding over the previous two inequalities we have
\begin{align*}
\Pr_{\substack{
\{A_{\Pi,i}\la \mathsf{Dis}(\Pi_i,i,\rho(\cD))\}_{i\in[s(n)]}\\
\{B_{\Pi,i}\la \mathsf{Dis}(\Pi_i,i,\rho(\cD))\}_{i\in[s(n)]}
}}
\left[\frac{1}{s(n)}\left(\sum_{i \in [s(n)]} A_{\Pi,i}(x_i)-\sum_{i\in[s(n)]}B_{\Pi,i}(x_i^*)\right)<n^{-100c}\right]\geq 1-2\cdot 2^{-\frac{s(n)}{8n^{200c+1}}}
\end{align*}
for all $n\in\N$.

When $\mathsf{Count}[\Pi]\geq \frac{s(n)}{4n}$, with probability at least $1-\negl(n)$ over $\{A_{\Pi,i}\la \mathsf{Dis}(\Pi_i,i,\rho(\cD))\}_{i\in[s(n)]}$ and $\{B_{\Pi,i}\la \mathsf{Dis}(\Pi_i,i,\rho(\cD))\}_{i\in[s(n)]}$
\begin{align*}
\frac{1}{s(n)}\left(\sum_{i \in [s(n)]} A_{\Pi,i}(x_i)-\sum_{i\in[s(n)]}B_{\Pi,i}(x_i^*)\right)<n^{-100c}
\end{align*}
for all $\Pi\in\bit^{\leq \log(n)}$ .
From Claim~\ref{claim:Chernoff_s}, $\mathsf{Count}[\Pi]\geq \frac{s(n)}{4(n)}$ with probability at least $1-\negl(n)$, and hence Claim~\ref{claim:union_s} follows.
\end{proof}

\begin{proof}[Proof of Claim~\ref{claim:Chernoff_s}]    
First, we note that
\begin{align*}
\mathbb{E}_{\{\Pi_i\}_{i\in[s(n)]}\la\bit^{\leq \log(n)}}[\mathsf{Count}[\Pi]]=\frac{s(n)}{2n-2}
\end{align*}
for each $\Pi\in\bit^{\leq\log(n)}$.
From Chernoff bound (Lemma~\ref{Chernoff}), for any $0<\delta<1$, we have
\begin{align*}
\Pr_{\{\Pi_i\}_{i\in[s(n)]}\la\bit^{\leq \log(n)}}\left[\mathsf{Count}[\Pi]\leq (1-\delta)\frac{s(n)}{2n-2}\right]\leq 2^{\left(-\frac{\delta^2s(n)}{6(n-1)}\right)}
\end{align*}
for all $n\in\N$.
If we set $\delta=\frac{1}{2}$, then we have
\begin{align*}
\Pr_{\{\Pi_i\}_{i\in[s(n)]}\la\bit^{\leq \log(n)}}\left[\mathsf{Count}[\Pi]\leq \frac{s(n)}{4n}\right]\leq 2^{\left(-\frac{s(n)}{24n}\right)}
\end{align*}
for all $n\in\N$.
This completes the proof.
\end{proof}

\paragraph{Adaptive-Soundness.}

In the following, we show that for any $t(n)$-time quantum algorithm $\cA$ such that
\begin{align*}
    \mathsf{TD}(\mathsf{Marginal}_{\cA}(1^n),\cD(1^n))\geq \epsilon(n),
\end{align*}
we have
\begin{align*}
\Pr[\Vrfy(\rho(\Pi_\cA)_{\cR[1],...,\cR[s(n)]}) :\rho(\Pi_\cA)_{\cR[1],...,\cR[s(n)]}\la\cA(1^n)]\leq 1-\epsilon(n)+n^{-c}
\end{align*}
for all sufficiently large $n\in\N$.

For showing adaptive soundness, it is sufficient to show Claim~\ref{claim:Markov_s}.
\begin{claim}\label{claim:Markov_s}
For all sufficiently large $n\in\N$, there exists $\Pi_{\cA}\in\bit^{\leq \log(n)}$ such that, with probability at least 
$$\epsilon(n)-n^{-c},$$ 
we have
\begin{align*}
\frac{1}{\mathsf{Count}[\Pi_\cA]}\sum_{i\in[s(n)]:\Pi_i=\Pi_{\cA}}\left( A_{\Pi_\cA,i}-B_{\Pi_\cA,i}\right)\geq n^{-100c},
\end{align*}
where $\mathsf{Count}[\Pi_\cA]\seteq \{i\in[s(n)]:\Pi_i=\Pi_\cA\}$ is the number of $i\in[s(n)]$ such that $\Pi_i=\Pi_\cA$.
Here, the probability is taken over $\{\Pi_i\la\bit^{\leq \log(n)}\}_{i\in[s(n)]}$, $\{A_{\Pi_\cA,i}\la \mathsf{Dis}(\Pi_\cA,i,\rho(\Pi_\cA)_{\cR[i]})\}_{i\in[s(n)]}$ and $\{B_{\Pi_\cA,i}\la \mathsf{Dis}(\Pi_\cA,i,\rho(\cD))\}_{i\in[s(n)]}$.
\end{claim}
Once we have obtained Claim~\ref{claim:Markov_s}, adaptive soundness directly follows.
Recall that from the construction of $\vrfy$, if there exists a $\Pi\in\bit^{\leq \log(n)}$ such that
\begin{align*}
    \frac{1}{s(n)}\sum_{i\in[s(n)]}(A_{\Pi,i}-B_{\Pi,i})\geq n^{-100c},
\end{align*}
then $\Vrfy$ outputs $\bot$.
Therefore, Claim~\ref{claim:Markov_s} implies that 
\begin{align*}
\Pr[\bot\la\Vrfy(\rho(\Pi_\cA)_{\cR[1],...,\cR[s(n)]}):\rho(\Pi_\cA)_{\cR[1],...,\cR[s(n)]}\la\cA(1^n) ]\geq \epsilon(n)-n^{-c}
\end{align*}
for all sufficiently large $n\in\N$.
This implies that
\begin{align*}
\Pr[\top\la\Vrfy(\rho(\Pi_\cA)_{\cR[1],...,\cR[s(n)]}):\rho(\Pi_\cA)_{\cR[1],...,\cR[s(n)]}\la\cA(1^n) ]\leq 1-\epsilon(n)+n^{-c}
\end{align*}
for all sufficiently large $n\in\N$, which completes the proof.

The remaining part of the proof is to show Claim~\ref{claim:Markov_s}.
For this, we use the following Claim~\ref{claim:measurement_s}.
\begin{claim}\label{claim:measurement_s}
For all sufficiently large $n\in\N$, there exists $\Pi_\cA\in\bit^{\leq \log(n)}$ such that,
with probability at least $1-\negl(n)$,
we have
\begin{align*}
&\frac{1}{\mathsf{Count}[\Pi_\cA]}\left(\sum_{i\in[s(n)]:\Pi_i=\Pi_\cA} \Pr[1\la\mathsf{Dis}(\Pi_\cA,i,\rho(\Pi_\cA)_{\cR[i]})]-\sum_{i\in[s(n)]:\Pi_i=\Pi_\cA}\Pr[1\la\mathsf{Dis}(\Pi_\cA,i,\rho(\cD))]\right)\\
&\geq \epsilon(n)- 4\epsilon(n)\cdot n^{-50c},
\end{align*}
where $\mathsf{Count}[\Pi_\cA]\seteq \{i\in[s(n)]:\Pi_i=\Pi_\cA\}$ is the number of $i\in[s(n)]$ such that $\Pi_i=\Pi_\cA$.
Here, the probability is taken over $\{\Pi_i\la \bit^{\leq \log(n)}\}_{i\in[s(n)]}$.
\end{claim}
For showing Claim~\ref{claim:measurement_s}, we use the following Claim~\ref{claim:triangle_s}.
\begin{claim}\label{claim:triangle_s}
We have
\begin{align*}
&\Pr_{\{\Pi_i\}_{i\in[s(n)]}\la \bit^{\leq \log(n)}}\left[\frac{1}{\mathsf{Count}[\Pi_\cA]}\sum_{i\in[s(n)]: \Pi_i=\Pi_\cA}\mathsf{TD}(\rho(\Pi_\cA)_{\cR[i]},\rho(\cD))\geq\epsilon(n) - 2\epsilon(n)\cdot n^{-50c}\right]\\
&\geq 1-\negl(n)
\end{align*}
for all sufficiently large $n\in\N$, where $\mathsf{Count}[\Pi_\cA]\seteq \{i\in[s(n)]:\Pi_i=\Pi_\cA\}$ is the number of $i\in[s(n)]$ such that $\Pi_i=\Pi_\cA$.
\end{claim}

\begin{proof}[Proof of Claim~\ref{claim:Markov_s}]
From Claim~\ref{claim:measurement_s}, for all sufficiently large $n\in\N$, there exists $\Pi_\cA\in\bit^{\leq \log(n)}$ such that, with probability at least $1-\negl(n)$ over $\{\Pi_i\}_{i\in[s(n)]}\la\bit^{\leq \log(n)}$, we have
\begin{align*}
&\frac{1}{\mathsf{Count}[\Pi_\cA]}\left(\sum_{i\in[s(n)]:\Pi_i=\Pi_\cA} \Pr[1\la\mathsf{Dis}(\Pi_\cA,i,\rho(\Pi_\cA)_{\cR[i]})]-\sum_{i\in[s(n)]:\Pi_i=\Pi_\cA}\Pr[1\la\mathsf{Dis}(\Pi_\cA,i,\rho(\cD))]\right)\\
&\geq \epsilon(n)-4\epsilon(n)\cdot n^{-50c}.
\end{align*}

We denote $X$ to be the random variable \matthew{should the $1/s(n)$ be $1/\mathsf{Count}[\Pi_\mathcal{A}]$?}
\begin{align*}
\frac{1}{\mathsf{Count}[\Pi_\mathcal{A}]} \sum_{i\in[s(n)]:\Pi_i=\Pi_\cA}(A_{\Pi_\cA,i}-B_{\Pi_\cA,i}).
\end{align*}
Then, from union bound, we have
\begin{align*}
    \mathop{\mathbb{E}}_{\substack{\{\Pi_i\}_{i\in[s(n)]}\la\bit^{\leq \log(n)}\\ A_{\Pi_\cA,i}\la\mathsf{Dis}(\Pi_\cA,i,\rho(\Pi_\cA)_{\cR[i]})\\
    B_{\Pi_\cA,i}\la\mathsf{Dis}(\Pi_\cA,i,\rho(\cD))}}[X]\geq \epsilon(n)-4\epsilon(n)\cdot n^{-50c}-\negl(n)
\end{align*}
for all sufficiently large $n\in\N$.
We have
\begin{align*}
\epsilon(n)-4\epsilon(n)\cdot n^{-50c}-\negl(n)
&\leq \mathbb{E}[X]\leq n^{-100c}\Pr[X< n^{-100c}]+\Pr[ X \geq n^{-100c}]\\
&=n^{-100c}+(1-n^{-100c})\Pr[X \geq n^{-100c}].
\end{align*}
This implies that
\begin{align*}
  \Pr[X \geq n^{-100c}]\geq   \frac{\epsilon(n)-4\epsilon(n)\cdot n^{-50c}-n^{-100c}-\negl(n)}{1-n^{-100c}}\geq\epsilon(n)-n^{-c} 
\end{align*}
for all sufficiently large $n\in\N$.
\end{proof}

\begin{proof}[Proof of Claim~\ref{claim:measurement_s}]
Let us denote $\cT$  
\begin{align*}
    \cT\seteq \{i:\mathsf{TD}(\rho(\Pi)_{\cR[i]},\rho(\cD) )\geq 2\epsilon(n)\cdot n^{-100c}\}.
\end{align*}
From the definition of $\mathsf{Dis}$, for all $i\in\cT$, we have
\begin{align*}
\Pr[1\la\mathsf{Dis}(\Pi_\cA,i,\rho(\Pi_\cA)_{\cR[i]})]-\Pr[1\la\mathsf{Dis}(\Pi_\cA,i,\rho(\cD))]\geq \mathsf{TD}(\rho(\Pi)_{\cR[i]},\rho(\cD))-\epsilon(n)\cdot n^{-100c}
\end{align*}
for all sufficiently large $n\in\N$.
From the definition of $\mathsf{Dis}$ and $\cT$, we have
\begin{align*}
&\frac{1}{\mathsf{Count}[\Pi_\cA]}\sum_{i\in[s(n)]:\Pi_i=\Pi_\cA}\left(\Pr[1\la\mathsf{Dis}(\Pi_\cA,i,\rho(\Pi_\cA)_{\cR[i]})]-\Pr[1\la\mathsf{Dis}(\Pi_\cA,i,\rho(\cD))]\right)\\
&\geq \frac{1}{\mathsf{Count}[\Pi_\cA]}\sum_{i\in\cT:\Pi_i=\Pi_\cA}\left(\mathsf{TD}(\rho(\Pi_\cA)_{\cR[i]},\rho(\cD))-\epsilon(n)\cdot n^{-100c}\right)\\
&\geq \frac{1}{\mathsf{Count}[\Pi_\cA]}\sum_{i\in\cT:\Pi_i=\Pi_\cA}\mathsf{TD}(\rho(\Pi_\cA)_{\cR[i]},\rho(\cD))-\epsilon(n)\cdot n^{-100c}
\end{align*}
for all sufficiently large $n\in\N$.

Furthermore, with probability at least $1-\negl(n)$ over $\{\Pi_i \la\bit^{\log(n)}\}_{i\in[s(n)]}$ ,
\begin{align*}
&\frac{1}{\mathsf{Count}[\Pi_\cA]}\sum_{i\in\cT:\Pi_i=\Pi_\cA}\mathsf{TD}(\rho(\Pi_\cA)_{\cR[i]},\rho(\cD))\\
&= \frac{1}{\mathsf{Count}[\Pi_\cA]}\left(\sum_{i\in[s(n)]:\Pi_i=\Pi_\cA}\mathsf{TD}(\rho(\Pi_\cA)_{\cR[i]},\rho(\cD))-\sum_{i\in [s(n)]\backslash\cT:\Pi_i=\Pi_\cA}\mathsf{TD}(\rho(\Pi_\cA)_{\cR[i]},\rho(\cD))\right)\\
&\geq \epsilon(n)-2\epsilon(n)n^{-50c}-\frac{1}{\mathsf{Count}[\Pi_\cA]}\sum_{i\in [s(n)]\backslash\cT:\Pi_i=\Pi_\cA} 2\epsilon(n)\cdot n^{-100c}\geq \epsilon(n)-4\epsilon(n)\cdot n^{-50c},
\end{align*}
where the first inequality follows from Claim~\ref{claim:triangle_s}\matthew{and the second inequality follows from Claim 6.3}. \taiga{Here, we do not use Claim 6.3. We use $ \sum_{i\in [s(n)]\backslash\cT:\Pi_i=\Pi_\cA} 1 \leq \mathsf{Count}[\Pi_\cA]$}
This completes the proof.
\matthew{I think it would be good to 1) say exactly what value of $\delta$ we are using for Claim 6.6; and 2) add in one extra step between the second and final equation where we have $8\epsilon(n) \cdot n^{-100c + 1}$ i.e.  $|\{s[n] \setminus  \cT \}|/\mathsf{Count[\cA]} \cdot 2\epsilon(n)\cdot n^{-100c}$ before we over estimate it with $2\epsilon(n) - 4\epsilon(n)\cdot n^{-50c}$}
\taiga{For (1), I agree with you. I slightly modify the statement of Claim 6.6 so that it is easier to follow.}
\taiga{For (2), I am not sure that I understand what you are caring about. The last inequality directly follows once one understands $\sum_{i\in [s(n)]\backslash\cT:\Pi_i=\Pi_\cA}$, (while this might be unclear.)}
\end{proof}

\begin{proof}[Proof of Claim~\ref{claim:triangle_s}]
Because $\cA$ is a uniform algorithm, for all sufficiently large $n\in\N$, there exists a $\Pi_\cA\in\bit^{\leq \log(n)}$ such that $\mathsf{QU}^{t(n)}(1^n,\Pi_A)$ is statistically equivalent to $\cA(1^n)$.

First, from the triangle inequality, we have
\begin{align*}
\epsilon(n)&\leq \mathsf{TD}(\mathsf{Marginal}_{\cA}(1^n),\rho(\cD))\\
&= \frac{1}{2}\norm{\sum_{i\in[s(n)]}\frac{1}{s(n)}\rho(\Pi_{\cA})_{\cR[i]}-\rho(\cD)}_1\\
&\leq \frac{1}{2s(n)}\sum_{i\in[s(n)]}\norm{\rho(\Pi_\cA)_{\cR[i]}-\rho(\cD)}_1=\frac{1}{s(n)}\sum_{i\in[s(n)]}\mathsf{TD}(\rho(\Pi_\cA)_{\cR[i]},\rho(\cD)).
\end{align*}

Let $b_i\in\bit$ be a random variable such that $b_i=1$ if and only if $\Pi_i=\Pi_\cA$.
With probability $\frac{1}{2(n-1)}$ over $\Pi_i\la\bit^{\leq \log(n)}$, we have $\Pi_i=\Pi_\cA$. Therefore, it holds that 
\begin{align}
\mathbb{E}_{\{\Pi_i\}_{i\in[s(n)]}\la \bit^{\leq \log(n)}}\left[\sum_{i\in[s(n)]} b_i\right]=\frac{s(n)}{2(n-1)}\label{avg_1}.
\end{align}
From Chernoff bound (Lemma~\ref{Chernoff}), for any $\delta>0$, we have
\begin{align}
\Pr_{\{\Pi_i\}_{i\in[s(n)]}\la \bit^{\leq \log(n)}}\left[\sum_{i\in[s(n)]} b_i \geq (1+\delta) \mathbb{E}[\sum_{i\in[s(n)]} b_i] \right]
&\leq\exp\left(-\frac{\delta^2\mathbb{E}[\sum_{i\in s[(n)]}b_i]}{3}\right)\notag\\
&\leq\exp\left(-\frac{\delta^2\cdot s(n)}{6n}\right)\label{ineq_1}
\end{align}

On the other hand, we have
\begin{align}
\mathbb{E}_{\{\Pi_i\}_{i\in[s(n)]}\la\bit^{\leq \log(n)}}\left[\sum_{i\in[s(n)]}b_i\mathsf{TD}(\rho(\Pi_\cA)_{\cR[i]},\rho(\cD))\right]\notag
&=\sum_{i\in[s(n)]}\mathbb{E}[b_i]\mathsf{TD}(\rho(\Pi_\cA)_{\cR[i]},\rho(\cD))\\
&\geq \frac{1}{2(n-1)}\left(s(n)\cdot\epsilon(n)\right)\label{avg_2}.
\end{align}
From Chernoff bound (Lemma~\ref{Chernoff}), for any $\delta>0$, we have
\begin{align}
&\Pr_{\{\Pi_i\}_{i\in[s(n)]}\la\bit^{\leq\log(n)} }\left[\sum_{i\in[s(n)]}b_i\mathsf{TD}(\rho(\Pi_\cA)_{\cR[i]},\rho(\cD)) < (1-\delta) \mathbb{E}\left[\sum_{i\in[s(n)]}b_i\mathsf{TD}(\rho(\Pi_\cA)_{\cR[i]},\rho(\cD))\right]\right]\notag\\
&\leq\exp\left(\frac{-\delta^2}{2}\mathbb{E}\left[\sum_{i\in[s(n)]}b_i\mathsf{TD}(\rho(\Pi_\cA)_{\cR[i]},\rho(\cD))\right] \right)
\leq \exp\left(\frac{-\delta^2\cdot s(n)\epsilon(n)}{4n} \right)\label{ineq_2}
\end{align}

From \cref{ineq_1,ineq_2,avg_1,avg_2}, with probability at least
\begin{align*}
    1-\exp\left(\frac{-\delta^2\cdot s(n)\epsilon(n)}{4n} \right)-\exp\left(-\frac{\delta^2\cdot s(n)}{6n}\right)
\end{align*}
over $\{\Pi_i\la\bit^{\leq \log(n)}\}_{i\in[s(n)]}$,
we have
\begin{align*}
&\frac{1}{\mathsf{Count}[\Pi_\cA]}\sum_{i\in[s(n)]:\Pi_i=\Pi_\cA}\mathsf{TD}(\rho(\Pi_\cA)_{\cR[i]},\rho(\cD))\\
&=\frac{1}{\sum_{i\in[s(n)]} b_i}\sum_{i\in[s(n)]} b_i\mathsf{TD}(\rho(\Pi_\cA)_{\cR[i]},\rho(\cD))\\ 
&\geq \frac{(1-\delta)\mathbb{E}[\sum_{i\in[s(n)]}b_i\mathsf{TD}(\rho(\Pi_\cA)_{\cR[i]},\rho(\cD))] }{(1+\delta)\mathbb{E}[\sum_{i\in[s(n)]}b_i]}\geq \frac{1-\delta}{1+\delta}\epsilon(n).
\end{align*}
By setting $\delta=n^{-50c}$, with probability at least $1-\negl(n)$ over $\{\Pi_i\la \bit^{\leq \log(n)}\}_{i\in[s(n)]}$, we have
\begin{align*}
\frac{1}{\mathsf{Count}[\Pi_\cA]}\sum_{i\in[s(n)]:\Pi_i=\Pi_\cA}\mathsf{TD}(\rho(\Pi_\cA)_{\cR[i]},\rho(\cD))\geq \epsilon(n)-2n^{-50c}\epsilon(n).
\end{align*}
\end{proof}

\if0\taiga{kokokara}

In the following, we consider that an unknown algorithm $\cA$ satisfies the following condition:
\begin{align*}
    \mathsf{TD}(\mathsf{Marginal}_{\cA}(1^n),\cD(1^n))\geq \epsilon(n).
\end{align*}
For any $t(n)$-time quantum algorithm $\cA$, there exists a $n_0\in\N$ such that, for all $n\geq n_0$, there exists $\Pi\in\bit^{\leq \log(n)}$ such that $\widetilde{\mathsf{QU}}^{t(n)}(1^n,\Pi)$ is statistically equivalent to $\cA(1^n)$ because $\cA$ is a uniform quantum algorithm. We denote such $\Pi$ as $\Pi_\cA$.

For showing adaptive soundness, we use the following \cref{claim:Markov_state}.
\begin{claim}\label{claim:Markov_state}
With probability at least $\epsilon(n)-4\epsilon(n)\cdot n^{-50c}-n^{-100c}-\negl(n)$ over $\{\Pi_i\}_{i\in[s(n)]}\la\bit^{\leq \log(n)}$ and over the randomness $A_{\Pi_\cA,i}\la \mathsf{Dis}(\Pi_\cA,i,\rho(\Pi_\cA)_{\cR[i]})$ and $B_{\Pi_\cA,i}\la \mathsf{Dis}(\Pi_\cA,i,\rho(\cD))$, we have
\begin{align*}
\frac{1}{\mathsf{Count}[\Pi_\cA]}\sum_{i\in[s(n)]:\Pi_i=\Pi_{\cA}}\left( A_{\Pi_\cA,i}-B_{\Pi_\cA,i}\right)\geq n^{-100c},
\end{align*}
where $\mathsf{Count}[\Pi_\cA]\seteq \{i\in[s(n)]:\Pi_i=\Pi_\cA\}$ is the number of $i\in[s(n)]$ such that $\Pi_i=\Pi_\cA$.
\end{claim}

For showing \cref{claim:Markov_state}, we use the following \cref{claim:measurement}.
\begin{claim}\label{claim:measurement}
For any $\delta>0$, with probability at least $1-\exp\left(\frac{-\delta^2\cdot s(n)\epsilon(n)}{4n}\right)-\exp\left(-\frac{\delta^2\cdot s(n)}{6n}\right)$ over $\{\Pi_i\}_{i\in[s(n)]}\la \bit^{\leq \log(n)}$, we have
\begin{align*}
&\frac{1}{\mathsf{Count}[\Pi_\cA]}\left(\sum_{i\in[s(n)]:\Pi_i=\Pi_\cA} \Pr[1\la\mathsf{Dis}(\Pi_\cA,i,\rho(\Pi_\cA)_{\cR[i]})]-\sum_{i\in[s(n)]:\Pi_i=\Pi_\cA}\Pr[1\la\mathsf{Dis}(\Pi_\cA,i,\rho(\cD))]\right)\\
&\geq \frac{1-\delta}{1+\delta}\epsilon(n)- 2\epsilon(n)\cdot n^{-100c},
\end{align*}
where $\mathsf{Count}[\Pi_\cA]\seteq \{i\in[s(n)]:\Pi_i=\Pi_\cA\}$ is the number of $i\in[s(n)]$ such that $\Pi_i=\Pi_\cA$.
\end{claim}
For showing \cref{claim:measurement}, we use the following \cref{claim:triangle_state}.
\begin{claim}\label{claim:triangle_state}
For any $\delta>0$,
\begin{align*}
&\Pr_{\{\Pi_i\}_{i\in[s(n)]}\la \bit^{\leq \log(n)}}\left[\frac{1}{\mathsf{Count}[\Pi_\cA]}\sum_{i\in[s(n)]: \Pi_i=\Pi_\cA}\mathsf{TD}(\rho(\Pi_\cA)_{\cR[i]},\rho(\cD))\geq\frac{1-\delta}{1+\delta} \epsilon(n)\right]\\
&\geq 1-\exp\left(\frac{-\delta^2\cdot s(n)\epsilon(n)}{4n}\right)-\exp\left(-\frac{\delta^2\cdot s(n)}{6n}\right)
\end{align*}
for all $n\in\N$, where $\mathsf{Count}[\Pi_\cA]\seteq \{i\in[s(n)]:\Pi_i=\Pi_\cA\}$ is the number of $i\in[s(n)]$ such that $\Pi_i=\Pi_\cA$.
\end{claim}

Once we have obtained \cref{claim:Markov_state}, adaptive soundness directly follows.
Recall that from the construction of $\vrfy$, if there exists a $\Pi\in\bit^{\leq \log(n)}$ such that
\begin{align*}
    \frac{1}{s(n)}\sum_{i\in[s(n)]}(A_{\Pi,i}-B_{\Pi,i})\geq n^{-100c},
\end{align*}
then $\Vrfy$ outputs $\bot$.
Therefore, \cref{claim:Markov_state} implies that 
\begin{align*}
\Pr[\bot\la\Vrfy(\rho_\cA)]\geq \epsilon(n)-4\epsilon(n)\cdot n^{-50c}-\negl(n)
\end{align*}
for all sufficiently large $n\in\N$.
This implies that
\begin{align*}
\Pr[\top\la\Vrfy(\rho_\cA)]\leq 1-\epsilon(n)+4\epsilon(n)\cdot n^{-50c}+n^{-100c}+\negl(n)\cdot n^{-50c}\leq 1-\epsilon(n)+n^{-c}
\end{align*}
for all sufficiently large $n\in\N$, which completes the proof.

\begin{proof}[Proof of \cref{claim:Markov_state}]
In the following, we set $\delta=n^{-50c}$ so that, with probability at least $1-\negl(n)$ over $\{\Pi_i\}_{i\in[s(n)]}\la\bit^{\leq \log(n)}$, we have
\begin{align*}
&\frac{1}{\mathsf{Count}[\Pi_\cA]}\left(\sum_{i\in[s(n)]:\Pi_i=\Pi_\cA} \Pr[1\la\mathsf{Dis}(\Pi_\cA,i,\rho(\Pi_\cA)_{\cR[i]})]-\sum_{i\in[s(n)]:\Pi_i=\Pi_\cA}\Pr[1\la\mathsf{Dis}(\Pi_\cA,i,\rho(\cD))]\right)\\
&\geq \frac{1-\delta}{1+\delta}\epsilon(n)- 2\epsilon(n)\cdot n^{-100c}\geq (1-2\delta)\epsilon(n)-2\epsilon(n)\cdot n^{-100c}\geq \epsilon(n)-4\epsilon(n)\cdot n^{-50c},
\end{align*}
which follows from \cref{claim:measurement}.

We denote $X$ to mean the random variable 
\begin{align*}
\frac{1}{s(n)} \sum_{i\in[s(n)]:\Pi_i=\Pi_\cA}(A_{\Pi_\cA,i}-B_{\Pi_\cA,i}).
\end{align*}
Then, we have
\begin{align*}
    \mathbb{E}_{\substack{\{\Pi_i\}_{i\in[s(n)]}\la\bit^{\leq \log(n)}\\ A_{\Pi_\cA,i}\la\mathsf{Dis}(\Pi_\cA,i,\rho(\Pi_\cA)_{\cR[i]})\\
    B_{\Pi_\cA,i}\la\mathsf{Dis}(\Pi_\cA,i,\rho(\cD))}}[X]\geq \epsilon(n)-4\epsilon(n)\cdot n^{-50c}-\negl(n)
\end{align*}
for all sufficiently large $n\in\N$.
We have
\begin{align*}
\epsilon(n)-4\epsilon(n)\cdot n^{-50c}-\negl(n)
&\leq \mathbb{E}[X]\leq n^{-100c}\Pr[X< n^{-100c}]+\Pr[n^{-100c}\leq X]\\
&=n^{-100c}+(1-n^{-100c})\Pr[n^{-100c}\leq X].
\end{align*}
This implies that
\begin{align*}
    \frac{\epsilon(n)-4\epsilon(n)\cdot n^{-50c}-n^{-100c}-\negl(n)}{1-n^{-100c}}\leq \Pr[n^{-100c}\leq X]
\end{align*}
for all sufficiently large $n\in\N$.
\end{proof}

\begin{proof}[Proof of \cref{claim:measurement}]
Let us denote $\cT$  
\begin{align*}
    \cT\seteq \{i:\mathsf{TD}(\rho(\Pi)_{\cR[i]},\rho(\cD) )\geq 2\epsilon(n)\cdot n^{-100c}\}.
\end{align*}
From the definition of $\mathsf{Dis}$, for all $i\in\cT$, we have
\begin{align*}
\Pr[1\la\mathsf{Dis}(\Pi_\cA,i,\rho(\Pi_\cA)_{\cR[i]})]-\Pr[1\la\mathsf{Dis}(\Pi_\cA,i,\rho(\cD))]\geq \mathsf{TD}(\rho(\Pi)_{\cR[i]},\rho(\cD))-\epsilon(n)\cdot n^{-100c}
\end{align*}
for all sufficiently large $n\in\N$.
From the definition of $\mathsf{Dis}$ and $\cT$, we have
\begin{align*}
&\frac{1}{\mathsf{Count}[\Pi_\cA]}\sum_{i\in[s(n)]:\Pi_i=\Pi_\cA}\left(\Pr[1\la\mathsf{Dis}(\Pi_\cA,i,\rho(\Pi_\cA)_{\cR[i]})]-\Pr[1\la\mathsf{Dis}(\Pi_\cA,i,\rho(\cD))]\right)\\
&\geq \frac{1}{\mathsf{Count}[\Pi_\cA]}\sum_{i\in\cT:\Pi_i=\Pi_\cA}\left(\mathsf{TD}(\rho(\Pi_\cA)_{\cR[i]},\rho(\cD))-\epsilon(n)\cdot n^{-100c}\right)\\
&\geq \frac{1}{\mathsf{Count}[\Pi_\cA]}\sum_{i\in\cT:\Pi_i=\Pi_\cA}\mathsf{TD}(\rho(\Pi_\cA)_{\cR[i]},\rho(\cD))-\epsilon(n)\cdot n^{-100c}
\end{align*}
for all sufficiently large $n\in\N$.

Furthermore, with probability at least $1-\exp\left(\frac{-\delta^2\cdot s(n)\epsilon(n)}{4n}\right)-\exp\left(-\frac{\delta^2\cdot s(n)}{6n}\right)$ over $\{\Pi_i\}_{i\in[s(n)]} \la\bit^{\log(n)}$,
\begin{align*}
&\frac{1}{\mathsf{Count}[\Pi_\cA]}\sum_{i\in\cT:\Pi_i=\Pi_\cA}\mathsf{TD}(\rho(\Pi_\cA)_{\cR[i]},\rho(\cD))\\
&= \frac{1}{\mathsf{Count}[\Pi_\cA]}\left(\sum_{i\in[s(n)]:\Pi_i=\Pi_\cA}\mathsf{TD}(\rho(\Pi_\cA)_{\cR[i]},\rho(\cD))-\sum_{i\in [s(n)]\backslash\cT:\Pi_i=\Pi_\cA}\mathsf{TD}(\rho(\Pi_\cA)_{\cR[i]},\rho(\cD))\right)\\
&\geq \frac{1-\delta}{1+\delta}\epsilon(n)-\frac{1}{\mathsf{Count}[\Pi_\cA]}\sum_{i\in [s(n)]\backslash\cT:\Pi_i=\Pi_\cA} 2\epsilon(n)\cdot n^{-100c}\geq \frac{1-\delta}{1+\delta}\epsilon(n)-2\epsilon(n)\cdot n^{-100c},
\end{align*}
where the first inequality follows from \cref{claim:triangle_state}.
This completes the proof.
\end{proof}

\begin{proof}[Proof of \cref{claim:triangle_state}]
First, from the triangle inequality, we have
\begin{align*}
\epsilon(n)&\leq \mathsf{TD}(\mathsf{Marginal}_{\cA}(1^n),\rho(\cD))\\
&= \frac{1}{2}\sum_{i\in[s(n)]}\norm{\frac{1}{s(n)}\rho(\Pi_{\cA})_{\cR[i]}-\rho(\cD)}_1\\
&\leq \frac{1}{2s(n)}\sum_{i\in[s(n)]}\norm{\rho(\Pi_\cA)_{\cR[i]}-\rho(\cD)}_1=\frac{1}{s(n)}\sum_{i\in[s(n)]}\mathsf{TD}(\rho(\Pi_\cA)_{\cR[i]},\rho(\cD)).
\end{align*}

Let $b_i\in\bit$ be a random variable such that $b_i=1$ if and only if $\Pi_i=\Pi_\cA$.
With probability $\frac{1}{2(n-1)}$ over $\Pi_i\la\bit^{\leq \log(n)}$, we have $\Pi_i=\Pi_\cA$. Therefore, it holds that 
\begin{align}
\mathbb{E}_{\{\Pi_i\}_{i\in[s(n)]}\la \bit^{\leq \log(n)}}\left[\sum_{i\in[s(n)]} b_i\right]=\frac{s(n)}{2(n-1)}\label{avg_1}.
\end{align}
From Chernoff bound, for any $\delta$, we have
\begin{align}
\Pr_{\{\Pi_i\}_{i\in[s(n)]}\la \bit^{\leq \log(n)}}\left[\sum_{i\in[s(n)]} b_i \geq (1+\delta) \mathbb{E}[\sum_{i\in[s(n)]} b_i] \right]
&\leq\exp\left(-\frac{\delta^2\mathbb{E}[\sum_{i\in s[(n)]}b_i]}{3}\right)\notag\\
&\leq\exp\left(-\frac{\delta^2\cdot s(n)}{6n}\right)\label{ineq_1}
\end{align}

On the other hand, we have
\begin{align}
\mathbb{E}_{\{\Pi_i\}_{i\in[s(n)]}\la\bit^{\leq \log(n)}}\left[\sum_{i\in[s(n)]}b_i\mathsf{TD}(\rho(\Pi_\cA)_{\cR[i]},\rho(\cD))\right]\notag
&=\sum_{i\in[s(n)]}\mathbb{E}[b_i]\mathsf{TD}(\rho(\Pi_\cA)_{\cR[i]},\rho(\cD))\\
&\geq \frac{1}{2(n-1)}\left(s(n)\cdot\epsilon(n)\right)\label{avg_2}.
\end{align}
From Chernoff bound, we have
\begin{align}
&\Pr_{\{\Pi_i\}_{i\in[s(n)]}\la\bit^{\leq\log(n)} }\left[\sum_{i\in[s(n)]}b_i\mathsf{TD}(\rho(\Pi_\cA)_{\cR[i]},\rho(\cD)) < (1-\delta) \mathbb{E}\left[\sum_{i\in[s(n)]}b_i\mathsf{TD}(\rho(\Pi_\cA)_{\cR[i]},\rho(\cD))\right]\right]\notag\\
&\leq\exp\left(\frac{-\delta^2}{2}\mathbb{E}\left[\sum_{i\in[s(n)]}b_i\mathsf{TD}(\rho(\Pi_\cA)_{\cR[i]},\rho(\cD))\right] \right)
\leq \exp\left(\frac{-\delta^2\cdot s(n)\epsilon(n)}{4n} \right)\label{ineq_2}
\end{align}

From \cref{ineq_1,ineq_2,avg_1,avg_2}, with probability at least
\begin{align*}
    1-\exp\left(\frac{-\delta^2\cdot s(n)\epsilon(n)}{4n} \right)-\exp\left(-\frac{\delta^2\cdot s(n)}{6n}\right)
\end{align*}
over $\Pi_i\la\bit^{\leq \log(n)}$,
we have
\begin{align*}
&\frac{1}{\mathsf{Count}[\Pi_\cA]}\sum_{i\in[s(n)]:\Pi_i=\Pi_\cA}\mathsf{TD}(\rho(\Pi_\cA)_{\cR[i]},\rho(\cD))\\
&=\frac{1}{\sum_{i\in[s(n)]} b_i}\sum_{i\in[s(n)]} b_i\mathsf{TD}(\rho(\Pi_\cA)_{\cR[i]},\rho(\cD))\\ 
&\geq \frac{(1-\delta)\mathbb{E}[\sum_{i\in[s(n)]}b_i\mathsf{TD}(\rho(\Pi_\cA)_{\cR[i]},\rho(\cD))] }{(1+\delta)\mathbb{E}[\sum_{i\in[s(n)]}b_i]}\geq \frac{1-\delta}{1+\delta}\epsilon(n).
\end{align*}
This completes the proof.

\taiga{kokomade}
\end{proof}
\fi

\subsection{Hardness of Verifying Quantum States from EFI}

\begin{proposition}[Restatement of \cref{informal:state_verification}]\label{thm:hardness_EFI}
    Suppose that infinitely-often EFI exists.
    Then, there exists a QPT algorithm $\cD$ with quantum state outputs, which is hard to verify in the sense of Definition~\ref{def:verify_state}.
\end{proposition}

\begin{proof}[Proof of \cref{thm:hardness_EFI}]
Let us consider an EFI $\Gen$.
We can see that $\Gen(1^n,0)$ is hard to verify.
From the definition of EFI, for any QPT algorithm $\Vrfy$ and any polynomial $s$, if we have
\begin{align*}
\Pr[\top\la\Vrfy(\rho_0^{\otimes s(n)}):\rho_0^{\otimes s(n)}\la\Gen(1^n,0)^{\otimes s(n)}]\geq 1-\negl_1(n),
\end{align*}
then we have
\begin{align*}
\Pr[\top\la\Vrfy(\rho_1^{\otimes s(n)}):\rho_1^{\otimes s(n)}\la\Gen(1^n,1)^{\otimes s(n)}]\geq 1-\negl_2(n),
\end{align*}
On the other hand, we have
\begin{align*}
\mathsf{TD}(\Gen(1^n,0),\Gen(1^n,1))\geq 1-\negl(n).
\end{align*}
This means that any QPT algorithm $\Vrfy$ cannot satisfy correctness and adaptive-soundness at the same time.
\end{proof}

\section{Quantum Distribution Verification without Kolmogorov Complexity}
\label{sec:ver-without-K}

\subsection{OWPuzzs from Hardness of Quantum Distribution Verification}
Our main theorem in this section is \cref{thm:OWPuzz_Verify}, which shows how to efficiently verify any quantum distributions assuming the non-existence of OWPuzzs.
Let us emphasize that \cref{thm:OWPuzz_Verify} supersedes \cref{thm:q_easiness} in that the underlying assumption is weaker.
In \cref{thm:q_easiness}, we show that $\mathbf{BQP}=\mathbf{PP}$ implies the easiness of verifying quantum distributions.
While $\mathbf{BQP}=\mathbf{PP}$ implies the non-existence of OWPuzzs, the other direction is open.

\begin{theorem}[Restatement of \cref{inf:q_verify}]\label{thm:OWPuzz_Verify}
Assume that infinitely-often OWPuzzs do not exist.
Then, any QPT algorithm $\cD$ is adaptively-verifiable with a QPT algorithm.
\end{theorem}

\begin{remark}
We prove \cref{thm:OWPuzz_Verify} based on the same underlying idea as in the proof of \cref{thm:state_verification}.
However, in the present setting the protocol and its analysis admit a simpler formulation, and we therefore include the proof for clarity.
Concretely, in the quantum state verification protocol, for each $i\in[s(n)]$, we randomly sample $\Pi_i\la\bit^{\leq \log(n)}$ and run $\mathsf{Dis}_{\Pi_i,i}(\rho_{\cR_i})$.
This is because we might not be able to run $\mathsf{Dis}_{\Pi,i}(\rho_{\cR_i})$ for all $\Pi\in\bit^{\leq \log(n)}$ since quantum measurements may collapse the state $\rho_{\cR_i}$.
For distributions, such per-index sampling of $\Pi_i$ is unnecessary, which leads to a simpler analysis.
\end{remark}

For showing \cref{thm:OWPuzz_Verify}, we will use \cref{thm:owpuzz}.
We provide the proof of \cref{thm:owpuzz} in \cref{sec:owpuzz_non_uniform_QEFID}.
\begin{theorem}\label{thm:owpuzz}
Suppose that infinitely-often OWPuzzs do not exist.
Then, infinitely-often weak non-uniform QEFID with computational indistinguishability and short advice do not exist.
\end{theorem}

\begin{remark}
Let us remark that this result is itself new to the field.
%\cref{thm:owpuzz}
\cite{C:ChuGolGra24} 
also studies how to construct OWPuzzs from QEFID; however, their technique cannot be applied in our setting.
More specifically, they show that QEFID implies OWPuzzs, and that weak OWPuzzs can be amplified to full OWPuzzs in certain parameter regimes.
A natural approach would therefore be to first convert weak non-uniform QEFID into weak OWPuzzs, and then lift the latter to obtain a full OWPuzzs.
However, the transformation from weak OWPuzzs to OWPuzzs in \cite{C:ChuGolGra24} requires parameter constraints that are not satisfied in our setting, so to the best of our knowledge this approach does not apply.
We therefore develop an alternative technique.
\end{remark}

\begin{proof}[Proof of \cref{thm:OWPuzz_Verify}]

In the following, we show that, for any polynomial $t$, any constant $c$, and any inverse polynomial $\epsilon$, there exists a QPT algorithm $\Vrfy$ and a polynomial $s$ that satisfies the correctness and adaptive soundness assuming the non-existence of infinitely-often OWPuzzs.

Before describing our verification protocol, we introduce the notation used throughout.
\paragraph{Notations:}
Let $U$ be a classical universal Turing machine. For each $n\in\N$ and each $\Pi\in\bit^*$, we define $\mathsf{QU}_{s,m}^{t(n)}(1^n,\Pi)$.
In the proof of \cref{thm:OWPuzz_Verify}, we define $\widetilde{\mathsf{QU}}_{s,m}^{t(n)}(1^n,\Pi)$, which outputs a quantum state.
Compared to $\widetilde{\mathsf{QU}}_{s,m}^{t(n)}(1^n,\Pi)$, $\mathsf{QU}_{s,m}^{t(n)}(1^n,\Pi)$ outputs a classical string by measuring a quantum state in the computational basis.
For clarity, let us describe $\mathsf{QU}_{s,m}^{t(n)}(1^n,\Pi)$.

\begin{description}
\item[$\mathsf{QU}_{s,m}^{t(n)}(1^n,\Pi)$:]$ $
\begin{itemize}
\item Run a classical universal Turing machine $U$ on $1^n$ and $\Pi$, halts within $t(n)$-steps, and obtains $c\in\bit^*$.
\item Consider $c$ as an encoding of output length $m(n)\cdot s(n)$, an input size $k(n)$, and quantum circuits $C$.
\item Prepare $C\ket{0^{k(n)}}$, measure the first $m(n)\cdot s(n)$-qubits of $C\ket{0^{k(n)}}$ in the computational basis, and output its outcome.\footnote{In quantum computing, the standard model of computation is the quantum circuit model rather than quantum Turing machine. Therefore, we define $\mathsf{QU}_{s,m}^{t(n)}(1^n,\Pi)$ using quantum circuits instead of using a quantum Turing machine.}  
\end{itemize}
\end{description}
We denote $\mathsf{QU}_{s,m}^{t(n)}(1^n,\Pi,i)$ to mean the following algorithm:
\begin{enumerate}
\item Run $x_1\|...\|x_{s(n)}\la \mathsf{QU}_{s,m}^{t(n)}(1^n,\Pi)$.
\item Output $x_i$.
\end{enumerate}

\if0
\color{red}
From \cref{cor:non_uniform}, for each $\Pi\in\bit^{\log(n)}$, $i\in[s(n)]$, any polynomial $t$ and any polynomial $p$, there exists a QPT algorithm $\mathsf{Dis}_{\Pi,i}$ such that the following is satisfied:
\begin{align}
\Pr[1\la\mathsf{Dis}_{\Pi,i}(x):x\la \mathsf{QU}_{s,m}^{t(n)}(1^n,\Pi,i)]-\Pr[1\la\mathsf{Dis}_{\Pi,i}(x):x\la\cD(1^n) ]\geq\mathsf{SD}(\mathsf{QU}_{s,m}^{t(n)}(1^n,\Pi,i),\cD(1^n))-\frac{1}{p(n)}
\end{align}
for all sufficiently large $n\in\N$.
\color{black}
\fi

From \cref{thm:owpuzz} and Observation~\ref{cor:non_uniform}, there exists a QPT algorithm $\mathsf{Dis}$ such that the followings are satisfied:
$\mathsf{Dis}$ takes $\Pi\in\bit^{\leq \log(n)}$, $i\in[s(n)]$, and $x\in\bit^{m(n)}$ as input, outputs $b\in\bit$, and satisfies the following properties for all $\Pi\in\bit^{\leq \log(n)}$ and $i\in[s(n)]$:

If
\begin{align*}
\mathsf{SD}(\mathsf{QU}_{s,m}^{t(n)}(1^n,\Pi,i),\cD(1^n))\geq \alpha(n)
\end{align*}
for some function $\alpha(n)>\epsilon(n)\cdot 2n^{-100c}$,
then we have
\begin{align*}
&\Pr[1\la\mathsf{Dis}(\Pi,i,x):x\la \mathsf{QU}_{s,m}^{t(n)}(1^n,\Pi,i)]-\Pr[1\la\mathsf{Dis}(\Pi,i,x):x\la\cD(1^n) ]\\
&\geq \mathsf{SD}(\mathsf{QU}_{s,m}^{t(n)}(1^n,\Pi,i),\cD(1^n))-\epsilon(n)\cdot n^{-100c}
\end{align*}
for all sufficiently large $n\in\N$.

In the following, we set $s(n)\seteq 2n^{10}\cdot n^{200c}$.~\footnote{Our goal is not to analyze the sample complexity of distribution verification in our definition, but its computational complexity. Hence, we take $s(n)$ quite large polynomial.
We leave to optimize the sample complexity as future work.}
\paragraph{Description of Verification Algorithm:}
Now, we give our construction of $\Vrfy$ as follows:
\begin{description}
\item[$\Vrfy\left(1^n,x_1\|...\|x_{s(n)} \right)$:]$ $
\begin{enumerate}
\item For each \( i \in [s(n)] \) and $\Pi\in\bit^{\leq \log(n)} $, run 
$
A_{\Pi,i}(x_i)\la\mathsf{Dis}_{\Pi,i}(x_i)
$.
\item For each \( i \in [s(n)] \) and $\Pi\in\bit^{\leq \log(n)}$, run $x_i^*\la\cD(1^n)$, run $B_{\Pi,i}(x_i^*)\la\mathsf{Dis}_{\Pi,i}(x_i^*)$.
\item If there exists a $\Pi\in\bit^{\leq\log(n)}$ such that
\begin{align*}
\frac{1}{s(n)}\left(\sum_{i \in [s(n)]} A_{\Pi,i}(x_i)-\sum_{i\in[s(n)]}B_{\Pi,i}(x_i^*)\right)\geq n^{-100c},
\end{align*}
then output $\bot$.
Otherwise, output $\top$.
\end{enumerate}
\end{description}

In the following, we show that $\Vrfy$ satisfies the correctness and adaptive-soundness.
\paragraph{Correctness:}
In the following, we show that
\begin{align*}
\Pr[\top\la\Vrfy(1^n,x_1,...,x_{s(n)}):x_1,...,x_{s(n)}\la\cD(1^n)^{\otimes s(n)}]\geq 1-\negl(n).
\end{align*}

For this, it is sufficient to show the following claim~\ref{claim:union}.
\begin{claim}\label{claim:union}
With probability at least $1-\negl(n)$, we have    
\begin{align*}
\frac{1}{s(n)}\left(\sum_{i \in [s(n)]} A_{\Pi,i}(x_i)-\sum_{i\in[s(n)]}B_{\Pi,i}(x_i^*)\right)<n^{-100c}
\end{align*}
for all $\Pi\in\bit^{\leq \log(n)}$.
Here, the probability is taken over $x_1,...,x_{s(n)}\la\cD(1^n)^{\otimes s(s)}$, $x_1^*,...,x_{s(n)}^*\la\cD(1^n)$, $\{A_{\Pi,i}(x_i)\la\mathsf{Dis}_{\Pi,i}(x_i)\}_{i\in[s(n)],\Pi\in\bit^{\leq \log(n)}}$, and $\{B_{\Pi,i}(x_i)\la\mathsf{Dis}_{\Pi,i}(x_i)\}_{i\in[s(n)],\Pi\in\bit^{\leq \log(n)}}$.
\end{claim}
Claim~\ref{claim:union} directly follows from union bound and Hoeffding's inequality.
For clarity, we describe the proof.
\begin{proof}[Proof of \cref{claim:union}]
From Hoeffdings's inequality~(Lemma~\ref{Hoeffding}), for each fixed $\Pi\in\bit^{\leq \log(n)}$, we have
\begin{align*}
&\Pr_{\substack{x_1,...,x_{s(n)}\la\cD(1^n)^{\otimes s(n)}\\ \{A_{\Pi,i}(x_i)\la\mathsf{Dis}_{\Pi,i}(x_i)\}_{i\in[s(n)]} } }\left[\frac{1}{s(n)}\left(\sum_{i\in[s(n)]}A_{\Pi,i}(x_i)-\sum_{i\in[s(n)]}\Pr[1\la\mathsf{Dis}(\Pi,i,x):x\la\cD(1^n)]\right)>\frac{ n^{-100c}}{2}\right]\\
&\hspace{15cm}\leq 2^{-\frac{s(n)}{2n^{200c}}}.
\end{align*}
This follows by considering the procedure $x_i\la\cD(1^n)$ and $A_{\Pi,i}(x_i)\la\mathsf{Dis}_{\Pi,i}(x_i) $ as a random variable $X_i$ in Lemma~\ref{Hoeffding}.

Similarly, for each fixed $\Pi\in\bit^{\leq \log(n)}$, we have
\begin{align*}
&\Pr_{\substack{x_1^*,...,x_{s(n)}^*\la\cD(1^n)^{\otimes s(n)}\\ \{ B_{\Pi,i}(x_i^*)\la \mathsf{Dis}_{\Pi,i}(x_i^*) \}_{i\in[s(n)]}}  }\left[\frac{1}{s(n)}\left(\sum_{i\in[s(n)]}\Pr[\mathsf{Dis}(\Pi,i,x):x\la\cD(1^n)]-\sum_{i\in[s(n)]}B_{\Pi,i}(x_i^*)\right)>\frac{n^{-100c}}{2}\right]\\
&\hspace{15cm}\leq 2^{-\frac{s(n)}{2n^{200c}}}.
\end{align*}
From union bound, two inequalities above imply that for each fixed $\Pi\in\bit^{\leq \log(n)}$, we have
\begin{align*}
\Pr_{\substack{x_1,...,x_{s(n)}\la\cD(1^n)^{\otimes s(n)},\\
x_1^*,...,x_{s(n)}^*\la\cD(1^n)^{\otimes s(n)}\\
\{A_{\Pi,i}(x_i)\la \mathsf{Dis}_{\Pi,i}(x_i)\}_{i\in[s(n)]}\\
\{B_{\Pi,i}(x_i^*)\la \mathsf{Dis}_{\Pi,i}(x_i^*)\}_{i\in[s(n)]}
}}
\left[\frac{1}{s(n)}\left(\sum_{i \in [s(n)]} A_{\Pi,i}(x_i)-\sum_{i\in[s(n)]}B_{\Pi,i}(x_i^*)\right)<n^{-100c}\right]\geq 1-2\cdot 2^{-\frac{s(n)}{2n^{200c}}}
\end{align*}
for all $n\in\N$.
By applying union bound again, Claim~\ref{claim:union} follows.
\end{proof}

\paragraph{Adaptive-Soundness:}

In the following, we show that for any $t(n)$-time quantum algorithm $\cA$ such that
\begin{align*}
\mathsf{SD}(\mathsf{Marginal}_{\cA}(1^n),\cD(1^n))\geq \epsilon(n),
\end{align*}
we have
\begin{align*}
\Pr[\top\la \Vrfy(1^n,x_1,...,x_{s(n)}):x_1,...,x_{s(n)}\la\cA(1^n)]\leq 1-\epsilon(n)+n^{-c}
\end{align*}
for all sufficiently large $n\in\N$.

For this, it is sufficient to show Claim~\ref{claim:Markov}.
\begin{claim}\label{claim:Markov}
For all sufficiently large $n\in\N$, there exists a $\Pi_\cA\in\bit^{\leq \log(n)}$ such that, with probability at least 
\begin{align*}
    \epsilon(n)-n^{-c},
\end{align*}
we have
\begin{align*}
\frac{1}{s(n)}\sum_{i\in[s(n)]}\left(A_{\Pi_{\cA},i}(x_i)-B_{\Pi_{\cA},i}(x_i^*)\right)\geq n^{-100c}.
\end{align*}
Here, the probability is taken over $x_1,...,x_{s(n)}\la\cA(1^n)$, $x_1^*,...,x_{s(n)}^*\la\cD(1^n)^{\otimes s(n)}$, $\{A_{\Pi_\cA,i}(x_i)\la\mathsf{Dis}_{\Pi_\cA,i}(x_i)\}_{i\in[s(n)]}$, and $\{B_{\Pi_\cA,i}(x_i^*)\la\mathsf{Dis}_{\Pi_\cA,i}(x_i^*)\}_{i\in[s(n)]}$
\end{claim}
Once we have obtained Claim~\ref{claim:Markov}, adaptive-soundness directly follows.
Recall that from the construction of $\Vrfy$, if there exists $\Pi_\cA\in\bit^{\leq \log(n)}$ such that
\begin{align*}
\frac{1}{s(n)}\sum_{i\in[s(n)]}\left(A_{\Pi_\cA,i}(x_i)-B_{\Pi_\cA,i}(x_i^*)\right)\geq n^{-100c},
\end{align*}
then $\Vrfy$ outputs $\bot$.
Therefore, from Claim~\ref{claim:Markov}, we have
\begin{align*}
&\Pr[\bot\la\Vrfy(x_1,...x_{s(n)}):(x_1,...,x_{s(n)})\la\cA(1^n)]\geq \epsilon(n)-n^{-c}
\end{align*}
for all sufficiently large $n\in\N$.
This implies that
\begin{align*}
&\Pr[\top\la\Vrfy(x_1,...x_{s(n)}):(x_1,...,x_{s(n)})\la\cA(1^n)]\leq 1-\epsilon(n)+n^{-c}
\end{align*}
for all sufficiently large $n\in\N$, which completes the proof of adaptive-soundness.

The remaining part of the proof is to show Claim~\ref{claim:Markov}. For showing Claim~\ref{claim:Markov}, we will use the following Claim~\ref{claim:probabilistic_argument}.

\begin{claim}\label{claim:probabilistic_argument}
For all sufficiently large $n\in\N$, there exists $\Pi_\cA\in\bit^{\leq \log(n)}$ such that
\begin{align*}
&\frac{1}{s(n)}\left(\sum_{i\in[s(n)]}\Pr[1\la\mathsf{Dis}(\Pi_{\cA},i,x):x\la\mathsf{QU}_{s,m}^{t(n)}(1^n,\Pi_\cA,i)]-\Pr[1\la\mathsf{Dis}(\Pi_{\cA},i,x):x\la\cD(1^n)]\right)\\
&\geq \epsilon(n)-3\epsilon(n)\cdot n^{-100c}.
\end{align*}    
\end{claim}
For showing Claim~\ref{claim:probabilistic_argument}, we will use the following Claim~\ref{claim:triangle}.

\begin{claim}\label{claim:triangle}
For all sufficiently large $n\in\N$, there exists a $\Pi_\cA\in\bit^{\leq\log(n)}$ such that
\begin{align*}
\frac{1}{s(n)}\sum_{i\in[s(n)]}\mathsf{SD}(\mathsf{QU}_{s,m}^{t(n)}(1^n,\Pi_\cA,i),\cD(1^n))\geq \mathsf{SD}(\mathsf{Marginal}_{\cA}(1^n),\cD(1^n))\geq \epsilon(n).
\end{align*}
\end{claim}

First, let us describe the proof of Claim~\ref{claim:Markov}.

\begin{proof}[Proof of Claim~\ref{claim:Markov}]
From Claim~\ref{claim:probabilistic_argument}, for all sufficiently large $n\in\N$, there exits $\Pi_{\cA}\in\bit^{\leq \log(n)}$ such that
\begin{align*}
&\mathbb{E}_{\substack{x_1,...,x_{s(n)}\la\cA(1^n),\\
x_1^*,...,x_{s(n)}^*\la\cD(1^n)^{\otimes s(n)},\\ 
A_{\Pi_\cA,i}(x_i)\la\mathsf{Dis}_{\Pi_\cA,i}(x_i)\\
B_{\Pi_\cA,i}(x_i^*)\la\mathsf{Dis}_{\Pi_\cA,i}(x_i^*)}}\left[\frac{1}{s(n)}\sum_{i\in[s(n)]}\left(A_{\Pi_\cA,i}(x_i)-B_{\Pi_\cA,i}(x_i^*)\right)\right]\\
&=\frac{1}{s(n)}\sum_{i\in[s(n)]}\left(\Pr[1\la\mathsf{Dis}_{\Pi_\cA,i}(x_i):x_i\la\mathsf{QU}_{s,m}^{t(n)}(1^n,\Pi_\cA,i)]-\Pr[1\la\mathsf{Dis}(\Pi_\cA,i,x):x\la\cD(1^n)] \right)\\
&\geq \epsilon(n)-3\epsilon(n)\cdot n^{-100c}.
\end{align*}
In the following, for simplicity, we write $X$ to denote $\frac{1}{s(n)}\sum_{i\in[s(n)]}\left(A_{\Pi_\cA,i}(x_i)-B_{\Pi_\cA,i}(x_i^*)\right)$.

Then, we have
\begin{align*}
\epsilon(n)-3\epsilon(n)\cdot n^{-100c}\leq \mathbb{E}[X]
&\leq n^{-100c}\Pr\left[X<n^{-100c}\right] +\Pr\left[n^{-100c}\leq X\right]\\
&=n^{-100c}+\left(1-n^{-100c}\right) \Pr\left[n^{-100c}\leq X\right].
\end{align*}
This implies that
\begin{align*}
\Pr\left[n^{-100c}\leq X\right]\geq\frac{\epsilon(n)-3\epsilon(n)\cdot n^{-100c}-n^{-100c}}{1-n^{-100c}}\geq \epsilon(n)-n^{-c}
\end{align*}
for all sufficiently large $n\in\N$.
This completes the proof.
\end{proof}

\begin{proof}[Proof of Claim~\ref{claim:probabilistic_argument}]
Since $\cA$ is a uniform algorithm, for all sufficiently large $n\in\N$, there exists a description $\Pi_\cA\in\bit^{\leq \log(n)}$ of an algorithm $\cA$ such that $\mathsf{QU}_{s,m}^{t(n)}(1^n,\Pi_{\cA})$ is statistically equivalent to $\cA(1^n)$.

Let us define a family $\cT_n$ of $i\in [s(n)]$ as follows:
\begin{align*}
\cT_n\seteq \{i\in[s(n)] :\mathsf{SD}(\mathsf{QU}_{s,m}^{t(n)}(1^n,\Pi_\cA,i),\cD(1^n))\geq 2\epsilon(n)\cdot n^{-100c}\}.
\end{align*}
From the definition of $\mathsf{Dis}$, for all $i\in\cT_n$, we have
\begin{align*}
&\Pr[1\la\mathsf{Dis}(\Pi_{\cA},i,x):x\la\mathsf{QU}_{s,m}^{t(n)}(1^n,\Pi_\cA,i)]-\Pr[1\la\mathsf{Dis}(\Pi_{\cA},i,x):x\la\cD(1^n)]\\
&\geq \mathsf{SD}(\mathsf{QU}_{s,m}^{t(n)}(1^n,\Pi_\cA,i),\cD(1^n)) -\epsilon(n)\cdot n^{-100c}
\end{align*}
for all sufficiently large $n\in\N$.
Therefore, we have 
\begin{align}
&\frac{1}{s(n)}\left(\sum_{i\in[s(n)]}\Pr[1\la\mathsf{Dis}(\Pi_{\cA},i,x):x\la\mathsf{QU}_{s,m}^{t(n)}(1^n,\Pi_\cA,i)]-\Pr[1\la\mathsf{Dis}(\Pi_{\cA},i,x):x\la\cD(1^n)] \right)\notag\\
&\geq \frac{1}{s(n)}\left(\sum_{i\in\cT_n}\Pr[1\la\mathsf{Dis}(\Pi_{\cA},i,x):x\la\mathsf{QU}_{s,m}^{t(n)}(1^n,\Pi_\cA,i)]-\Pr[1\la\mathsf{Dis}(\Pi_{\cA},i,x):x\la\cD(1^n)] \right)\notag\\
&\geq \frac{1}{s(n)}\left(\sum_{i\in\cT_n}(\mathsf{SD}(\mathsf{QU}_{s,m}^{t(n)}(1^n,\Pi_\cA,i),\cD(1^n))-\epsilon(n)\cdot n^{-100c})\right)\notag\\
&\geq \frac{1}{s(n)}\left(\sum_{i\in\cT_n}\mathsf{SD}(\mathsf{QU}_{s,m}^{t(n)}(1^n,\Pi_\cA,i),\cD(1^n))\right)-\epsilon(n)\cdot n^{-100c}\label{ineq:bound_1}
\end{align}
for all sufficiently large $n\in\N$.
Furthermore, 
\begin{align}
&\frac{1}{s(n)}\sum_{i\in\cT_n}\mathsf{SD}(\mathsf{QU}_{s,m}^{t(n)}(1^n,\Pi_\cA,i),\cD(1^n))\notag\\
&\geq \frac{1}{s(n)}\left(\sum_{i\in[s(n)]}\mathsf{SD}(\mathsf{QU}_{s,m}^{t(n)}(1^n,\Pi_\cA,i),\cD(1^n))-\sum_{i\notin \cT_n}\mathsf{SD}(\mathsf{QU}_{s,m}^{t(n)}(1^n,\Pi_\cA,i),\cD(1^n))\right)\notag\\
&\geq\epsilon(n)-\frac{1}{s(n)}\sum_{i\notin \cT}2\epsilon(n)\cdot n^{-100c}\notag\\
&\geq \epsilon(n)-2\epsilon(n)\cdot n^{-100c}\label{ineq:bound_2},
\end{align}
for all $n\in\N$, where, in the second inequality, we have used Claim~\ref{claim:triangle}.

By combining \cref{ineq:bound_1,ineq:bound_2}, we have 
\begin{align*}
&\frac{1}{s(n)}\left(\sum_{i\in[s(n)]}\Pr[1\la\mathsf{Dis}(\Pi_{\cA},i,x):x\la\mathsf{QU}_{s,m}^{t(n)}(1^n,\Pi_\cA,i)]-\Pr[1\la\mathsf{Dis}(\Pi_{\cA},i,x):x\la\cD(1^n)] \right)\\
&\geq \epsilon(n)-3\epsilon(n)\cdot n^{-100c}
\end{align*}
for all sufficiently large $n\in\N$.
\end{proof}

\begin{proof}[Proof of Claim~\ref{claim:triangle}]
Because $\cA$ is a uniform algorithm, for all sufficiently large $n\in\N$, there exists a $\Pi_\cA\in\bit^{\leq \log(n)}$ such that $\mathsf{QU}_{s,m}^{t(n)}(1^n,\Pi_{\cA})$ is statistically equivalent to $\cA(1^n)$.

Such $\Pi_\cA$ satisfies
\begin{align*}
&\frac{1}{s(n)}\sum_{i\in[s(n)]}\mathsf{SD}(\mathsf{QU}_{s,m}^{t(n)}(1^n,\Pi_\cA,i),\cD(1^n))\\
&=\frac{1}{s(n)}\sum_{i\in[s(n)]}\left(\sum_{x\in\bit^*}\abs{\Pr[x\la\mathsf{QU}_{s,m}^{t(n)}(1^n,\Pi_\cA,i)]-\Pr[x\la\cD(1^n)]}\right)\\
&=\sum_{x\in\bit^*}\left(\frac{1}{s(n)}\sum_{i\in[s(n)]}\abs{\Pr[x\la\mathsf{QU}_{s,m}^{t(n)}(1^n,\Pi_\cA,i)]-\Pr[x\la\cD(1^n)]}\right)\\
&\geq\sum_{x\in\bit^*}\abs{\frac{1}{s(n)}\sum_{i\in[s(n)]}\Pr[x\la\mathsf{QU}_{s,m}^{t(n)}(1^n,\Pi_\cA,i)]-\Pr[x\la\cD(1^n)] }\\
&=\sum_{x\in\bit^*}\abs{\Pr[x\la\mathsf{Marginal}_{\cA}(1^n)]-\Pr[x\la\cD(1^n)]}\\
&=\mathsf{SD}(\mathsf{Marginal}_{\cA}(1^n),\cD(1^n))\geq \epsilon(n).
\end{align*}
Here, in the inequality we have used triangle inequality (Lemma~\ref{triangle}).
\end{proof}
\end{proof}

\subsection{OWPuzzs from Weak Non-Uniform QEFID}\label{sec:owpuzz_non_uniform_QEFID}

In this section, we prove \cref{thm:owpuzz}, which states that OWPuzzs can be constructed from weak non-uniform QEFID.
For showing this, we use the following \cref{thm:extrapolate}.

\begin{theorem}[\cite{C:ChuGolGra24}]\label{thm:extrapolate}
Suppose that infinitely-often OWPuzzs do not exist.
Then, for any constant $c>1$, and any QPT algorithm $\cQ$, which takes $1^n$ as input, and outputs $(k,s)\in\bit^*$, there exists a QPT algorithm $\mathsf{Ext}$ such that the followings are satisfied:
\begin{align*}
\mathsf{SD}((k,s)_{(k,s)\la\cQ(1^n)},(\mathsf{Ext}(s),s)_{(k,s)\la\cQ(1^n)})\leq n^{-c}
\end{align*}
for all sufficiently large $n\in\N$.
\end{theorem}

What we actually need is Lemma~\ref{lem:distOWPuzz}, which is a slight generalization of \cref{thm:extrapolate}.
Specifically, \cref{thm:extrapolate} considers a quantum algorithm $\cQ$ that takes $1^n$ as input and outputs $(k,s)$.
Instead, in the generalized Lemma~\ref{lem:distOWPuzz}, we consider an algorithm $\cQ^*$  that is also allowed to take a short advice string $\mu\in[n]$ as input, and for all short advice the QPT algorithm $\mathsf{Ext}$ must succeed in extrapolation.
This proposition can be proved using a basic probabilistic argument.
For completeness, we include the proof at the end of this section.

\begin{lemma}\label{lem:distOWPuzz}
Suppose that infinitely-often OWPuzzs do not exist.
Then, for any constant $c>1$, for any QPT algorithm $\cQ$, which takes $1^n$ and $\mu\in[n]$ as input, and outputs a pair of classical string $(k,s)\in\bit^*$, there exists a QPT algorithm $\mathsf{Ext}$ such that for all $\mu\in[n]$,
\begin{align*}
\mathsf{SD}((k,s)_{k,s\la\cQ(1^n,\mu)},(\mathsf{Ext}(1^n,\mu,s),s)_{s\la\cQ(1^n,\mu)} )\leq n^{-c}
\end{align*}
for all sufficiently large $n\in\N$.
\end{lemma}

Now, we provide the proof of \cref{thm:owpuzz}.
\begin{proof}[Proof of \cref{thm:owpuzz}]
Let $\Gen$ be an arbitrary candidate of weak non-uniform QEFID that takes $1^n$, $\mu\in[n]$, and $b\in\bit$ as input, and outputs $x\in\bit^*$.
Assume that infinitely-often OWPuzzs do not exist.
Then, for any constant $c$, we construct a QPT algorithm $\cA$ such that the followings are satisfied:
For any $\mu\in[n]$ such that
\begin{align*}
\mathsf{SD}(\Gen(1^n,\mu,0),\Gen(1^n,\mu,1))\geq 2n^{-c},
\end{align*}
we have
\begin{align*}
\abs{\Pr_{x\la\Gen(1^n,\mu,0)}[1\la\cA(1^n,\mu,x)]-\Pr_{x\la\Gen(1^n,\mu,1)}[1\la\cA(1^n,\mu,x)] }\geq 
\mathsf{SD}(\Gen(1^n,\mu,0),\Gen(1^n,\mu,1))-n^{-c}
\end{align*}
for all sufficiently large $n\in\N$.

To describe $\cA$, let us consider a QPT algorithm $\cQ^*$:
\begin{description}
\item[$\cQ^*(1^n,\mu)$:]$ $
\begin{enumerate}
\item Sample $b\la\bit$.
\item Run $x\la\Gen(1^n,\mu,b)$.
\item Output $(b,x)$.
\end{enumerate}
\end{description}
For the construction of $\cA$, we will use a QPT algorithm $\mathsf{Ext}$ such that the following is satisfied:
For any $\mu\in[n]$, 
\begin{align*}
\mathsf{SD}\left(\left(b,x\right)_{(b,x)\la\cQ^*(1^n,\mu)}, \left(\mathsf{Ext}(\mu,x),x)\right)_{x\la\cQ^*(1^n,\mu)}\right)\leq n^{-100c}
\end{align*}
for all sufficiently large $n\in\N$. 
Note that such $\mathsf{Ext}$ exists if infinitely-often OWPuzz does not exist from Lemma~\ref{lem:distOWPuzz}.

Now, we give the construction of $\cA$:
\begin{description}
\item[$\cA(1^n,\mu,x)$:]$ $
\begin{enumerate}
\item For all $i\in[36n^{20c+1}]$, run $b_i\la\mathsf{Ext}(1^n,\mu,x)$.
\item If 
\begin{align*}
    \frac{1}{36n^{20c+1}}\sum_{i\in[36n^{20c+1}]}b_i\geq \frac{1}{2},
\end{align*}
output $1$.
Otherwise, output $0$.
\end{enumerate}
\end{description}

Now, we analyze $\cA(1^n,\mu,x)$.
For this, let us introduce notations and claims.
We denote $\mathsf{Post}(1^n,\mu,x)$ to mean a post-selection algorithm $\cQ^*(1^n,\mu)$ defined as follows:
\begin{description}
\item[$\mathsf{Post}(1^n,\mu,x)$:]$ $ 
\begin{enumerate}
\item Run $(b^*,x^*)\la\cQ^*(1^n,\mu)$ many times until $x=x^*$.
\item Output $b^*$ when $x=x^*$.
\end{enumerate}
\end{description}
Let 
\begin{align*}
\mathsf{Good}_{\mu,n}\seteq \left\{x\in\bit^*:\frac{1}{2}\sum_{b\in\bit}\abs{\Pr[b\la\mathsf{Ext}(\mu,x)]-\Pr[b\la\mathsf{Post}(\mu,x)]}\leq n^{-25c} \right\}.
\end{align*}
For any $\mu\in[n]$, $n\in\N$, and $b\in\bit$, let 
\begin{align*}
\cS_{\mu,n}(b)\seteq \left\{x\in\bit^*: \Pr[x\la\Gen(1^n,\mu,b)]\geq \left(1+n^{-10c}\right)\Pr[x\la\Gen(1^n,\mu,b\oplus 1)] \right\}.
\end{align*}

Now, we have the following claims:
\begin{claim}\label{claim:bound_1}

For all $\mu\in[n]$, we have
\begin{align*}
\sum_{x\in\mathsf{Good}_{\mu,n}}\Pr[x\la\Gen(1^n,\mu,0)]\geq 1-2n^{-75c}
\end{align*}
and
\begin{align*}
\sum_{x\in\mathsf{Good}_{\mu,n}}\Pr[x\la\Gen(1^n,\mu,1)]\geq 1-2n^{-75c}
\end{align*}
for all sufficiently large $n\in\N$.
\end{claim}

\begin{claim}\label{claim:bound_2}
We have
\begin{align*}
\sum_{x\in\cS_{\mu,n}(b)}\left(\Pr[x\la\Gen(1^n,\mu,b)]-\Pr[x\la\Gen(1^n,\mu,b\oplus 1)]\right)\geq \mathsf{SD}(\Gen(1^n,\mu,b\oplus 1),\Gen(1^n,\mu,b))-n^{-10c}.
\end{align*}
\end{claim}

\begin{claim}\label{claim:bound_3}
For all sufficiently large $n\in\N$, for any $\mu\in[n]$, and $b\in\bit$, we have the following:
For all $x\in\cS_{\mu,n}(b)\cap \mathsf{Good}_{\mu,n}$, we have
\begin{align*}
\Pr[b\la\cA(1^n,\mu,x)]\geq 1-\negl(n).
\end{align*}
\end{claim}
We defer the proof of Claims~\ref{claim:bound_1} to \ref{claim:bound_3} to the end of the proof.
Once we have obtained Claims~\ref{claim:bound_1} to \ref{claim:bound_3}, we can conclude the proof.

For all sufficiently large $n\in\N$ and $\mu\in[n]$, we have
\begin{align}
&\Pr_{x\la\Gen(1^n,\mu,1)}[1\la\cA(1^n,\mu,x)]-\Pr_{x\la\Gen(1^n,\mu,0)}[1\la\cA(1^n,\mu,x)]\notag\\
&=\sum_{x\in\bit^*}\left(\Pr[x\la\Gen(1^n,\mu,1)]-\Pr[x\la\Gen(1^n,\mu,0)]\right)\Pr[1\la\cA(1^n,\mu,x)]\notag\\
&=\sum_{x\in S_{\mu,n}(1)}\left(\Pr[x\la\Gen(1^n,\mu,1)]-\Pr[x\la\Gen(1^n,\mu,0)]\right)\Pr[1\la\cA(1^n,\mu,x)]\label{eq:sum_1}\\
&+\sum_{x\in S_{\mu,n}(0)}\left(\Pr[x\la\Gen(1^n,\mu,1)]-\Pr[x\la\Gen(1^n,\mu,0)]\right)\Pr[1\la\cA(1^n,\mu,x)]\label{eq:sum_2}\\
&+\sum_{x\in \overline{S_{\mu,n}(1)}\cap \overline{S_{\mu,n}(0)}}\left(\Pr[x\la\Gen(1^n,\mu,1)]-\Pr[x\la\Gen(1^n,\mu,0)]\right)\Pr[1\la\cA(1^n,\mu,x)]\label{eq:sum_3}
\end{align}
In the following, we bound \cref{eq:sum_1,eq:sum_2,eq:sum_3}.

First, we bound \cref{eq:sum_1}.
For all $\mu\in[n]$ and all sufficiently large $n\in\N$, we have
\begin{align}
&\mbox{\cref{eq:sum_1}}\notag\\
&=\sum_{x\in S_{\mu,n}(1)\cap\mathsf{Good}_{\mu,n}}\left(\Pr[x\la\Gen(1^n,\mu,1)]-\Pr[x\la\Gen(1^n,\mu,0)]\right)\Pr[1\la\cA(1^n,\mu,x)]\notag\\
&+\sum_{x\in \cS_{\mu,n}(1)\cap\overline{\mathsf{Good}_{\mu,n}} }\left(\Pr[x\la\Gen(1^n,\mu,1)]-\Pr[x\la\Gen(1^n,\mu,0)]\right)\Pr[1\la\cA(1^n,\mu,x)]\notag\\
&\geq \sum_{x\in S_{\mu,n}(1)\cap\mathsf{Good}_{\mu,n}}\left(\Pr[x\la\Gen(1^n,\mu,1)]-\Pr[x\la\Gen(1^n,\mu,0)]\right)\Pr[1\la\cA(1^n,\mu,x)]\notag\\
&\hspace{2cm}-\sum_{x\in\overline{\mathsf{Good}_{\mu,n}} }\Pr[x\la\Gen(1^n,\mu,0)]\notag\\
&\geq(1-\negl(n))\sum_{x\in S_{\mu,n}(1)\cap\mathsf{Good}_{\mu,n}}\left(\Pr[x\la\Gen(1^n,\mu,1)]-\Pr[x\la\Gen(1^n,\mu,0)]\right)\notag\\
&\hspace{2cm}-\sum_{x\in\overline{\mathsf{Good}_{\mu,n}} }\Pr[x\la\Gen(1^n,\mu,0)]\notag\\
&\geq (1-\negl(n))\left(\sum_{x\in S_{\mu,n}(1)\cap\mathsf{Good}_{\mu,n}}\left(\Pr[x\la\Gen(1^n,\mu,1)]-\Pr[x\la\Gen(1^n,\mu,0)]\right)\right)-2n^{-75c}.\label{eq:sum_4}
\end{align}
Here, in the second inequality, we have used $$\Pr[1\la \cA(x)]\geq 1-\negl(n)$$ for all $x\in\cS_{\mu,n}(1)\cap \mathsf{Good}_{\mu,n}$ for all sufficiently large, which is the statement of Claim~\ref{claim:bound_3}.
In the final inequality, we have also used $$\sum_{x\in\overline{\mathsf{Good}_{\mu,n}}} \Pr[x\la \Gen(1^n,\mu,0) ] \leq 2n^{-75c},$$ which follows from Claim~\ref{claim:bound_1}.

Moreover, for all $\mu\in[n]$ and all sufficiently large $n\in\N$, we have
\begin{align*}
&\mbox{\cref{eq:sum_4}}\\
&\geq (1-\negl(n))\left(\sum_{x\in S_{\mu,n}(1)\cap\mathsf{Good}_{\mu,n}}\left(\Pr[x\la\Gen(1^n,\mu,1)]-\Pr[x\la\Gen(1^n,\mu,0)]\right)\right)-2n^{-75c}\\
&\geq (1-\negl(n))\left(\sum_{x\in S_{\mu,n}(1)}\left(\Pr[x\la\Gen(1^n,\mu,1)]-\Pr[x\la\Gen(1^n,\mu,0)]\right)\right)\\
&\,\,-(1-\negl(n))\left(\sum_{x\in S_{\mu,n}(1)\cap\overline{\mathsf{Good}_{\mu,n}} }\left(\Pr[x\la\Gen(1^n,\mu,1)]-\Pr[x\la\Gen(1^n,\mu,0)]\right)\right)-2n^{-75c}\\
&\geq (1-\negl(n))\left(\sum_{x\in S_{\mu,n}(1)}\left(\Pr[x\la\Gen(1^n,\mu,1)]-\Pr[x\la\Gen(1^n,\mu,0)]\right)\right)\\
&-(1-\negl(n))\sum_{x\in S_{\mu,n}(1)\cap\overline{\mathsf{Good}_{\mu,n}} }\Pr[x\la\Gen(1^n,\mu,1)]-2n^{-75c}\\
&\geq (1-\negl(n))\left(\sum_{x\in S_{\mu,n}(1)}\left(\Pr[x\la\Gen(1^n,\mu,1)]-\Pr[x\la\Gen(1^n,\mu,0)]\right)\right)-4n^{-75c}\\
&\geq(1-\negl(n))\cdot \left(\mathsf{SD}(\Gen(1^n,\mu,1),\Gen(1^n,\mu,0))-n^{-10c}\right)-4n^{-75c}.
\end{align*}
Here, in the fourth inequality, we have used $$\sum_{x\in\overline{\mathsf{Good}_{\mu,n}}}\Pr[x\la \Gen(1^n,\mu,1)]\leq 2n^{-75c},$$ which follows from Claim~\ref{claim:bound_1}.
In the final inequality, we have also used $$\sum_{x\in S_{\mu,n}(1)}\left(\Pr[x\la\Gen(1^n,\mu,1)]-\Pr[x\la\Gen(1^n,\mu,0)]\right) \geq \mathsf{SD}(\Gen(1^n,\mu,1),\Gen(1^n,\mu,0))-n^{-10c},$$ which is the statement of Claim~\ref{claim:bound_2}.

Second, we bound \cref{eq:sum_2}.
For all $\mu\in[n]$ and all sufficiently large $n\in\N$, we have
\begin{align*}
&\mbox{\cref{eq:sum_2}}\\
&=\sum_{x\in S_{\mu,n}(0)\cap\mathsf{Good}_{\mu,n}}\left(\Pr[x\la\Gen(1^n,\mu,1)]-\Pr[x\la\Gen(1^n,\mu,0)]\right)\Pr[1\la\cA(1^n,\mu,x)]\\
&+\sum_{x\in \cS_{\mu,n}(0)\cap\overline{\mathsf{Good}_{\mu,n}} }\left(\Pr[x\la\Gen(1^n,\mu,1)]-\Pr[x\la\Gen(1^n,\mu,0)]\right)\Pr[1\la\cA(1^n,\mu,x)]\\
&\geq -\sum_{x\in S_{\mu,n}(0)\cap\mathsf{Good}_{\mu,n}}\Pr[x\la\Gen(1^n,\mu,0)]\Pr[1\la\cA(1^n,\mu,x)]-\sum_{x\in \cS_{\mu,n}(0)\cap\overline{\mathsf{Good}_{\mu,n}} }-\Pr[x\la\Gen(1^n,\mu,0)]\\
&\geq -\negl(n)-\sum_{x\in \cS_{\mu,n}(0)\cap\overline{\mathsf{Good}_{\mu,n}} }\Pr[x\la\Gen(1^n,\mu,0)]\\
&\geq -\negl(n)-2n^{-75c}.
\end{align*}
Here, in the second inequality, we have used 
$$
\Pr[1\la \cA(1^n,\mu,x)]\leq \negl(n)
$$
for all $x\in\cS_{\mu,n}(0)\cap\mathsf{Good}_{\mu,n}$
which is the statement of Claim~\ref{claim:bound_3}.
In the final inequality, we have also used 
$$
\sum_{x\in \overline{\mathsf{Good}_{\mu,n}} }\Pr[x\la\Gen(1^n,\mu,0)]\leq 2n^{-75c},
$$
which is the statement of Claim~\ref{claim:bound_1}.

Finally, we bound \cref{eq:sum_3}.
From the definition of $\cS_{\mu,n}(0)$ and $\cS_{\mu,n}(1)$, we have
\begin{align*}
&\overline{\cS_{\mu,n}(0)}\cap \overline{\cS_{\mu,n}(1)}\\
&=\left\{x\in\bit^*:\frac{1}{1+n^{-10c}}\Pr[x\la\Gen(1^n,\mu,1)]<\Pr[x\la\Gen(1^n,\mu,0)]<(1+n^{-10c})\Pr[x\la\Gen(1^n,\mu,1)] \right\}.
\end{align*}
For all $\mu\in[n]$ and all sufficiently large $n\in\N$, we have
\begin{align*}
&\mbox{\cref{eq:sum_3}}\\
&\geq \sum_{x: (1+n^{-10c})\Pr[x\la\Gen(1^n,\mu,1)]>\Pr[x\la\Gen(1^n,\mu,0)]\geq\Pr[x\la\Gen(1^n,\mu,1)] }\left(\Pr[x\la\Gen(1^n,\mu,1)]-\Pr[x\la\Gen(1^n,\mu,0)]\right)\\
&\geq -n^{-10c}\sum_{x: (1+n^{-10c})\Pr[x\la\Gen(1^n,\mu,1)]>\Pr[x\la\Gen(1^n,\mu,0)]\geq\Pr[x\la\Gen(1^n,\mu,1)] }\Pr[x\la\Gen(1^n,\mu,0)]\\
&\geq -n^{-10c}.
\end{align*}

\if0
\taiga{kokokara}
Finally, we bound \cref{eq:sum_3}.
For all $\mu\in[n]$ and all sufficiently large $n\in\N$, we have
\begin{align*}
&\mbox{\cref{eq:sum_3}}\\
&\geq \sum_{x: (1+n^{-10c})\Pr[x\la\Gen(1^n,\mu,1)]>\Pr[x\la\Gen(1^n,\mu,0)]\geq\Pr[x\la\Gen(1^n,\mu,1)] }\left(\Pr[x\la\Gen(1^n,\mu,1)]-\Pr[x\la\Gen(1^n,\mu,0)]\right)\\
&\geq -n^{-10c}.
\end{align*}
Here, in the first inequality, we have used the definition of $\cS_{\mu,n}(0)$ and $\cS_{\mu,n}(1)$, and in the final inequality, we have used the fact that
\begin{align*}
&\sum_{x: (1+n^{-10c})\Pr[1\la\Gen(1^n,\mu,1)]>\Pr[1\la\Gen(1^n,\mu,0)]\geq\Pr[1\la\Gen(1^n,\mu,1)] }\left(\Pr[x\la\Gen(1^n,\mu,0)]-\Pr[x\la\Gen(1^n,\mu,1)]\right)\\
&=\mathsf{SD}(\Gen(1^n,\mu,0),\Gen(1^n,\mu,1))-\sum_{x\in \cS_{\mu,n}(0)}\left(\Pr[x\la\Gen(1^n,\mu,0)]-\Pr[x\la\Gen(1^n,\mu,1)]\right)\leq n^{-10c},
\end{align*}
where the final inequality follows from \cref{claim:bound_3}.
\taiga{kokomade}
\fi

By combining these bounds, we have
\begin{align*}
&\Pr[1\la\cA(1^n,\mu,x):x\la\Gen(1^n,\mu,1)]-\Pr[1\la\cA(1^n,\mu,x):x\la\Gen(1^n,\mu,0)]\\
&=\mbox{\cref{eq:sum_1}}+\mbox{\cref{eq:sum_2}}+\mbox{\cref{eq:sum_3}}\\
&\geq (1-\negl(n))(\mathsf{SD}(\Gen(1^n,\mu,1),\mathsf{SD}(1^n,\mu,0)) -n^{-10c})-4n^{-75c} -(\negl(n)+2n^{-75c})-n^{-10c}\\
&\geq\mathsf{SD}(\Gen(1^n,\mu,1),\mathsf{SD}(1^n,\mu,0))-8n^{-10c}-\negl(n)\\
&\geq\mathsf{SD}(\Gen(1^n,\mu,1),\mathsf{SD}(1^n,\mu,0))-n^{-c}
\end{align*}
for all $\mu\in[n]$ and all sufficiently large $n\in\N$.
This completes the proof.

\begin{proof}[Proof of Claim~\ref{claim:bound_1}]
For any $\mu\in[n]$, we have
\begin{align*}
n^{-100c}
&\geq \mathsf{SD}\left((b,x)_{(b,x)\la\cQ^*(1^n,\mu)},(\mathsf{Ext}(\mu,x),x)_{x\la\cQ^*(1^n,\mu)}\right)\\
&=\frac{1}{2}\sum_{x\in\bit^*}\Pr[x\la\cQ^*(1^n,\mu)]\sum_{b\in\bit}\abs{\Pr[b\la\mathsf{Ext}(\mu,x)]-\Pr[b\la\mathsf{Post}(\mu,x)]}\\
&=\frac{1}{2}\sum_{x\in\mathsf{Good}_{\mu,n}}\Pr[x\la\cQ^*(1^n,\mu)]\sum_{b\in\bit}\abs{\Pr[b\la\mathsf{Ext}(\mu,x)]-\Pr[b\la\mathsf{Post}(\mu,x)]}\\
&+\frac{1}{2}\sum_{x\in\overline{\mathsf{Good}_{\mu,n}}}\Pr[x\la\cQ^*(1^n,\mu)]\sum_{b\in\bit}\abs{\Pr[b\la\mathsf{Ext}(\mu,x)]-\Pr[b\la\mathsf{Post}(\mu,x)]}\\
&\geq\sum_{x\in\overline{\mathsf{Good}_{\mu,n}}}\Pr[x\la\cQ^*(1^n,\mu)]n^{-25c}
\end{align*}
for all sufficiently large $n\in\N$.
This means that 
\begin{align*}
    \sum_{x\in\mathsf{Good}_{\mu,n}}\Pr[x\la\cQ^*(1^n,\mu)]\geq 1-n^{-75c}
\end{align*}
for all $\mu\in[n]$ for all sufficiently large $n\in\N$.

Furthermore, for each $b\in\bit$,
\begin{align*}
1-n^{-75c}&\leq\sum_{x\in\mathsf{Good}_{\mu,n} }\Pr[x\la\cQ^*(1^n,\mu)] \\
&= \sum_{x\in\mathsf{Good}_{\mu,n}} \frac{1}{2}\left(\Pr[x\la\Gen(1^n,\mu,0)]+\Pr[x\la\Gen(1^n,\mu,1)]\right)\\
&\leq \frac{1}{2}+\sum_{x\in\mathsf{Good}_{\mu,n}} \frac{1}{2}\Pr[x\la\Gen(1^n,\mu,b)]
\end{align*}
for all $\mu\in[n]$
for all sufficiently large $n\in\N$.
Hence, for each $b\in\bit$, we have
\begin{align*}
\sum_{x\in\mathsf{Good}_{\mu,n}} \Pr[x\la\Gen(1^n,\mu,b)]\geq 1-2n^{-75c}
\end{align*}
for all $\mu\in[n]$ and all sufficiently large $n\in\N$.
\end{proof}

\begin{proof}[Proof of Claim~\ref{claim:bound_2}]
For all $\mu\in[n]$, $b\in\bit$, and all sufficiently large $n\in\N$, we have
\begin{align*}
&\mathsf{SD}(\Gen(\mu,b\oplus 1),\Gen(\mu,b))\\
&=\sum_{x:\Pr[x\la \Gen(\mu,b)]>\Pr[x\la \Gen(\mu,b\oplus 1)]}\Pr[x\la\Gen(\mu,b)]-\Pr[x\la\Gen(\mu,b\oplus 1)]\\
&=\sum_{x\in\cS_{\mu,n}(b)}\Pr[x\la\Gen(\mu,b)]-\Pr[x\la\Gen(\mu,b\oplus 1)]\\
&+\sum_{x:(1+n^{-25c})\Pr[x\la \Gen(\mu,b\oplus 1)]>\Pr[x\la \Gen(\mu,b)]>\Pr[x\la \Gen(\mu,b\oplus 1)]}\Pr[x\la\Gen(\mu,b)]-\Pr[x\la\Gen(\mu,b\oplus 1)]\\
&\leq \sum_{x\in\cS_{\mu,n}(b)}\Pr[x\la\Gen(\mu,b)]-\Pr[x\la\Gen(\mu,b\oplus 1)]\\
&+\sum_{x:(1+n^{-25c})\Pr[x\la \Gen(\mu,b\oplus 1)]>\Pr[x\la \Gen(\mu,b)]>\Pr[x\la \Gen(\mu,b\oplus 1)]}n^{-25c}\Pr[x\la\Gen(\mu,b\oplus 1)]\\
&\leq \sum_{x\in\cS_{\mu,n}(b)}\left(\Pr[x\la\Gen(\mu,b)]-\Pr[x\la\Gen(\mu,b\oplus 1)]\right)+n^{-25c}.
\end{align*}
This means that, for all $\mu\in[n]$, $b\in\bit$, and all sufficiently large $n\in\N$,
\begin{align*}
\sum_{x\in\cS_{\mu,n}(b)}\left(\Pr[x\la\Gen(\mu,b)]-\Pr[x\la\Gen(\mu,b\oplus 1)]\right)\geq \mathsf{SD}(\Gen(\mu,b\oplus 1),\Gen(\mu,b))-n^{-25c} 
\end{align*}
which completes the proof.

\end{proof}

\begin{proof}[Proof of Claim~\ref{claim:bound_3}]
From the definition of $\cS_{\mu,n}(b)$, for all $x\in\cS_{\mu,n}(b)$, we have 
\begin{align*}
\Pr[b\la\mathsf{Post}(1^n,\mu,x)]=\frac{\Pr[x\la\Gen(1^n,\mu,b)]}{\Pr[x\la\Gen(1^n,\mu,0)]+\Pr[x\la\Gen(1^n,\mu,1)]}\geq \frac{1+n^{-10c}}{2+n^{-10c}}
\end{align*}
for all $\mu\in[n]$ and $n\in\N$.
From the definition of $\mathsf{Good}_{\mu,n}$, for all $x\in\cS_{\mu,n}(b)\cap\mathsf{Good}_{\mu,n}$, 
\begin{align*}
n^{-25c}&\geq \frac{1}{2}\sum_{b^*\in\bit}\abs{\Pr[b^*\la\mathsf{Ext}(\mu,x)]-\Pr[b^*\la\mathsf{Post}(\mu,x)]}\\
&\geq  \frac{1}{2}\abs{\Pr[b\la\mathsf{Ext}(\mu,x)]-\Pr[b\la\mathsf{Post}(\mu,x)]}
\end{align*}
for all sufficiently large $n\in\N$.
This means that 
\begin{align*}
\Pr[b\la\mathsf{Ext}(\mu,x)]\geq \frac{1+n^{-10c}}{2+n^{-10c}}-2n^{-25c} 
&=\frac{1}{2}+\left(\frac{n^{-10c}}{4+2n^{-10c}}-2n^{-25c}\right)\\
&\geq \frac{1}{2}+\frac{1}{6}n^{-10c}
\end{align*}
for all $x\in\cS_{\mu,n}(b)\cap\mathsf{Good}_{\mu,n}$, all $\mu\in[n]$, and all sufficiently large $n\in\N$.

On the other hand, from Hoeffding's inequality, for all $\mu\in[n]$, all $x\in\bit^*$, and all $n\in\N$, we have
\begin{align*}
\Pr_{\{b_i\}_{i\in[36n^{20c+1}]}\la\mathsf{Ext}(1^n,\mu,x)^{\otimes 36n^{20c+1}}}\left[\left|\frac{1}{36n^{20c+1}}\sum_{i\in[36^{20c+1}]}b_i- \Pr[1\la\mathsf{Ext}(\mu,x)]\right|\leq\frac{1}{6}n^{-10c} \right]\geq 1-2^{-n+1}.
\end{align*}
From the construction of $\cA$, this means that for all $x\in\cS_{\mu,n}(b)\cap\mathsf{Good}_{\mu,n}$, all $\mu\in[n]$, and all sufficiently large $n\in\N$, we have
\begin{align*}
    \Pr[b\la\cA(x)]\geq 1-2^{-n+1}\geq 1-\negl(n).
\end{align*}
\end{proof}
\end{proof}

\begin{proof}[Proof of Lemma~\ref{lem:distOWPuzz}]
For any QPT algorithm $\cQ$, which takes $1^n$ and $\mu\in[n]$ as input, and outputs $(k,s)\in\bit^*$, we consider the following QPT algorithm $\cQ^*$.
\begin{description}
    \item[$\cQ^*(1^n)$:]$ $
    \begin{enumerate}
        \item Sample $\mu\la [n]$.
        \item Run $(k,s)\la\cQ(1^n,\mu) $, and output $k^*=k$ and $s^*=(s,\mu)$.
    \end{enumerate}
\end{description}
Suppose that infinitely-often OWPuzzs do not exist. 
Then, from \cref{thm:extrapolate}, for any constant $c>0$, there exists a QPT algorithm $\mathsf{Ext}^*$ such that
\begin{align*}
\mathsf{SD}((k,(s,\mu))_{(k,(s,\mu))\la\cQ^*(1^n)},(\mathsf{Ext}^*(s,\mu),(s,\mu))_{(s,\mu)\la\cQ^*(1^n)} )\leq n^{-c-1}
\end{align*}
for all sufficiently large $n\in\N$.
We have
\begin{align*}
&n^{-c-1}\\
&\geq\mathsf{SD}((k,(s,\mu))_{(k,(s,\mu))\la\cQ^*(1^n)},(\mathsf{Ext}^*(s,\mu),(s,\mu))_{(s,\mu)\la\cQ^*(1^n)} )\\
&=\frac{1}{2}\sum_{(k,s)\in\bit^*,\mu\in[n]}\abs{\Pr[(k,s,\mu)\la\cQ^*(1^n)]-\Pr[k\la\mathsf{Ext}^*(s,\mu)]\Pr[s\la\cQ(\mu)]\Pr[\mu\la[n] ]}\\
&=\frac{1}{2}\sum_{(k,s)\in\bit^*,\mu\in[n]}\abs{\frac{1}{n}\Pr[(k,s)\la\cQ(1^n,\mu) ]-\frac{1}{n}\Pr[k\la\mathsf{Ext}^*(s,\mu)]\Pr[s\la\cQ(\mu) ]}\\
&=\frac{1}{n}\sum_{\mu\in[n]}\sum_{(k,s)\in\bit^*}\frac{1}{2}\abs{\Pr[(k,s)\la\cQ(1^n,\mu) ]-\Pr[k\la\mathsf{Ext}^*(s,\mu)]\Pr[s\la\cQ(\mu)]}\\
&=\frac{1}{n}\sum_{\mu\in[n]}\mathsf{SD}\left((k,s)_{(k,s)\la\cQ(1^n,\mu)},(\mathsf{Ext}^*(s,\mu),s)_{(k,s)\la\cQ(1^n,\mu)} \right)
\end{align*}

For all $\mu^*\in[n]$, we have
\begin{align*}
n^{-c-1}&\geq \frac{1}{n}\sum_{\mu\in[n]}\mathsf{SD}\left((k,s)_{(k,s)\la\cQ(1^n,\mu)},(\mathsf{Ext}^*(s,\mu),s)_{(k,s)\la\cQ(1^n,\mu)} \right)\\
&\geq \frac{1}{n}\mathsf{SD}\left((k,s)_{(k,s)\la\cQ(1^n,\mu^*)},(\mathsf{Ext}^*(s,\mu^*),s)_{(k,s)\la\cQ(1^n,\mu^*)} \right)
\end{align*}
for all sufficiently large $n\in\N$.
This implies that, for all $\mu^*\in[n]$,
\begin{align*}
    n^{-c}\geq \mathsf{SD}\left((k,s)_{(k,s)\la\cQ(1^n,\mu^*)},(\mathsf{Ext}^*(s,\mu^*),s)_{(k,s)\la\cQ(1^n,\mu^*)} \right)
\end{align*}
for all sufficiently large $n\in\N$, which completes the proof.
\end{proof}

\subsection{Hardness of Quantum Distribution Verification from QEFID}

\begin{proposition}[Restatement of \cref{inf:q_verify}]\label{thm:q_hardness}
Assume that infinitely-often QEFID exists.
Then, there exists a QPT algorithm $\cQ$, which is not selectively-verifiable with a QPT algorithm.
\end{proposition}
\begin{proof}[Proof of \cref{thm:q_hardness}]

Let us consider a infinitely-often QEFID $\Gen$.
Then, from the definition, we have
\begin{align*}
\mathsf{SD}(\Gen(1^n,0),\Gen(1^n,1))\geq 1-\negl(n)
\end{align*}
and for any QPT algorithm $\Vrfy$, any polynomial $s$, and any polynomial $p$, we have
\begin{align*}
&\Pr[\top\la\Vrfy(x_1,\dots,x_s):(x_1,\dots,x_{s})\la\Gen(1^n,0)^{\otimes s}]\\
&\hspace{2cm}-\Pr[\top\la\Vrfy(x_1,\dots,x_s):(x_1,\dots,x_s)\la\Gen(1^n,1)^{\otimes s}]\leq \frac{1}{p(n)}
\end{align*}
for infinitely many $n\in\N$.
This means that, for each $b\in\bit$, $\Gen(1^n,b)$ is not  selectively-verifiable in the sense of Definition~\ref{def:adaptive}.
\end{proof}

\section{Classical Certified Randomness with Inefficient Verification}
\label{sec:certified-randomness}

In this section, we prove that certified randomness can be obtained unconditionally when the verification algorithm is allowed to be inefficient.
Several prior works~\cite{STOC:AarHun23,BBFGT26} in the certified randomness literature also allow inefficient verification procedures, but their constructions rely on additional assumptions related to sampling quantum advantage and require the prover to run quantum algorithms.
Our construction, in contrast, does not rely on any unproven assumptions and avoids the need to execute quantum algorithms.
Because the security definitions of \cite{STOC:AarHun23} are different from ours, the results are incomparable.
On the other hand, \cite{BBFGT26} considers side-information free security, but based on the random-oracle model while our construction is unconditional.

\begin{definition}[Non-Interactive Inefficiently-Verifiable Classical Certified Randomness]\label{def:cert_randomness}
Let $0<c<1$ be an arbitrary constant.
A non-interactive inefficiently-verifiable certified randomness with $c$-(side-information free) security is a pair $(\cV,\cP)$ of algorithms.
$\cV$ is an unbounded algorithm that takes $1^n$ and $x\in\bit^{n}$ as input, and outputs $\top$ or $\bot$.
$\cP$ is a PPT algorithm that takes $1^n$ as input, and outputs $x\in\bit^{n}$.

We require the following two properties:
\begin{description}
\item[Correctness:]
\begin{align*}
\Pr_{x\la\cP(1^n)}[\top\la\cV(1^n,x)]\geq 1-\negl(n).
\end{align*}
\item[$c$-(Side-Information Free) Security:] 
For any polynomial $p$ and any QPT algorithm $\cA$, the following holds for all sufficiently large $n\in\mathbb{N}$:
If
\begin{align*}
\Pr_{x\la\cA(1^n)}[\top\la\cV(1^n,x) ]\geq \frac{1}{p(n)},
\end{align*}
then
\begin{align*}
H_{min}(\cA_{\top}(1^n))\geq cn.
\end{align*}
Here, $\cA_{\top}(1^n)$ denotes the output of the algorithm $\cA(1^n)$ conditioned on the event that $\cV(1^n,x)$ outputs $\top$.
\end{description}
\end{definition}

\begin{theorem}[Restatement of \cref{inf:cert_rand}]\label{thm:certified_randomness}
For an arbitrary $0<c<1$, a non-interactive inefficiently-verifiable classical certified randomness with $c$-(side-information free) security unconditionally exists.
\end{theorem}

\begin{proof}[Proof of \cref{thm:certified_randomness}]
We provide a construction of a non-interactive inefficiently-verifiable classical certified randomness with $c$-(side-information free) security $(\cV,\cP)$.
\begin{description}
\item[$\cP(1^n)$:]$ $
\begin{enumerate}
\item Uniformly randomly sample $x\la\bit^n$.
\item Output $x$.
\end{enumerate}
\item[$\cV(1^n,x)$:]$ $ 
\begin{enumerate}
\item Compute $quK^{n^{\log(n)}}(x|1^n)$.
\item Output $\top$ if
\begin{align*}
quK^{n^{\log(n)}}(x|1^n)\geq \frac{c+1}{2}n.
\end{align*}
Otherwise, output $\bot$.
\end{enumerate}
\end{description}

\paragraph{Correctness.}
From \cref{thm:q_kraft}, we have
\begin{align*}
\Pr_{x\la\cP(1^n)}[quK^{n^{\log(n)}}(x|1^n) \geq \frac{c+1}{2}n]\geq 1-2^{-\frac{1-c}{2}n}
\end{align*}
for all sufficiently large $n\in\N$.
This implies that
\begin{align*}
\Pr_{x\la\cP(1^n)}[\top\la\cV(1^n,x)]\geq 1-\negl(n).
\end{align*}

\paragraph{Side-information free security.}
For showing this, it is sufficient to show that
\begin{align*}
\Pr[x\la \cA_{\top}(1^n)]\leq 2^{-cn}
\end{align*}
for all $x \in\bit^n$,
when
\begin{align}
\Pr_{x\la\cA(1^n)}[\top\la\cV(1^n,x)]\geq \frac{1}{p(n)}\label{eq:condition}
\end{align}
for some polynomial $p$.
From the construction of $\cV$, we have
\begin{align}
\Pr[x\la\cA_{\top}(1^n)]=
\left\{
\begin{array}{ll}
\frac{1}{\Pr_{x\la\cA(1^n)}[\top\la\cV(x) ]}\Pr[x\la\cA(1^n)]  & (x:quK^{n^{\log(n)}}(x|1^n)\geq \frac{c+1}{2}n)\\
0 & (x:quK^{n^{\log(n)}}(x|1^n)<\frac{c+1}{2}n)
\end{array}
\right.\label{eq:definition}
\end{align}
From \cref{lem:q_coding}, for any QPT algorithm $\cA$,
\begin{align*}
quK^{n^{\log(n)}}(x|1^n)\leq -\log(\Pr[x\la \cA(1^n) ])+\log(n)
\end{align*}
for all sufficiently large $n\in\N$.
This implies that, for all sufficiently large $n\in\N$,
for all $x\in\bit^n$ such that
\begin{align*}
quK^{n^{\log(n)}}(x|1^n)\geq \frac{c+1}{2}n,
\end{align*}
we have
\begin{align}
\Pr[x\la \cA(1^n) ]\leq \frac{1}{n}2^{-\frac{c+1}{2}n}\label{eq:kolmogorov}
\end{align}
for all sufficiently large $n\in\N$.

By applying \cref{eq:condition,eq:definition,eq:kolmogorov}, for all $x\in\bit^n$, we have
\begin{align*}
\Pr[x\la \cA_{\top}(1^n)]\leq \frac{p(n)}{n}2^{-\frac{c+1}{2}n}\leq 2^{-cn}
\end{align*}
for all sufficiently large $n\in\N$.
\end{proof}

\section{Kolmogorov-Complexity-Based Verification of Any Quantum Sampling Advantage}
\label{sec:universal-ver}

\subsection{Results}
In this section, we prove the following \cref{lem:QAS-qkt}.

\begin{theorem}[Kolmogorov complexity as a benchmark for quantum advantage (Restatement of \cref{inf:qcd})]\label{lem:QAS-qkt}
For any polynomial $t$, there exists another polynomial $t^*$ such that
the value 
\begin{align*}
  qcd^{t^*(n)}(x|1^n)\seteq  uK^{t^*(n)}(x|1^n)-quK^{t^*(n)}(x|1^n)
\end{align*}
satisfy the following two conditions:
\begin{description}
\item[Quantum Easiness:]
For any $t(n)$-time quantum algorithm $\cD_\cQ(1^n)$ such that, for any PPT algorithm $\cD$,
\begin{align*}
\mathsf{SD}(\cD_\cQ(1^n),\cD(1^n))\geq 1-\negl(n),
\end{align*}
we have
\begin{align*}
\Pr_{x\la\cD_\cQ(1^n)}[c^*\log(n)\leq qcd^{t^*(n)}(x|1^n)]\geq 1-n^{-c}
\end{align*}
for all sufficiently large $n\in\N$.
Here, $c,c^*>0$ is an arbitrary constant.
\item[Classical Hardness:]
For any $\alpha>0$, for any $t(n)$-time classical algorithm $\cA(1^n)$, there exists a constant $C$ such that
\begin{align*}
\Pr_{x\la\cA(1^n)}[ \alpha \leq qcd^{t^*(n)}(x|1^n)]\leq C \abs{x} 2^{-\alpha}
\end{align*}
for all sufficiently large $n\in\N$.
\end{description}    
\end{theorem}

Combining~\Cref{lem:QAS-qkt} with the fact that if one-way puzzles don't exist one can estimate $qcd^t$ gives us

\begin{corollary}\label{thm:QAS_Verification}
    Suppose that infinitely-often OWPuzzs do not exist. Then, for any polynomial $t$ and any constant $c > 0$, there exists a QPT algorithm $\Vrfy$ that, on input $x \in \{0,1\}^{m(n)}$, outputs $\top$ or $\bot$, where $m(n)$ is a polynomial satisfying $m(n) \ge n$ for all $n \in \mathbb{N}$.
    \begin{description}
    \item[Correctness:]
    For any $t(n)$-time quantum algorithm $\cD_\cQ(1^n)$ such that, for any PPT algorithm $\cD$, 
    \begin{align*}
    \mathsf{SD}(\cD_\cQ(1^n),\cD(1^n))\geq 1-\negl(n),
    \end{align*}
    we have
    \begin{align*}
    \Pr_{x\la\cD_\cQ(1^n)}[\top\la\Vrfy(x)] \geq 1-n^{-c}
    \end{align*}
    for all sufficiently large $n\in\N$.
    \item[Soundness:] 
    For any $t(n)$-time classical algorithm $\cA$, we have
    \begin{align*}
        \Pr_{x\la\cA(1^n)}[\top\la\Vrfy(x)]\leq n^{-c}
    \end{align*}
    for all sufficiently large $n\in\N$.
    \end{description}
\end{corollary}
\begin{remark}
Let us make three remarks on this result.
First, \cref{thm:QAS_Verification} is incomparable with \cref{thm:OWPuzz_Verify}.
Informally, \cref{thm:QAS_Verification} guarantees that there exists a single verification algorithm testing whether a given sample $x$ is generated by a classical algorithm or by a quantum algorithm that cannot be efficiently simulated by any PPT algorithm.
In contrast, \cref{thm:OWPuzz_Verify} guarantees that, for each arbitrary quantum algorithm $\cD$, there exists a verification algorithm $\Vrfy_\cD$ that can check whether an unknown distribution is close to the specific quantum algorithm $\cD$.
This means that the verification algorithm $\Vrfy_\cD$ might reject $x$ even if it is generated by a quantum algorithm which is hard to simulate by PPT algorithm.
Second, the assumption of the non-existence of OWPuzzs is used only for constructing a QPT verification algorithm.
Without this assumption, our construction still yields a verifier, but it is not efficient.
Importantly, even constructing such a verifier is new to field.
Third, in our correctness requirement, we consider a quantum algorithm $\cD_Q(1^n)$ such that, for any PPT algorithms, $\mathsf{SD}(\cD_\cQ(1^n), \cD(1^n))\geq 1-\negl(n)$.
We can see that there exists such a quantum algorithm under plausible cryptographic assumptions.
More specifically, if quantum pseudorandom generators
(QPRG) secure against QPT algorithms querying an NP oracle exist, then a strong QAS exists\footnote{However  the existence of such QPRGs implies the existence of OWPuzz and therefore the hardness of computing $quK^t$. Consequently in this setting the universal verifier we discuss in this section is only implementable as an inefficient algorithm.}
We refer the proof to \cref{app:strong_qas}.
\end{remark}

\begin{proof}[Proof of \cref{lem:QAS-qkt}]
    We first show quantum easiness.
   \paragraph{Quantum Easiness.} 
       
    From \cref{lem:q_coding}, we have $quK^{t^*(n)}(x|1^n) \leq -\log(\Pr[x \la \cD_Q(1^n)]) +O(1)$.
    On the other hand, we have $uK^{t^*(n)}(x|1^n) = -\log(\Pr[x \la \mathcal{U}^{t^*(n)}(r,1^n)]) $.
    
    Therefore, it is sufficient to prove that, for any $c',c>0$, with probability at least $1-1/n^c$ over $x\la\cD_{\cQ}(1^n)$,
    $$ c'\log(n) \leq -\log(\Pr[x \la \mathcal{U}^{t^*(n)}(r,1^n)])+\log(\Pr[x \la \cD_Q(1^n)])$$
    or equivalently 
    \begin{align}\label{eq:QAS-prb}
    \Pr[x \la \mathcal{U}^{t^*(n)}(r,1^n)]/\Pr[x \la \cD_Q(1^n)] \leq \frac{1}{n^{c'}}.
    \end{align}
    for all sufficiently large $n\in\N$.

    For convenience of notation we define $Q(x) = \Pr[x \la \cD_Q(1^n)] -
    \Pr[x \la \mathcal{U}^{t^*}(r,1^n)]$, $R(x) = \Pr[x \la
    \mathcal{U}^{t^*}(r,1^n)]$ $/ \Pr[x \la \cD_Q(1^n)]$, $S_D = \{x: Q(x) \geq
    0\}$, and $S_{B,c'} = \{x: R(x) \geq 1/n^{c'}\}$.
    For the sake of contradiction with \cref{eq:QAS-prb} we assume that there exists $c>0$
    and $c'>0$ such that 
    \begin{align}\label{eq:contraposition}
    \Pr_{x \la \cD_Q}[x \in S_{B,c'}] \geq 1/n^{c}
    \end{align}
    for infinitely many $n\in\N$.

    By the definition of $S_{B,c'}$ for all $x \in S_{B,c'}$ 
\begin{align*}
\Pr[x \la \mathcal{U}^{t^*}(r,1^n)] \geq \frac{1}{n^{c'}}\Pr[x \la \cD_Q(1^n)].
\end{align*}

    Subtracting the above equation from $\Pr[x \la \cD_Q(1^n)] = \Pr[x \la \cD_Q(1^n)]$ we get that for all $x \in S_{B,c}$

    \begin{align*}
    Q(x) = \Pr[x \la \cD_Q(1^n)] - \Pr[x \la \mathcal{U}^{t^*}(r,1^n)] & \leq \Pr[x \la \cD_Q(1^n)] - \frac{1}{n^{c'}}\Pr[x \la \cD_Q(1^n)] \\
    & = \left(1- \frac{1}{n^{c'}}\right)\Pr[x \la \cD_Q(1^n)] 
    \end{align*}

    From the definition of $\cD_Q$,

    \begin{align*}
    1 - \negl(n) & \leq \Delta(D_Q(1^n),\mathcal{U}^{t^*}(r,1^n)) = \sum_{x \in S_D}Q(x) \\
    & = \sum_{x \in S_D / S_{B,c'}}Q(x) + \sum_{x \in S_D \cap S_{B,c}}Q(x) \\
    & \leq \sum_{x \in S_D / S_{B,c'}}\Pr[x \la \cD_Q(1^n)] + \sum_{x \in S_D \cap S_{B,c'}}Q(x) \\
    & = \Pr_{x \la \cD_Q(1^n)}[x \in S_D / S_{B,c'}] + \sum_{x \in S_D \cap S_{B,c'}}Q(x) \\
    & \leq \Pr_{x \la \cD_Q(1^n)}[x \in S_D / S_{B,c'}] + \Pr_{x \la \cD_Q(1^n)}[x \in S_D \cap S_{B,c'}](1-\frac{1}{n^{c'}}) \\
    & = 1 - \Pr_{x \la \cD_Q(1^n)}[x \in S_D \cap S_{B,c'}] - \Pr_{x \la \cD_Q(1^n)}[x \not\in S_D] \nonumber \\
    & \hspace{3em}+ \Pr_{x \la \cD_Q(1^n)}[x \in S_D \cap S_{B,c'}](1-\frac{1}{n^{c'}}) \\
    & = 1 - \frac{1}{n^{c'}} \Pr_{x \la \cD_Q(1^n)}[x \in S_D \cap S_{B,c'}] - \negl(n) \\
    & \leq 1 - \frac{1}{n^{c+c'}} -\negl(n) < 1- \negl(n)
    \end{align*}
    which is a contradiction.
    This arises from \cref{eq:contraposition}, and hence, for all $c,c'>0$, we have
    \begin{align*}
        \Pr_{x\la\cD_\cQ(1^n)}[c'\log(n)\leq qcd^{t^*(n)}(x|1^n)]\geq 1-n^{-c}
    \end{align*}
    for all sufficiently large $n\in\N$.

    \paragraph{Classical Hardness:}
    For this, it is sufficient to prove that, for any $\alpha>0$, for any PPT algorithm $\cA$, there exists a constant $C$ such that for all x
    \begin{align*}
        \Pr_{x\la\cA(1^n)}[quK^{t^*(n)}(x|1^n)\leq uK^{t^*(n)}(x|1^n)-\alpha]\leq C\abs{x}2^{-\alpha}
    \end{align*}
    for all sufficiently large $n\in\N$.
    From \cref{lem:coding}, for any $t(n)$-time classical algorithm $\cA(1^n)$, there exists a constant $A$ such that
    \begin{align*}
        uK^{t^*(n)}(x|1^n)\leq -\log(\Pr[x\la\cA(1^n)])+ A.
    \end{align*}
    for all sufficiently large $n\in\N$.
    This implies that
    \begin{align*}
        & \Pr_{x\la\cA(1^n)}[quK^{t^*(n)}(x|1^n)\leq uK^{t^*(n)}(x|1^n)-\alpha ]\\
        & \leq \Pr_{x\la\cA(1^n)}[quK^{t^*(n)}(x|1^n)\leq -\log(\Pr[x\la\cA(1^n)])+A-\alpha]
    \end{align*}
    
    Furthermore, from \cref{thm:q_kraft}, for some constant $c$, we have
    \begin{align*}
        \Pr_{x\la\cA(1^n) }[quK^{t^*(n)}(x|1^n)\leq -\log(\Pr[x\la\cA(1^n)])+A-\alpha ]\leq (\abs{x}+c) 2^{A+1-\alpha}\leq C\abs{x}2^{-\alpha} 
    \end{align*}
    for all sufficiently large $n\in\N$.
    This completes the proof.
\end{proof}

\subsection{Proof of \cref{thm:QAS_Verification}}

To show \cref{thm:QAS_Verification}, we additionally use the following \cref{thm:qvk_estimate-avg,thm:pq-ukest}.

\begin{theorem}[No i.o.-pq-OWFs implies $uK^t$ estimation on all PPT distributions~\cite{FOCS:HirNan23}]\label{thm:pq-ukest}
Let $\cD$ be a PPT algorithm running in time at most $t$ that takes $1^n$ as input, and outputs $x\in\bit^{m(n)}$, where $m$ is an arbitrary polynomial.
    Assume that there do not exist infinitely-often post-quantum OWFs.
    %$\not\exists \mathsf{io}$-$\mathsf{OWPuzz}$.
Then, for any constant $c\in\N$ and polynomial $t$, there exists a QPT algorithm $\cM$ such that for all sufficiently large $n\in\N$
\begin{align*}
\Pr_{\substack{x \gets \mathcal{D}(1^n) \\ \cM}}[uK^{t(n)}(x|1^n)\leq \cM(x)\leq uK^{t(n)}(x|1^n)+1]\geq 1-n^{-c}
\end{align*}

Furthermore, for all sufficiently large $n$ and for all $x\in\bit^{m(n)}$,
    \begin{align*}
    \Pr_{\substack{\cM}}[\cM(x)\geq uK^{t(n)}(x|1^n)]\geq 1-n^{-c}.
    \end{align*} 
\end{theorem}

\begin{theorem}\label{thm:qvk_estimate-avg}
Let $\cD$ be a QPT algorithm running in time at most $t(n)$ that takes $1^n$ as input, and outputs $x\in\bit^{m(n)}$, where $m$ is an arbitrary polynomial. Assume that there do not exist infinitely-often OWPuzzs.
Then, for any constant $c\in\N$ and polynomial $t$, there exists a polynomial $t_0$ and a QPT algorithm $\cM$ such that
\begin{align*}
\Pr_{\substack{x \gets \mathcal{D}(1^n) \\ \cM}}[quK^{t_0(t(n))}(x|1^n)\leq \cM(x)\leq quK^{t_0(t(n))}(x|1^n)+1]\geq 1-n^{-c}
\end{align*}
for all sufficiently large $n\in\N$.
\end{theorem}

Let us make two remarks on \cref{thm:qvk_estimate-avg,thm:pq-ukest}.
First, \cref{thm:qvk_estimate-avg} can be viewed as a quantum generalization of \cref{thm:pq-ukest}.
Using a similar idea to that of \cite{FOCS:HirNan23}, we can prove \cref{thm:qvk_estimate-avg}.
For completeness, we present the proof in the end of this section.

Second, a variant of \cref{thm:pq-ukest} was proved in \cite{FOCS:HirNan23}.
The difference is that \cref{thm:pq-ukest} assumes the non-existence of \emph{post-quantum} one-way functions, whereas \cite{FOCS:HirNan23} assumes the non-existence of one-way functions.
We observe that \cref{thm:pq-ukest} can be proved in exactly the same way as in \cite{FOCS:HirNan23}, replacing one-way functions by post-quantum functions.
Therefore, we omit the proof of \Cref{thm:pq-ukest}.
The proof of \cref{thm:qvk_estimate-avg} is also very similar, but we provide a proof at the end of this section for completeness.

In the following, we describe a proof of \cref{thm:QAS_Verification}.

\begin{proof}[Proof of \cref{thm:QAS_Verification}]

Let $\cD_\cQ(1^n)$ be an arbitrary $t(n)$-time quantum algorithm such that,
for any PPT algorithms $\cD$, we have
\begin{align*}
\mathsf{SD}(\cD_Q(1^n),\cD(1^n))\geq 1-\negl(n).
\end{align*}

Let $t^*$ be the polynomial given by \cref{lem:QAS-qkt} for $t$. Let $m(n)\geq n$ be the length of the outputs of the $\cD_\cQ(1^n)$.

$\cM_Q$ is the QPT algorithm given in \cref{thm:qvk_estimate-avg} such that for all sufficiently large $n$
\begin{align*}
\Pr_{\substack{x \gets \mathcal{D}(1^n) \\ \cM_Q}}[quK^{t^*(n)}(x|1^n)\leq \cM_Q(x)\leq quK^{t^*(n)}(x|1^n)+1]\geq 1-n^{-3c}.
\end{align*}

$\cM_C$ is the QPT algorithm given in \cref{thm:pq-ukest} such that for all sufficiently large $n$
\begin{align*}
\Pr_{\substack{x \xleftarrow{\$} \cD_C(1^n) \\ \cM_C}}[uK^{t^*(n)}(x|1^n)\leq \cM_C(x)\leq uK^{t^*(n)}(x|1^n)+1]\geq 1-n^{-3c},
\end{align*}

and
    \begin{align*}
    \Pr_{\cM_C}[\cM_C(x,1^n) \geq uK^{t^*(n)}(x)] \geq 1 - n^{-3c}.
    \end{align*} 

\paragraph{Construction.}
We give a construction of $\Vrfy$.

\begin{description}
    \item[$\Vrfy$:]$ $
    \begin{enumerate}
        \item Take $1^n$ and $x$ as input.
        \item Compute $\cM_C(x)$.
        \item Compute $\cM_Q(x)$.
        \item Output $\top$ if 
        \begin{align*}
          (1+3c) \log(m(n)) \leq \cM_C(x)-\cM_Q(x).
        \end{align*}
        Otherwise, output $\bot$.
    \end{enumerate}
\end{description}
In the following, we show that $\Vrfy$ satisfies correctness and soundness.

\paragraph{Correctness:}
    In the following, we show that
    \begin{equation*}
        \Pr_{x\la\cD_Q(1^n)}[\top\la\Ver(x) ] \geq 1 - n^{-c}
    \end{equation*}
    for all sufficiently large $n\in\N$.

    From \cref{lem:QAS-qkt},  we have
\begin{align*}
    \Pr_{x \la \cD_Q}[ (1+3c)\log(m(n))+1 < uK^t(x|1^n)-quK^t(x|1^n)] \geq 1-n^{-3c}
\end{align*}
for all sufficiently large $n\in\N$.

    From \cref{thm:qvk_estimate-avg} we get that
    \begin{align*}
    \Pr_{\substack{x \gets \mathcal{D}_{\cQ}(1^n) \\ \cM_Q}}[\cM_Q(x)\leq quK^{t(n)}(x|1^n)+1]\geq 1-n^{-3c}
    \end{align*}
    for all sufficiently large $n\in\N$.
    From \cref{thm:pq-ukest} we get that 
    \begin{align*}
        \Pr_{\cM_C}[\cM_C(x) \geq uK^t(x|1^n)] \geq 1-n^{-3c}
    \end{align*}
    for all sufficiently large $n\in\N$.
    
    By union bound over the above three conditions we get that
    \begin{align*}
        \Pr_{\substack{x \la \cD_Q(1^n) \\ \cM_Q \\ \cM_C}}[(1+3c)\log(m(n))\leq \cM_C(x)-\cM_Q(x)] \geq 1-3n^{-3c}
    \end{align*}    
    for all sufficiently large $n\in\N$.

    This implies that 
    \begin{equation*}
        \Pr_{x\la\cD_Q(1^n)}[\top\la\Ver(x) ] \geq 1-3n^{-3c} \geq  1 - n^{-c}
    \end{equation*}
    for all sufficiently large $n\in\N$.

\paragraph{Soundness:}

From \cref{lem:QAS-qkt}, for any $t(n)$-time classical algorithm $\cA(1^n)$, there exists a constant $C$ such that
\begin{align*}
\Pr_{x\la \cA(1^n)}[ (1+3c)\log(m(n))-1\leq uK^{t^*(n)}(x|1^n) -quK^{t^*(n)}(x|1^n)] \leq C\cdot m(n) 2^{ 1- (1+3c)\log(m(n))} 
\end{align*}
for all sufficiently large $n\in\N$.
For any polynomial $m(n)\geq n$, we have 
\begin{align*}
\Pr_{x\la \cA(1^n)}[ (1+3c)\log(m(n))-1\leq uK^{t^*(n)}(x|1^n) -quK^{t^*(n)}(x|1^n)] \leq 2Cn^{-3c} 
\end{align*}
for all sufficiently large $n\in\N$.
The inequality above implies that
\begin{align*}
\Pr_{x\la \cA(1^n)}[(1+3c)\log(m(n))-1\leq  uK^{t^*(n)}(x|1^n) -quK^{t^*(n)}(x|1^n)]  \geq 1- 2Cn^{-3c}
\end{align*}
for all sufficeintly large $n\in\N$.

    From \cref{thm:qvk_estimate-avg} we get that
    \begin{align*}
    \Pr_{\substack{x \gets \cA(1^n) \\ \cM_Q}}[ quK^{t(n)}(x|1^n)\leq \cM_Q(x)]\geq 1-n^{-3c}
    \end{align*}
    and from \cref{thm:pq-ukest} we get that 
    \begin{align*}
        \Pr_{\substack{x \gets \cA(1^n) \\ \cM_C}}[\cM_C(x) \leq uK^t(x|1^n)+1] \geq 1-n^{-3c}.
    \end{align*}
    By union bound over the above three conditions, we have
    \begin{align*}
    \Pr_{\substack{x \gets \cA(1^n) \\ \cM_Q\\\cM_C}}[(1+3c)\log(m(n)) < \cM_C(x)-\cM_Q(x) ] \geq 1-(2+2C) n^{-3c}
    \end{align*}
    for all sufficiently large $n\in\N$.
    This implies that
    \begin{align*}
        \Pr_{x\la \cA(1^n)}[ \top\la\Vrfy(x)]\leq (2+2C) n^{-3c}\leq n^{-c}
    \end{align*}
    for all sufficiently large $n\in\N$.
\end{proof}

\begin{proof}[Proof of \cref{thm:qvk_estimate-avg}]

For showing \cref{thm:qvk_estimate-avg}, we use the following \cref{lem:avg-prob-est}.
\begin{theorem}[\cite{EC:CGGH25,C:HirMor25,STOC:KhuTom25}]\label{lem:avg-prob-est}
Assume that there do not exist infinitely-often OWPuzzs. Then, for any constant $c\in\N$, and for any QPT algorithm $\cD$, 
which takes $1^n$ as input and outputs $x\in\bit^{m(n)}$, where $m$ is a polynomial,
there exists a QPT algorithm $\mathsf{Approx}$ such that
    \begin{align*}
        \Pr_{\substack{x \gets \mathcal{D}(1^n) \\ \mathsf{Approx}}}
\left[
\frac{1}{2}\Pr[x\gets\cD(1^n)] \leq \mathsf{Approx}(x) \leq \Pr[x\gets\cD(1^n)]
\right]
\geq 1 - n^{-c}
    \end{align*}    
    for all sufficiently large $n\in\mathbb{N}$.
\end{theorem}

Assume that infinitely-often OWPuzzs do not exist.
Then, from \cref{lem:avg-prob-est}, there exists a QPT algorithm $\mathsf{Approx}$
\begin{align*}
\Pr_{\substack{x\la\cQ^t(1^n)\\ \mathsf{Approx}}}\left[\frac{1}{2}\Pr[x\la\cQ^{t_0(t(n))}(1^n)]\leq \mathsf{Approx}(x)\leq \Pr[x\la\cQ^{t_0(t(n))}(1^n)]\right]\geq 1-n^{-4c}
\end{align*}
for all sufficiently large $n\in\N$.
We let 
\begin{align*}
\cS_n\seteq \left\{x\in\bit^*: \Pr_{\mathsf{Approx}}\left[\frac{1}{2}\Pr[x\la\cQ^{t_0(t(n))}(1^n)]\leq \mathsf{Approx}(x)\leq \Pr[x\la\cQ^{t_0(t(n))}(1^n)]\right] \leq 1-n^{-2c}\right\}.
\end{align*}
From Markov inequality, we have
\begin{align*}
\Pr_{x\la\cQ^{t_0(t(n))}(1^n)}\left[x\in\cS_n\right]\leq n^{-2c}.
\end{align*}
This implies that, for any QPT algorithm $\cD$ that takes $1^n$ as input, and outputs $x\in\bit^{m(n)}$ running in time at most $t(n)$, there exists a constant $D$ such that
\begin{align*}
n^{-2c}
&\geq \Pr_{x\la\cQ^{t_0(t(n))}(1^n)}\left[x\in\cS_n\right]\\
&=\sum_{x\in\cS_n}\Pr[x\la\cQ^{t_0(t(n))}(1^n)]\\
&\geq \sum_{x\in\cS_n}\frac{\Pr[x\la\cD(1^n)]}{D}
\end{align*}
for all sufficiently large $n\in\N$.
Here, in the last inequality, we have used \cref{lem:q_coding}.
Therefore, we have
\begin{align*}    
\Pr_{x\la\cD(1^n)}[x\in\cS_n]\leq D\cdot n^{-2c}
\end{align*}
for all sufficiently large $n\in\N$.
This implies that 
\begin{align*}
\Pr_{x\la\cD(1^n)}\left[\Pr_{\mathsf{Approx}}\left[\frac{1}{2}\Pr[x\la\cQ^{t_0(t(n))}(1^n)]\leq \mathsf{Approx}(x)\leq \Pr[x\la\cQ^{t_0(t(n))}(1^n)]\right] \geq 1-n^{-2c} \right]\geq 1-D\cdot n^{-2c}
\end{align*}
for all sufficiently large $n\in\N$.
Hence, we have
\begin{align*}
\Pr_{\substack{x\la\cD(1^n)\\\mathsf{Approx}}}\left[\frac{1}{2}\Pr[x\la\cQ^{t_0(t(n))}(1^n)]\leq \mathsf{Approx}(x)\leq \Pr[x\la\cQ^{t_0(t(n))}(1^n)]\right]&\geq 1-(D+1)\cdot n^{-2c}\\
&\geq 1-n^{-c}
\end{align*}
for all sufficiently large $n\in\N$.
\end{proof}

\ifnum\anonymous=0
\paragraph{Acknowledgments.}

BC acknowledges support of Royal Society University Research Fellowship URF$\backslash$R1$\backslash$211106
and
EPSRC project EP/Z534158/1 on ``Integrated Approach to Computational Complexity: Structure, Self-Reference and Lower Bounds''.
TM is supported by
JST CREST JPMJCR23I3,
JST Moonshot R\verb|&|D JPMJMS2061-5-1-1, 
JST FOREST, 
MEXT QLEAP, 
the Grant-in Aid for Transformative Research Areas (A) 21H05183,
and 
the Grant-in-Aid for Scientific Research (A) No.22H00522. EG is supported by a NSF graduate research fellowship.
This work was done in part while Matthew Gray and Eli Goldin were visiting the Simons Institute for the Theory of Computing.
\else
\fi

\ifnum\llncs=1
\bibliographystyle{alpha} 
\bibliography{references/abbrev3,references/crypto,references/refs}
\else
\bibliographystyle{alpha} 
\bibliography{references/abbrev3,references/crypto,references/refs}

\newcommand{\etalchar}[1]{$^{#1}$}
\begin{thebibliography}{AFvMV06}

\bibitem[AA11]{STOC:AarArk11}
Scott Aaronson and Alex Arkhipov.
\newblock The computational complexity of linear optics.
\newblock In Lance Fortnow and Salil~P. Vadhan, editors, {\em 43rd ACM STOC},
  pages 333--342. {ACM} Press, June 2011.

\bibitem[Aar14]{Aaronson14}
Scott Aaronson.
\newblock The equivalence of sampling and searching.
\newblock {\em Theor. Comp. Sys.}, 55(2):281–298, August 2014.

\bibitem[ABCL25]{ABCL25}
Maryam Aliakbarpour, Vladimir Braverman, Nai-Hui Chia, and Yuhan Liu.
\newblock Adversarially robust quantum state learning and testing, 2025.

\bibitem[AC17]{CCC:AarChe17}
Scott Aaronson and Lijie Chen.
\newblock Complexity-theoretic foundations of quantum supremacy experiments.
\newblock CCC'17: Proceedings of the 32nd Computational Complexity Conference,
  2017.

\bibitem[ACC{\etalchar{+}}22]{DBLP:conf/crypto/AustrinCCFLM22}
Per Austrin, Hao Chung, Kai{-}Min Chung, Shiuan Fu, Yao{-}Ting Lin, and
  Mohammad Mahmoody.
\newblock On the impossibility of key agreements from quantum random oracles.
\newblock In Yevgeniy Dodis and Thomas Shrimpton, editors, {\em Advances in
  Cryptology - {CRYPTO} 2022 - 42nd Annual International Cryptology Conference,
  {CRYPTO} 2022, Santa Barbara, CA, USA, August 15-18, 2022, Proceedings, Part
  {II}}, Lecture Notes in Computer Science, pages 165--194. Springer, 2022.

\bibitem[AFvMV06]{AntFor06}
Luis Antunes, Lance Fortnow, Dieter van Melkebeek, and N.~V. Vinodchandran.
\newblock Computational depth: concept and applications.
\newblock {\em Theor. Comput. Sci.}, 354(3):391–404, April 2006.

\bibitem[AG20]{AaronsonGunn}
Scott Aaronson and Sam Gunn.
\newblock On the classical hardness of spoofing linear cross-entropy
  benchmarking, 2020.

\bibitem[AH23]{STOC:AarHun23}
Scott Aaronson and Shih-Han Hung.
\newblock Certified randomness from quantum supremacy.
\newblock In Barna Saha and Rocco~A. Servedio, editors, {\em 55th ACM STOC},
  pages 933--944. {ACM} Press, June 2023.

\bibitem[AQY22]{C:AnaQiaYue22}
Prabhanjan Ananth, Luowen Qian, and Henry Yuen.
\newblock Cryptography from pseudorandom quantum states.
\newblock In Yevgeniy Dodis and Thomas Shrimpton, editors, {\em CRYPTO~2022,
  Part~I}, volume 13507 of {\em {LNCS}}, pages 208--236. Springer, Cham, August
  2022.

\bibitem[BBF{\etalchar{+}}26]{BBFGT26}
Roozbeh Bassirian, Adam Bouland, Bill Fefferman, Sam Gunn, and Avishay Tal.
\newblock On {C}ertified {R}andomness from {F}ourier {S}ampling or {R}andom
  {C}ircuit {S}ampling.
\newblock {\em {Quantum}}, 10:2002, February 2026.

\bibitem[BCKM21]{C:BCKM21b}
James Bartusek, Andrea Coladangelo, Dakshita Khurana, and Fermi Ma.
\newblock One-way functions imply secure computation in a quantum world.
\newblock In Tal Malkin and Chris Peikert, editors, {\em CRYPTO~2021, Part~I},
  volume 12825 of {\em {LNCS}}, pages 467--496, Virtual Event, August 2021.
  Springer, Cham.

\bibitem[BCM{\etalchar{+}}18]{FOCS:BCMVV18}
Zvika Brakerski, Paul Christiano, Urmila Mahadev, Umesh~V. Vazirani, and Thomas
  Vidick.
\newblock A cryptographic test of quantumness and certifiable randomness from a
  single quantum device.
\newblock In Mikkel Thorup, editor, {\em 59th FOCS}, pages 320--331. {IEEE}
  Computer Society Press, October 2018.

\bibitem[BCM{\etalchar{+}}21]{JACM:BCMVV21}
Zvika Brakerski, Paul Christiano, Urmila Mahadev, Umesh Vazirani, and Thomas
  Vidick.
\newblock A cryptographic test of quantumness and certifiable randomness from a
  single quantum device.
\newblock {\em J. ACM}, 68(5), August 2021.

\bibitem[BCQ23]{ITCS:BCQ23}
Zvika Brakerski, Ran Canetti, and Luowen Qian.
\newblock On the computational hardness needed for quantum cryptography.
\newblock In Yael~Tauman Kalai, editor, {\em 14th Innovations in Theoretical
  Computer Science Conference, {ITCS} 2023, January 10-13, 2023, MIT,
  Cambridge, Massachusetts, {USA}}, volume 251 of {\em LIPIcs}, pages
  24:1--24:21. Schloss Dagstuhl - Leibniz-Zentrum f{\"{u}}r Informatik, 2023.

\bibitem[BFF{\etalchar{+}}01]{optidtest}
T.~Batu, E.~Fischer, L.~Fortnow, R.~Kumar, R.~Rubinfeld, and P.~White.
\newblock Testing random variables for independence and identity.
\newblock In {\em Proceedings 42nd IEEE Symposium on Foundations of Computer
  Science}, pages 442--451, 2001.

\bibitem[BFR{\etalchar{+}}00]{idtestclose}
T.~Batu, L.~Fortnow, R.~Rubinfeld, W.D. Smith, and P.~White.
\newblock Testing that distributions are close.
\newblock In {\em Proceedings 41st Annual Symposium on Foundations of Computer
  Science}, pages 259--269, 2000.

\bibitem[BJS11]{BreJozShe10}
Michael~J. Bremner, Richard Jozsa, and Dan~J. Shepherd.
\newblock Classical simulation of commuting quantum computations implies
  collapse of the polynomial hierarchy.
\newblock {\em Proceedings of the Royal Society A: Mathematical, Physical and
  Engineering Sciences}, 467:459--472, 2011.

\bibitem[BMS16]{BreMonShe16}
Michael~J. Bremner, Ashley Montanaro, and Dan~J. Shepherd.
\newblock Average-case complexity versus approximate simulation of commuting
  quantum computations.
\newblock {\em Physical Review Letters}, 117:080501, 2016.

\bibitem[BQSY24]{STOC:BQSY24}
John Bostanci, Luowen Qian, Nicholas Spooner, and Henry Yuen.
\newblock An efficient quantum parallel repetition theorem and applications.
\newblock In Bojan Mohar, Igor Shinkar, and Ryan {O'Donnell}, editors, {\em
  56th ACM STOC}, pages 1478--1487. {ACM} Press, June 2024.

\bibitem[Can20]{Canonne20}
Cl{\'e}ment~L. Canonne.
\newblock A survey on distribution testing: Your data is big. but is it blue?
\newblock {\em Electron. Colloquium Comput. Complex.}, TR15, 2020.

\bibitem[CGG24]{C:ChuGolGra24}
Kai-Min Chung, Eli Goldin, and Matthew Gray.
\newblock On central primitives for quantum cryptography with classical
  communication.
\newblock In Leonid Reyzin and Douglas Stebila, editors, {\em CRYPTO~2024,
  Part~VII}, volume 14926 of {\em {LNCS}}, pages 215--248. Springer, Cham,
  August 2024.

\bibitem[CGG{\etalchar{+}}25]{Cavalar_2025}
Bruno Cavalar, Eli Goldin, Matthew Gray, Peter Hall, Yanyi Liu, and Angelos
  Pelecanos.
\newblock On the computational hardness of quantum one-wayness.
\newblock {\em Quantum}, 9:1679, March 2025.

\bibitem[CGGH25]{EC:CGGH25}
Bruno~Pasqualotto Cavalar, Eli Goldin, Matthew Gray, and Peter Hall.
\newblock A meta-complexity characterization of quantum cryptography.
\newblock {LNCS}, pages 82--107. Springer, Cham, June 2025.

\bibitem[Cha69]{Chai69}
Gregory~J. Chaitin.
\newblock On the simplicity and speed of programs for computing infinite sets
  of natural numbers.
\newblock {\em J. ACM}, 16(3):407–422, July 1969.

\bibitem[CHK25]{CHK25}
Suvradip Chakraborty, James Hulett, and Dakshita Khurana.
\newblock On weak nizks, one-way functions and amplification.
\newblock In Yael Tauman~Kalai and Seny~F. Kamara, editors, {\em Advances in
  Cryptology -- CRYPTO 2025}, pages 580--610, Cham, 2025. Springer Nature
  Switzerland.

\bibitem[DKN15]{structidtest}
Ilias Diakonikolas, Daniel~M. Kane, and Vladimir Nikishkin.
\newblock Testing identity of structured distributions.
\newblock In {\em Proceedings of the Twenty-Sixth Annual ACM-SIAM Symposium on
  Discrete Algorithms}, SODA '15, page 1841–1854, USA, 2015. Society for
  Industrial and Applied Mathematics.

\bibitem[FKMO24]{Fawzi_2024}
Omar Fawzi, Richard Kueng, Damian Markham, and Aadil Oufkir.
\newblock Learning properties of quantum states without the iid assumption.
\newblock {\em Nature Communications}, 15(1), November 2024.

\bibitem[FR99]{FR99}
Lance Fortnow and John Rogers.
\newblock Complexity limitations on quantum computation.
\newblock {\em Journal of Computer and System Sciences}, 59(2):240--252, 1999.

\bibitem[G{\'a}c01]{Gac01}
Peter G{\'a}cs.
\newblock Quantum algorithmic entropy.
\newblock {\em Journal of Physics A: Mathematical and General}, 34(35):6859,
  August 2001.

\bibitem[GGM84]{C:GolGolMic84}
Oded Goldreich, Shafi Goldwasser, and Silvio Micali.
\newblock On the cryptographic applications of random functions.
\newblock In G.~R. Blakley and David Chaum, editors, {\em CRYPTO'84}, volume
  196 of {\em {LNCS}}, pages 276--288. Springer, Berlin, Heidelberg, August
  1984.

\bibitem[GKLO22]{CCC:GKLO22}
Halley Goldberg, Valentine Kabanets, Zhenjian Lu, and Igor~C. Oliveira.
\newblock Probabilistic kolmogorov complexity with applications to average-case
  complexity.
\newblock In {\em Proceedings of the 37th Computational Complexity Conference},
  CCC '22, Dagstuhl, DEU, 2022. Schloss Dagstuhl--Leibniz-Zentrum fuer
  Informatik.

\bibitem[GLSV21]{EC:GLSV21}
Alex~B. Grilo, Huijia Lin, Fang Song, and Vinod Vaikuntanathan.
\newblock Oblivious transfer is in {MiniQCrypt}.
\newblock In Anne Canteaut and Fran\c{c}ois-Xavier Standaert, editors, {\em
  EUROCRYPT~2021, Part~II}, volume 12697 of {\em {LNCS}}, pages 531--561.
  Springer, Cham, October 2021.

\bibitem[GV04]{GV04}
Peter Grunwald and Paul Vitanyi.
\newblock Shannon information and kolmogorov complexity, 2004.

\bibitem[HHM]{HHM_inpreparation}
Taiga Hiroka, Shih-Han Hung, and Tomoyuki Morimae.
\newblock In preparation.

\bibitem[HILL99]{Hill99}
Johan H{\aa}stad, Russell Impagliazzo, Leonid~A. Levin, and Michael Luby.
\newblock A pseudorandom generator from any one-way function.
\newblock {\em {SIAM} Journal on Computing}, 28(4):1364--1396, 1999.

\bibitem[Hir21]{STOC:Hir21}
Shuichi Hirahara.
\newblock Average-case hardness of np from exponential worst-case hardness
  assumptions.
\newblock In {\em Proceedings of the 53rd Annual ACM SIGACT Symposium on Theory
  of Computing}, STOC 2021, page 292–302, New York, NY, USA, 2021.
  Association for Computing Machinery.

\bibitem[HKEG19]{HKEG19}
Dominik Hangleiter, Martin Kliesch, Jens Eisert, and Christian Gogolin.
\newblock Sample complexity of device-independently certified ``quantum
  supremacy''.
\newblock {\em Phys. Rev. Lett.}, 122:210502, May 2019.

\bibitem[HKNY24]{TCC:HKNY24}
Taiga Hiroka, Fuyuki Kitagawa, Ryo Nishimaki, and Takashi Yamakawa.
\newblock Robust combiners and universal constructions for quantum
  cryptography.
\newblock In {\em TCC~2024, Part~II}, {LNCS}, pages 126--158. Springer, Cham,
  November 2024.

\bibitem[HM25]{C:HirMor25}
Taiga Hiroka and Tomoyuki Morimae.
\newblock Quantum cryptography and meta-complexity.
\newblock In Yael Tauman~Kalai and Seny~F. Kamara, editors, {\em Advances in
  Cryptology -- CRYPTO 2025}, pages 545--574, Cham, 2025. Springer Nature
  Switzerland.

\bibitem[HN23]{FOCS:HirNan23}
Shuichi Hirahara and Mikito Nanashima.
\newblock Learning in pessiland via inductive inference.
\newblock In {\em 64th FOCS}, pages 447--457. {IEEE} Computer Society Press,
  November 2023.

\bibitem[IL89]{FOCS:ImpLub89}
Russell Impagliazzo and Michael Luby.
\newblock One-way functions are essential for complexity based cryptography
  (extended abstract).
\newblock In {\em 30th FOCS}, pages 230--235. {IEEE} Computer Society Press,
  October~/~November 1989.

\bibitem[IL90]{FOCS:ImpLev90}
Russell Impagliazzo and Leonid~A. Levin.
\newblock No better ways to generate hard {NP} instances than picking uniformly
  at random.
\newblock In {\em 31st FOCS}, pages 812--821. {IEEE} Computer Society Press,
  October 1990.

\bibitem[IRS22]{STOC:IlaRenSan22}
Rahul Ilango, Hanlin Ren, and Rahul Santhanam.
\newblock Robustness of average-case meta-complexity via pseudorandomness.
\newblock In Stefano Leonardi and Anupam Gupta, editors, {\em 54th ACM STOC},
  pages 1575--1583. {ACM} Press, June 2022.

\bibitem[JLS18]{C:JiLiuSon18}
Zhengfeng Ji, Yi-Kai Liu, and Fang Song.
\newblock Pseudorandom quantum states.
\newblock In Hovav Shacham and Alexandra Boldyreva, editors, {\em CRYPTO~2018,
  Part~III}, volume 10993 of {\em {LNCS}}, pages 126--152. Springer, Cham,
  August 2018.

\bibitem[KAF17]{KA17}
Max Kessler and Rotem Arnon-Friedman.
\newblock Device-independent randomness amplification and privatization.
\newblock {\em IEEE Journal on Selected Areas in Information Theory},
  1:568--584, 2017.

\bibitem[KMR{\etalchar{+}}94]{KMRRSS94}
Michael Kearns, Yishay Mansour, Dana Ron, Ronitt Rubinfeld, Robert~E. Schapire,
  and Linda Sellie.
\newblock On the learnability of discrete distributions.
\newblock In {\em Proceedings of the Twenty-Sixth Annual ACM Symposium on
  Theory of Computing}, STOC '94, page 273–282, New York, NY, USA, 1994.
  Association for Computing Machinery.

\bibitem[Kol68]{Kolmogorov1968ThreeAT}
Andrei~N. Kolmogorov.
\newblock Three approaches to the quantitative definition of information.
\newblock {\em International Journal of Computer Mathematics}, 2:157--168,
  1968.

\bibitem[KR21]{PRXQuantum.2.010201}
Martin Kliesch and Ingo Roth.
\newblock Theory of quantum system certification.
\newblock {\em PRX Quantum}, 2:010201, Jan 2021.

\bibitem[KT24]{STOC:KhuTom24}
Dakshita Khurana and Kabir Tomer.
\newblock Commitments from quantum one-wayness.
\newblock In Bojan Mohar, Igor Shinkar, and Ryan {O'Donnell}, editors, {\em
  56th ACM STOC}, pages 968--978. {ACM} Press, June 2024.

\bibitem[KT25]{STOC:KhuTom25}
Dakshita Khurana and Kabir Tomer.
\newblock Founding quantum cryptography on quantum advantage, or, towards
  cryptography from {\#p} hardness.
\newblock In {\em 57th ACM STOC}, pages 178--188. {ACM} Press, June 2025.

\bibitem[LMW24]{STOC:LomMaWri24}
Alex Lombardi, Fermi Ma, and John Wright.
\newblock A one-query lower bound for unitary synthesis and breaking quantum
  cryptography.
\newblock In Bojan Mohar, Igor Shinkar, and Ryan {O'Donnell}, editors, {\em
  56th ACM STOC}, pages 979--990. {ACM} Press, June 2024.

\bibitem[LO22]{ZO22}
Zhenjian Lu and Igor~C. Oliveira.
\newblock Theory and applications of probabilistic kolmogorov complexity, 2022.

\bibitem[LOS21]{STOC:ZOS21}
Zhenjian Lu, Igor~C. Oliveira, and Rahul Santhanam.
\newblock Pseudodeterministic algorithms and the structure of probabilistic
  time.
\newblock In {\em Proceedings of the 53rd Annual ACM SIGACT Symposium on Theory
  of Computing}, STOC 2021, page 303–316, New York, NY, USA, 2021.
  Association for Computing Machinery.

\bibitem[LP20]{FOCS:LiuPas20}
Yanyi Liu and Rafael Pass.
\newblock On one-way functions and kolmogorov complexity.
\newblock In {\em 61st FOCS}, pages 1243--1254. {IEEE} Computer Society Press,
  November 2020.

\bibitem[LP21]{C:LiuPas21}
Yanyi Liu and Rafael Pass.
\newblock On the possibility of basing cryptography on {$\mathsf{EXP}\ne
  \mathsf{BPP}$}.
\newblock In Tal Malkin and Chris Peikert, editors, {\em CRYPTO~2021, Part~I},
  volume 12825 of {\em {LNCS}}, pages 11--40, Virtual Event, August 2021.
  Springer, Cham.

\bibitem[LP24]{TCC:LiuPas24}
Yanyi Liu and Rafael Pass.
\newblock On one-way functions, the worst-case hardness of time-bounded
  kolmogorov complexity, and computational depth.
\newblock In {\em TCC~2024, Part~I}, {LNCS}, pages 222--252. Springer, Cham,
  November 2024.

\bibitem[LV93]{LiVitanyi93}
Ming Li and Paul Vit\'{a}nyi.
\newblock {\em An introduction to Kolmogorov complexity and its applications}.
\newblock Springer-Verlag, Berlin, Heidelberg, 1993.

\bibitem[MSY25]{STOC:MorShiYam25}
Tomoyuki Morimae, Yuki Shirakawa, and Takashi Yamakawa.
\newblock Cryptographic characterization of quantum advantage.
\newblock In {\em 57th ACM STOC}, pages 1863--1874. {ACM} Press, June 2025.

\bibitem[MTH17]{MorTakHay17}
Tomoyuki Morimae, Yuki Takeuchi, and Masahito Hayashi.
\newblock Verification of hypergraph states.
\newblock {\em Phys. Rev. A}, 96:062321, Dec 2017.

\bibitem[MW16]{gs007}
Ashley Montanaro and Ronald~{de} Wolf.
\newblock {\em A Survey of Quantum Property Testing}.
\newblock Number~7 in Graduate Surveys. Theory of Computing Library, 2016.

\bibitem[MY22]{C:MorYam22}
Tomoyuki Morimae and Takashi Yamakawa.
\newblock Quantum commitments and signatures without one-way functions.
\newblock In Yevgeniy Dodis and Thomas Shrimpton, editors, {\em CRYPTO~2022,
  Part~I}, volume 13507 of {\em {LNCS}}, pages 269--295. Springer, Cham, August
  2022.

\bibitem[Nao90]{C:Naor89}
Moni Naor.
\newblock Bit commitment using pseudo-randomness.
\newblock In Gilles Brassard, editor, {\em CRYPTO'89}, volume 435 of {\em
  {LNCS}}, pages 128--136. Springer, New York, August 1990.

\bibitem[NC11]{NielsenChuang}
Michael~A. Nielsen and Isaac~L. Chuang.
\newblock {\em Quantum Computation and Quantum Information: 10th Anniversary
  Edition}.
\newblock Cambridge University Press, USA, 10th edition, 2011.

\bibitem[NR06]{NaoRot06}
Moni Naor and Guy~N. Rothblum.
\newblock Learning to impersonate.
\newblock In {\em Proceedings of the 23rd International Conference on Machine
  Learning}, ICML '06, page 649–656, New York, NY, USA, 2006. Association for
  Computing Machinery.

\bibitem[Pan08]{idtestlower}
Liam Paninski.
\newblock A coincidence-based test for uniformity given very sparsely sampled
  discrete data.
\newblock {\em IEEE Transactions on Information Theory}, 54(10):4750--4755,
  2008.

\bibitem[PFMO25]{PFMO25}
Giacomo~De Palma, Marco Fanizza, Connor Mowry, and Ryan O'Donnell.
\newblock Non-iid hypothesis testing: from classical to quantum, 2025.

\bibitem[Rom90]{STOC:Rompel90}
John Rompel.
\newblock One-way functions are necessary and sufficient for secure signatures.
\newblock In {\em 22nd ACM STOC}, pages 387--394. {ACM} Press, May 1990.

\bibitem[Sol64a]{SOLOMONOFF19641}
R.J. Solomonoff.
\newblock A formal theory of inductive inference. part i.
\newblock {\em Information and Control}, 7(1):1--22, 1964.

\bibitem[Sol64b]{Solomonoff1964_2}
R.J. Solomonoff.
\newblock A formal theory of inductive inference. part ii.
\newblock {\em Information and Control}, 7(2):224--254, 1964.

\bibitem[TM18]{MorTak18}
Yuki Takeuchi and Tomoyuki Morimae.
\newblock Verification of many-qubit states.
\newblock {\em Physical Review X}, 8(2), June 2018.

\bibitem[Val84]{Valiant84}
L.~G. Valiant.
\newblock A theory of the learnable.
\newblock {\em Commun. ACM}, 27(11):1134–1142, November 1984.

\bibitem[VV12]{STOC:VazVid12}
Umesh~V. Vazirani and Thomas Vidick.
\newblock Certifiable quantum dice: or, true random number generation secure
  against quantum adversaries.
\newblock In Howard~J. Karloff and Toniann Pitassi, editors, {\em 44th ACM
  STOC}, pages 61--76. {ACM} Press, May 2012.

\bibitem[VV17]{idtestupper}
Gregory Valiant and Paul Valiant.
\newblock An automatic inequality prover and instance optimal identity testing.
\newblock {\em SIAM Journal on Computing}, 46(1):429--455, 2017.

\bibitem[Yan22]{AC:Yan22}
Jun Yan.
\newblock General properties of quantum bit commitments (extended abstract).
\newblock In Shweta Agrawal and Dongdai Lin, editors, {\em ASIACRYPT~2022,
  Part~IV}, volume 13794 of {\em {LNCS}}, pages 628--657. Springer, Cham,
  December 2022.

\bibitem[YZ22]{FOCS:YamZha22}
Takashi Yamakawa and Mark Zhandry.
\newblock Verifiable quantum advantage without structure.
\newblock In {\em 63rd FOCS}, pages 69--74. {IEEE} Computer Society Press,
  October~/~November 2022.

\bibitem[ZH19]{ZhuHay19}
Huangjun Zhu and Masahito Hayashi.
\newblock General framework for verifying pure quantum states in the
  adversarial scenario.
\newblock {\em Physical Review A}, 100(6), December 2019.

\end{thebibliography}
\fi

\appendix
\section{Proof of \cref{lem:modaaronson}}\label{appendix}

\taiga{I checked the proof. The proof seems fine, but it is hard to follow. So, we need to elaborate if we have time.}
\bnote{A reviewer also complained that it is difficult to follow (particularly after applying Pinsker).}

\begin{proof}
    Define the distribution $\cB_n$ to be the distribution $\cG_n$ conditioned on outputting elements in $A_{n,s,\alpha}^{\cD}$.

    For any $Y=(y_1,\dots,y_s)$, define 
    $$p_Y\coloneqq \Pr[\cD_n^{\otimes s} \to Y]$$
    $$q_Y\coloneqq \Pr[\cG_n \to Y]$$
    and
    $$\widetilde{q}_Y\coloneqq \Pr[\cB_n \to Y]$$

    ~\Cref{lem:coding} tells us that for all $Y$,
    \begin{equation}\label{eq:aarlem1}
        uK^{t_0(t(n))}(Y|1^n) \leq \log_2 \frac{1}{q_Y} + |\cG| + O(1).
    \end{equation}

    Similarly, for all $Y\in A_{n,s,\alpha}^{\cD}$, we have
    \begin{equation}\label{eq:aarlem2}
        \log_2 \frac{1}{p_{y_1}\dots p_{y_s}} \leq uK^{t_0(t(n))}(Y|1^n) + \alpha(n).
    \end{equation}

    And so~\Cref{eq:aarlem1,eq:aarlem2} together give that for all $Y\in A_{n,s,\alpha}^{\cD}$,
    \begin{equation*}
        \log_2 \frac{q_Y}{p_{Y}} \leq |\cG| + O(1) + \alpha(n)
    \end{equation*}

    By Bayes law, for all $Y$,
    \begin{equation*}
        \widetilde{q}_Y \leq \frac{q_Y}{1-\epsilon(n)}
    \end{equation*}

    and so we also get
    \begin{equation*}
        \log_2 \frac{\widetilde{q}_Y}{p_{Y}} \leq \log_2 \frac{1}{1-\epsilon(n)} + \log_2 \frac{{q}_Y}{p_{Y}}\leq \log_2 \frac{1}{1-\epsilon(n)} + |\cG| + O(1) + \alpha(n)
    \end{equation*}

    We can explicitly bound $\SD(\cB_n,\cD_n^{\otimes t})$ by working with KL-divergence.
    \begin{align*}
        \begin{split}
            D_{KL}(\cB_n || \cD_n^{\otimes s}) &= \sum_{Y\in [2^n]^s} \widetilde{q}_Y \log_2 \frac{\widetilde{q}_Y}{p_Y}\\
            &\leq \max_{Y \in A_{n,s,\delta}^{\cD}} \log_2 \frac{\widetilde{q}_Y}{p_Y}\\
            &\leq \log_2 \frac{1}{1-\epsilon(n)} + |\cG| + O(1) + \alpha(n)
        \end{split}
    \end{align*}

    Let $\cB_n^i$ be the marginal distribution of $\cB_n$ restricted to the $i$th coordinate.
    We have
    \begin{align*}
        \sum_{i=1}^s D_{KL}(\cB_n^i || \cD_n) \leq D_{KL}(\cB_n||\cD_n^{\otimes s}) \leq \log_2 \frac{1}{1-\epsilon(n)} + |\cG| + O(1) + \alpha(n).
    \end{align*}
    From Pinsker's inequality, we have
    \begin{align*}
        \sum_{i=1}^s \frac{1}{2}\SD(\cB_n^i, \cD_n)^2\leq D_{KL}(\cB_n||\cD_n^{\otimes s}) \leq \log_2 \frac{1}{1-\epsilon(n)} + |\cG| + O(1) + \alpha(n).
    \end{align*}
    From Cauchy-Schwarz inequality, we obtain 
    \begin{align*}
        \sum_{i=1}^s\Delta(\cB_n^i,\cD_n)\leq \sqrt{s\sum_{i=1}^s \Delta(\cB_n^i,\cD_n)^2}\leq \sqrt{2s\left(\log_2 \frac{1}{1-\epsilon(n)} + |\cG| + O(1) + \alpha(n)\right)}.
    \end{align*}
    This implies that
    \begin{align*}
        \Delta(\mathsf{Marginal}_{\cB}(1^n),\cD_n)=\frac{1}{s}\sum_{i=1}^s\Delta(\cB_n^i,\cD_n)\leq\sqrt{\frac{2}{s}\left(\log_2 \frac{1}{1-\epsilon(n)} + |\cG| + O(1) + \alpha(n)\right)}.
    \end{align*}
    From triangle inequality, we have
    \begin{align*}
    \Delta(\mathsf{Marginal}_{\cG}(1^n),\cD_n)
    &\leq \Delta(\mathsf{Marginal}_{\cG}(1^n),\mathsf{Marginal}_{\cB}(1^n))+\Delta(\mathsf{Marginal}_{\cB},\cD_n)\\
    &\leq \epsilon(n)+\sqrt{\frac{2}{s}\left(\log_2 \frac{1}{1-\epsilon(n)} + |\cG| + O(1) + \alpha(n)\right)}.
    \end{align*}
\end{proof}

\section{Strong Quantum Advantage Sampler from Quantum Cryptography}\label{app:strong_qas}

In this section, we prove \cref{thm:strong_QAS_from_crypto}, which states that a strong quantum advantage samples exists from a plausible cryptographic assumption.

\begin{definition}
A strong quantum advantage sampler is a QPT algorithm $\cD_\cQ$ taking $1^n$ as input, and outputting strings $x\in\bit^{m(n)}$, where $m$ is an arbitrary polynomial, such that for all PPT algorithms $\cD_\cC$, we have
\begin{align*}
\mathsf{SD}(\cD_\cQ(1^n),\cD_\cC(1^n))\geq 1-\negl(n).
\end{align*}
\end{definition}

\begin{theorem}\label{thm:strong_QAS_from_crypto}
If there exists a QPRG secure against QPT algorithms querying to an $\mathbf{NP}$ oracle, then a strong QAS exists.
\end{theorem}

Here, a QPRG secure against QPT algorithms querying to an $\mathbf{NP}$ oracle is defined as follows.

\begin{definition}[Quantum Pseudorandom Generator secure against QPT algorithms querying to an $\mathbf{NP}$ oracle.]
Let $\Gen$ be a QPT algorithm that takes $1^n$ as input, and outputs $x\in\bit^n$.
We say that $\Gen$ is a quantum pseudorandom generator (QPRG) secure against QPT algorithms querying to an $\mathbf{NP}$ oracle if the following holds:
\begin{description}
\item[Statistically far:]
\begin{align*}
\Delta(\Gen(1^n),U_n)\geq 1-\negl(n)
\end{align*}
\item[Computationally indistinguishable:]
For any QPT adversary $\cA$ querying to an $\mathbf{NP}$ oracle, we have
\begin{align*}
\abs{\Pr_{x\la\Gen(1^n)}[1\la\cA^{\mathbf{NP}}(x)]-\Pr_{x\la U_n}[1\la\cA^{\mathbf{NP}}(x)]}\leq \negl(n)
\end{align*}
\end{description}
\end{definition}
For showing ~\cref{thm:strong_QAS_from_crypto}, we use the following Lemma~\ref{lem:amp}. Lemma~\ref{lem:amp} directly follows from a standard padding argument and parallel repetition, and this is also proven in \cite{C:HirMor25}. 
For clarity, we describe the proof in the end of this section.
\begin{lemma}\label{lem:amp}
Suppose that there exists a QPRG secure against QPT algorithms querying to $\mathbf{NP}$ oracle.
Then, for any $0<\tau<1$, there exists a $(1-2^{-n^\tau})$-statistically-far QPRG $\Gen^*$ such that the following holds:
\begin{description}
\item[$(1-2^{-n^\tau})$-statistically far:]
\begin{align*}
\Delta(\Gen^*(1^n),U_x)\geq 1-2^{-n^{\tau}}
\end{align*}
\item[Computationally indistinguishable:]
For any QPT adversary $\cA$ querying to $\mathbf{NP}$ oracle, we have
\begin{align*}
\abs{\Pr_{x\la\Gen^*(1^n)}[1\la\cA^{\mathbf{NP}}(x)]-\Pr_{x\la U_n}[1\la\cA^{\mathbf{NP}}(x)]}\leq \negl(n)
\end{align*}
\end{description}
\end{lemma}

\begin{proof}[Proof of Lemma~\ref{thm:strong_QAS_from_crypto}]
From \cref{lem:amp}, it is sufficient to construct a strong QAS from a $(1-2^{-n^{\tau}})$-statistically-far QPRG secure against QPT algorithms querying to an $\mathbf{NP}$ oracle for some $0<\tau<1$.

For contradiction, let us assume that there does not exist a strong QAS.
Then, we show that there does not exist a $(1-2^{-n^\tau})$-statistically-far QPRG secure against QPT algorithms querying to an $\mathbf{NP}$ oracle for any $0<\tau<1$.
More specifically, for any QPT algorithm $\cQ$ and any constant $0<\tau<1$ such that
\begin{align*}
\Delta(\cQ(1^n),U_n)\geq 1-2^{-n^\tau} 
\end{align*}
for all $n\in\N$, we construct a QPT algorithm $\cA$ querying to an $\mathbf{NP}$ oracle such that
\begin{align*}
\abs{\Pr_{x\la\cQ(1^n)}[1\la\cA^{\mathbf{NP}}(x) ] -\Pr_{x\la U_n}[1\la\cA^{\mathbf{NP}}(x)]}\geq n^{-\alpha}
\end{align*}
for some constant $0<\alpha<1$ for infinitely many $n\in\N$ assuming the non-existence of strong QAS.

For showing this, we use the following Claims~\ref{claim:far} and \ref{claim:est}.
We defer the proof of them.
\begin{claim}\label{claim:far}
For any algorithm $\cQ$ and any constant $0<\tau<1$ such that 
\begin{align*}
\Delta(\cQ(1^n),U_n)\geq 1-2^{-n^\tau}
\end{align*}
for all sufficiently large $n\in\N$,
we have
\begin{align*}
\Pr_{x\la\cQ(1^n)}[ \Pr[x\la\cQ(1^n)]\geq 2^{-n+\frac{n^{\tau}}{2}} ]\geq 1-2\cdot 2^{\frac{-n^{\tau}}{2}}
\end{align*}
for all sufficiently large $n\in\N$.
\end{claim}

\begin{claim}\label{claim:close}
For any constant $c>0$, and any algorithm $\cQ$ and $\cC$ such that
\begin{align*}
\Delta(\cQ(1^n),\cC(1^n))\leq 1-n^{-c}
\end{align*}
for all sufficiently large $n\in\N$,
we have
\begin{align*}
\Pr_{x\la\cQ(1^n)}\left[ \frac{1}{n^{2c}}\Pr[x\la\cQ(1^n) ]\leq \Pr[x\la\cC(1^n)] \right]\geq \frac{n^{-c}}{2}
\end{align*}
for all sufficiently large $n\in\N$.
\end{claim}

\begin{claim}\label{claim:est}
For any PPT algorithm $\cD$ and any $c>0$, there exists a PPT algorithm $\mathsf{Estimate}$ querying to an $\mathbf{NP}$ oracle such that
\begin{align*}
\Pr[ \frac{1}{2}\Pr[x\la\cD(1^n)] \leq \mathsf{Estimate}^{\mathbf{NP}}(1^n,x)\leq \Pr[x\la\cD(1^n)] ]\geq 1-n^{-c}
\end{align*}
for all $x\in\bit^*$ and all sufficiently large $n\in\N$.
\end{claim}

From the non-existence of strong QAS, there exists a PPT algorithm $\cC$ and a constant $c>0$ such that
\begin{align*}
\Delta(\cQ(1^n),\cC(1^n))\leq 1-n^{-c}
\end{align*}
for infinitely many $n\in\N$.
From Claim~\ref{claim:est}, there exists a PPT algorithm $\mathbf{Estimate}$ querying to an $\mathbf{NP}$ oracle such that
\begin{align*}
\Pr[ \frac{1}{2}\Pr[x\la\cC(1^n)] \leq \mathsf{Estimate}^{\mathbf{NP}}(1^n,x)\leq \Pr[x\la\cC(1^n)] ]\geq 1-n^{-2c}
\end{align*}
for all $x\in\bit^*$ and all sufficiently large $n\in\N$.

Now, we describe $\cA$.
\begin{description}
\item[The description of $\cA^{\mathbf{NP}}$:]$ $
\begin{enumerate}
\item Receive $x\in\bit^n$.
\item Run $p\la\mathsf{Estimate}^{\mathbf{NP}}(x)$.
\item If $p\geq n^{-2c}\cdot 2^{-n+\frac{n^\tau}{2}-1}$, output $1$.
Otherwise, sample $b\la\bit$ and output $b$.
\end{enumerate}
\end{description}

Let us define $\cS_{n}(A)$.
\begin{align*}
\cS_{n}(A)\seteq \left\{x\in\bit^n: \Pr[x\la\cC(1^n)]\geq A \right\}.
\end{align*}
We have
\begin{align*}
\Pr_{x\la U_n}[x\in\cS_{n}(A)]=\sum_{x\in\cS_{n}(A)} 2^{-n}\leq \abs{\cS_{n}(A)}2^{-n}\leq \frac{2^{-n}}{A}.
\end{align*}
Furthermore, from Claims~\ref{claim:far} and \ref{claim:close} and union bound, we have
\begin{align*}
\Pr_{x\la \cQ(1^n)}[x\in\cS_{n}(n^{-2c}\cdot 2^{-n+\frac{n^{\tau}}{2}})]\geq \frac{n^{-c}}{2}-2\cdot 2^{-\frac{n^\tau}{2}}\geq \frac{49}{100}n^{-c}
\end{align*}
for all sufficiently large $n\in\N$.
Hence, 
\begin{align*}
&\Pr_{x\la\cQ(1^n)}[1\la\cA^{\mathbf{NP}}(x)]-\Pr_{x\la U_n}[1\la\cA^{\mathbf{NP}}(x)]\\
&=\sum_{x\in\cS_{n}(n^{-2c}\cdot 2^{-n+\frac{n^{\tau}}{2}})}\Pr[1\la\cA^{\mathbf{NP}}(x)]\Pr[x\la\cQ(1^n)]+\sum_{x\notin\cS_{n}(n^{-2c}\cdot 2^{-n+\frac{n^{\tau}}{2}})}\Pr[1\la\cA^{\mathbf{NP}}(x)]\Pr[x\la\cQ(1^n)] \\
&-\left(\sum_{x\in\cS_{n}(n^{-2c}\cdot 2^{-n+\frac{n^{\tau}}{2}-2})}\Pr[1\la\cA^{\mathbf{NP}}(x)]\Pr[x\la U_n]+\sum_{x\notin\cS_{n}(n^{-2c}\cdot 2^{-n+\frac{n^{\tau}}{2}-2})}\Pr[1\la\cA^{\mathbf{NP}}(x)]\Pr[x\la U_n]\right)\\
&\geq \sum_{x\in\cS_{n}(n^{-2c}\cdot 2^{-n+\frac{n^{\tau}}{2}})}(1-n^{-2c})\Pr[x\la\cQ(1^n)]+\sum_{x\notin\cS_{n}(n^{-2c}\cdot 2^{-n+\frac{n^{\tau}}{2}})}\frac{1}{2}\Pr[x\la\cQ(1^n)]\\
&-\left(\sum_{x\in\cS_{n}(n^{-2c}\cdot 2^{-n+\frac{n^{\tau}}{2}-2})}\Pr[x\la U_n]+\sum_{x\notin\cS_{n}(n^{-2c}\cdot 2^{-n+\frac{n^{\tau}}{2}-2})}\left(\frac{1}{2}+\frac{n^{-2c}}{2}\right)\Pr[x\la U_n]\right)\\
&\geq \left(1-n^{-2c}\right) \cdot \frac{49}{100}n^{-c}+\frac{1}{2}\left(1-\frac{49}{100}n^{-c}\right)
-\left( 4n^{2c}\cdot 2^{\frac{-n^\tau}{2}} + \left(\frac{1}{2}+\frac{n^{-2c}}{2}\right)(1-4n^{2c}\cdot 2^{\frac{-n^\tau}{2}})\right)\\
&\geq \frac{1}{2}\left(1+\frac{49}{200}n^{-c}\right)-\left(\frac{1}{2}+\frac{n^{-2c}}{2}\right)\left(1+\negl(n)\right) \geq \frac{49}{400}n^{-2c}-\negl(n)
\end{align*}
for all sufficiently large $n\in\N$, which is a contradiction to the security of $\cQ$.
Here, in the first inequality, we have used that
\begin{align*}
\Pr[1\la\cA^{\mathbf{NP}}(x)]\geq 1-n^{-2c}
\end{align*}
for all $x\in\cS(n^{-2c}\cdot 2^{-n+n^{\frac{n^\tau}{2}}})$ and 
\begin{align*}
\Pr[1\la\cA^{\mathbf{NP}}(x)]\leq \frac{1}{2}+\frac{n^{-2c}}{2}
\end{align*}
for all $x\notin \cS(n^{-2c}\cdot 2^{-n+n^{\frac{n^\tau}{2}}-2})$.

\begin{proof}[Proof of Claim~\ref{claim:far}]
We define the following sets.
\begin{align*}
    &A\seteq \{x\in\bit^n:\Pr[x\la \cQ(1^n)]<2^{-n} \} \\
    &B\seteq \{x\in\bit^n: 2^{-n}\leq \Pr[x\la\cQ(1^n)]<2^{-n+\frac{n^{\tau}}{2}} \}\\
    &C\seteq \{x\in\bit^n:2^{-n+\frac{n^{\tau}}{2}}\leq \Pr[x\la\cQ(1^n)]\leq 1 \}.
\end{align*}
From the definition of $\cQ$ and total variation distance, we have
\begin{align*}
1-2^{-n^{\tau}}&\leq \sum_{x\in A}\Pr[x\la\bit^n]-\Pr[x\la\cQ(1^n)]\\
               &=\sum_{x\in A}2^{-n}-\Pr[x\la\cQ(1^n)]\\
               &\leq \abs{A}2^{-n}.
\end{align*}
This implies that 
\begin{align*}
\abs{B}+\abs{C}\leq 2^{n-n^{\tau}}.
\end{align*}
Furthermore, we have
\begin{align*}
1-2^{-n^{\tau}}
&\leq \sum_{x\in B}(\Pr[x\la\cQ(1^n)]-2^{-n})+\sum_{x\in C}(\Pr[x\la\cQ(1^n)]-2^{-n})\\
&\leq \sum_{x\in B}2^{-n+\frac{n^{\tau}}{2}}+\sum_{x\in C}\Pr[x\la\cQ(1^n)]\\
&\leq \abs{B}\cdot 2^{-n+\frac{n^\tau}{2}}+\sum_{x\in C}\Pr[x\la\cQ(1^n)].
\end{align*}
This implies that
\begin{align*}
    1-2\cdot 2^{-\frac{n^{\tau}}{2}}\leq \sum_{x\in C}\Pr[x\la\cQ(1^n)].
\end{align*}
\end{proof}

\begin{proof}[Proof of Claim~\ref{claim:close}]
Let us define
\begin{align*}
&A\seteq \left\{x\in\bit^n: \frac{1}{n^{2c}}\Pr[x\la\cQ(1^n)]\geq \Pr[x\la \cD(1^n)]\right\}.
\end{align*}
Then, we have
\begin{align*}
\sum_{x\in A}\Pr[x\la \cD(1^n)]&\leq \sum_{x\in A} \frac{1}{n^{2c}}\Pr[x\la\cQ(1^n)]\\
\sum_{x\in A}\Pr[x\la \cD(1^n)]-\Pr[x\la\cQ(1^n)]&\leq \sum_{x\in A} \left(\frac{1}{n^{2c}}-1\right)\Pr[x\la\cQ(1^n)]
\end{align*}
This implies that
\begin{align*}
\left(1-\frac{1}{n^{2c}}\right)\sum_{x\in A}\Pr[x\la\cQ(1^n)]\leq \sum_{x\in A}\Pr[x\la\cQ(1^n)]-\Pr[x\la\cD(1^n)]\leq \Delta(\cD(1^n),\cQ(1^n)).
\end{align*}
Therefore, we have
\begin{align*}
\sum_{x\in A}\Pr[x\la\cQ(1^n)]\leq \frac{1-n^{-c}}{1-n^{-2c}}\leq 1-\frac{n^{-c}}{2}.
\end{align*}
This implies that
\begin{equation*}
\Pr_{x\la\cQ(1^n)}\left[\frac{1}{n^{2c}}\Pr[x\la\cQ(1^n)]\leq
\Pr[x\la\cC(1^n)] \right]\geq \frac{n^{-c}}{2}.
\qedhere
\end{equation*}
\end{proof}

\begin{proof}[Proof of Claim~\ref{claim:est}]
\cref{claim:est} directly follows from Claim~\ref{thm:probest}.
\end{proof}
\end{proof}

\begin{proof}[Proof of \cref{lem:amp}]
Suppose that QPRG secure against QPT algorithms querying to an $\mathbf{NP}$ oracle exists.
Then, there exists $\Gen$ such that
\begin{align*}
    \mathsf{SD}((x)_{x\la\Gen(1^n)},(x)_{x\la \bit^n})\geq \frac{1}{n}
\end{align*}
for all sufficiently large $n\in\N$.

For an arbitrary real $0<\tau <1$, we consider the following $\Gen^*$.
\begin{description}
    \item[The description of $\Gen^*(1^{n})$:]$ $
    \begin{itemize}
        \item Let $A=n^{\frac{1-\tau}{2}}$ and $B=n^{\frac{1+\tau}{2}}$.
        \item Run $ x_i\la \Gen(1^A)$ for all $i\in[B]$.
        \item Output $\{x_i\}_{i\in[B]}$.
    \end{itemize} 
\end{description}
We have
\begin{align*} 
    \mathsf{SD}\left( (\{x_i\}_{i\in[B]})_{\{x_i\}_{i\in[B]}\la \Gen^*(1^n)},X_{X\la\bit^{n}}  \right)&\geq 1- \exp\left(- B\cdot\mathsf{SD}\left(x_{x\la\Gen(1^A)},x_{x\la\bit^A}\right) \right)  \\
    &\geq 1-\exp(-\frac{B}{A})=1-\exp(-n^{\tau})>1-2^{-n^\tau}.
\end{align*}
Furthermore, computationally indistinguishability of $\Gen^*$ follows from a standard hybrid argument.
\end{proof}

\section{Efficient verification of quantum distributions with a $\PP$ oracle using Kolmogorov complexity}

%\subsection{Quantum Distribution Verification}
%\subsection{Quantum Distribution Verification}

%\subsubsection{Efficient Verification of Quantum Distributions with $\mathsf{PP}$ Oracle.}

In this section, we provide quantum distribution verification by using quantum Kolmogorov complexity for clarity. Remark that \cref{thm:q_easiness} is superseded by the \cref{thm:extrapolate}.

\label{sec:BQP=PP}
%In the following, we show the following \cref{thm:q_easiness}.
\begin{theorem}\label{thm:q_easiness}
Every QPT algorithm $\cD$ is adaptively-verifiable with a classical deterministic polynomial-time algorithm querying to $\mathsf{PP}$ oracle.
\end{theorem}

Although the proof is similar to that of \cref{thm:easiness}, we describe the proof for clarity. 
For showing \cref{thm:q_easiness}, we use the following \cref{thm:q_estimate,thm:qvk_estimate}.
%which guarantees that for any QPT algorithm $\cD$, there exists a ${\rm P}^{\mathbf{PP}}$ algorithm $\mathsf{Approx}$ that can compute $\Pr[x\la\cD(1^n)]$ for all $x\in\bit^*$.
\begin{theorem}[Worst-case probability estimation with $\mathsf{PP}$ oracle \cite{FR99}]\label{thm:q_estimate}
For any constant $c\in\N$, and for any QPT algorithm $\cD$, which takes $1^n$ as input and outputs $x\in\bit^{m(n)}$, where $m$ is a polynomial,
there exists a deterministic polynomial-time algorithm $\mathsf{Approx}$ querying to $\mathsf{PP}$ oracle such that
\begin{align*}
\Pr\left[\mathsf{Approx}^{\mathsf{PP}}(x,1^n)=\Pr[x\la\cD(1^n)]\right]= 1
\end{align*}
for all $x\in\bit^{m(n)}$ and for all sufficiently large $n\in\N$.
%\mor{Isn't the success probability is 1 for any $n$?}\taiga{I am not sure because we might have error to implement $\mathbf{PP}$ oracle assuming $\mathbf{BQP}=\mathbf{PP}$.}
\end{theorem}

\cref{thm:q_estimate} directly implies the following \cref{thm:qvk_estimate}.
\begin{theorem}\label{thm:qvk_estimate}
For any constant $c\in\N$ and polynomial $t$, there exists a deterministic polynomial-time algorithm $\cM$ querying to $\mathsf{PP}$ oracle such that
\begin{align*}
\Pr[\cM^{\mathsf{PP}}(x,1^n)=quK^{t(\abs{x})}(x|1^n)]=1
\end{align*}
for all $x\in\bit^*$ and all sufficiently large $n\in\N$.
\end{theorem}

\begin{proof}[Proof of \cref{thm:q_easiness}]
In the following, let us introduce notations to describe $\Vrfy$.

\paragraph{Notations.}
We set $s(n,t(n),\epsilon(n)) = n^{4c}\left(\log_2\frac{1}{1-\epsilon(n)} + 2(\log(n))^2 \right) $.
For simplicity, we often denote $s=n^{4c}\left(\log_2\frac{1}{1-\epsilon(n)} + 2(\log(n))^2 \right)$ below.

We also set $\alpha=(\log(n))^2$.

$\mathsf{Approx}$ is a deterministic polynomial-time algorithm querying to $\mathsf{PP}$ oracle given in \Cref{thm:q_estimate} such that 
\begin{align*}
\Pr\left[\mathsf{Approx}^{\mathsf{PP}}(x,1^{n})= \Pr[x\la\cD(1^n)^{\otimes s}]\right] = 1
\end{align*}
for all $x\in\bit^{s\cdot m(n)}$ and all sufficiently large $n\in\N$.

$\cM$ is a deterministic polynomial-time algorithm querying to $\mathsf{PP}$ oracle given in \cref{thm:qvk_estimate} such that
\begin{align*}
\Pr[\cM^{\mathsf{PP}}(x,1^{t})= quK^{t(1^n)}(x|1^n)] =1
\end{align*}
for all $x\in\bit^{s\cdot m(n)}$ and for all sufficiently large $n\in\N$.

\paragraph{Construction:}
We give a construction of $\Vrfy$.

\begin{description}
    \item[$\Vrfy$:]$ $
    \begin{enumerate}
        \item Take $1^n$ and  $y_1,\dots,y_{s(n)}$ as input.
        \item Compute $k\la\cM^{\mathsf{PP}}(y_1,\dots,y_{s(n)},1^{t})$.
        \item Compute $p\la\mathsf{Approx}^{\mathsf{PP}}(y_1,\dots,y_{s(n)})$.
        \item Output $\top$ if $-\log_2 p \leq k + \alpha$.
        Otherwise, output $\bot$.
    \end{enumerate}
\end{description}
In the following, we show that $\Vrfy$ satisfies the correctness and security.

\paragraph{Correctness:}
From union bound, and the definition of $\cM$ and $\mathsf{Approx}$, we have
\begin{align*}
&\Pr_{y_1,\dots,y_s\la\cD(1^n)^{\otimes s}}[\bot\la\Ver(1^n,y_1,\dots,y_s)]\\
&= \Pr_{y_1,\dots,y_s\la\cD(1^n)^{\otimes s}}[-\log_2 \left(\mathsf{Approx}^{\mathsf{PP}}(y_1,\dots,y_s)\right)\geq \cM^{\mathsf{PP}}(y_1,\dots,y_s)+\alpha]\\
&= \Pr_{y_1,\dots,y_s\la\cD(1^n)^{\otimes s}}\left[-\log_2 \left(\Pr[\cD(1^n)^{\otimes s}\ra y_1,\dots,y_s]\right)  \geq quK^{t(n)}(y_1,\dots,y_s|1^n) + \alpha\right]
\end{align*}
for all sufficiently large $n\in\N$.
From \cref{thm:q_kraft}, we have
\begin{align*}
\Pr_{y_1,\dots,y_s\la\cD(1^n)^{\otimes s}}\left[-\log_2 (\Pr[y_1,\dots,y_s\la\cD(1^n)^{\otimes s}]) \geq quK^{t(n)}(y_1,\dots,y_s|1^n) + \alpha\right]\leq \frac{s\cdot m(n)}{n^{\log(n)}}\leq \frac{1}{n^{2c}}
\end{align*}
for all sufficiently large $n\in\N$.
This concludes the correctness.

\paragraph{Adaptive-Soundness:}

From union bound and the definition of $\cM$ and $\mathsf{Approx}$, we have
\begin{align*}
&\Pr_{y_1,\dots,y_s\la\cD(1^n)^{\otimes s}}[\Ver(1^n,y_1,\dots,y_s)\to \top]\\
&\Pr_{y_1,\dots,y_s\la\cD(1^n)^{\otimes s}}[-\log_2 \left(\mathsf{Approx}^{\mathsf{PP}}(y_1,\dots,y_s)\right)\leq \cM^{\mathsf{PP}}(y_1,\dots,y_s)+\alpha]\\
&= \Pr_{y_1,\dots,y_s\la\cD(1^n)^{\otimes s}}\left[-\log_2 \left(\Pr[y_1,\dots,y_s\la\cD(1^n)^{\otimes s}]\right)  \leq quK^{t(n)}(y_1,\dots,y_s|1^n) + \alpha\right]
\end{align*}
for all sufficiently large $n\in\N$.

Furthermore, from \cref{lem:q_modaaronson},
for any $t(n)$-time quantum adversary $\cA$, which takes $1^n$ and outputs strings of length $m(n)\cdot s$, and satisfies $\Delta(\mathsf{Marginal}_{\cA}(1^n),\cD_n)\geq \epsilon(n)$, we have 
\begin{align*}
&\Pr_{y_1,\dots,y_s\la\cD(1^n)^{\otimes s}}\left[-\log_2 \left(\Pr[y_1,\dots,y_s\la\cD(1^n)^{\otimes s}]\right) \leq quK^{t(n)}(y_1,\dots,y_s|1^n) + \alpha\right]\\
&\leq 1-\epsilon(n)+ \sqrt{\frac{\log_2\frac{1}{1-\epsilon(n)}+(\log(n))^2+C}{n^{4c}\left(\log_2\frac{1}{1-\epsilon(n)}+2(\log(n))^2 \right) }}\\
&\leq 1-\epsilon(n) + n^{-2c}
\end{align*}
for all sufficiently large $n\in\N$.
\end{proof}

\ifnum\cameraready=1
\else
\ifnum\submission=1
\newpage
\setcounter{tocdepth}{1}
%\tableofcontents
\else
\fi
\fi

\end{document}